\documentclass{article}

\usepackage{fullpage}
\usepackage{url,color,cite}
\usepackage{caption}
\usepackage{amsmath,amsthm,amsfonts,amssymb}
\usepackage{mathtools}
\DeclarePairedDelimiter{\ceil}{\lceil}{\rceil}
\DeclarePairedDelimiter{\floor}{\lfloor}{\rfloor}
\usepackage{float}
\usepackage{epstopdf}
\usepackage{nicefrac}
\usepackage{algorithm}
\usepackage[noend]{algpseudocode}
\usepackage{pifont}
\usepackage{subcaption}
\usepackage{sidecap}
\usepackage{multirow}
\usepackage{multicol}
\usepackage{bbm}

\newtheorem{lemma}{Lemma}
\newtheorem*{lemma*}{Lemma}
\newtheorem{theorem}{Theorem}
\newtheorem*{theorem*}{Theorem}
\theoremstyle{definition}
\newtheorem{defn}{Definition}
\theoremstyle{definition}

\newtheorem{remark}{Remark}
\theoremstyle{definition}
\newtheorem{corollary}{Corollary}
\newtheorem*{corollary*}{Corollary}
\theoremstyle{definition}
\newtheorem{claim}{Claim}
\newtheorem*{claim*}{Claim}

\newcommand{\same}{\text{same}}

\newcommand{\bbE}{\mathbb{E}}
\newcommand{\bbN}{\mathbb{N}}
\newcommand{\bbR}{\mathbb{R}}

\newcommand{\calA}{\mathcal{A}}

\newcommand{\calC}{\mathcal{C}}
\newcommand{\calD}{\mathcal{D}}
\newcommand{\calE}{\mathcal{E}}
\newcommand{\calI}{\mathcal{I}}
\newcommand{\calH}{\mathcal{H}}
\newcommand{\calM}{\mathcal{M}}

\newcommand{\calP}{\mathcal{P}}

\newcommand{\calR}{\mathcal{R}}
\newcommand{\calS}{\mathcal{S}}
\newcommand{\calT}{\mathcal{T}}
\newcommand{\calU}{\mathcal{U}}

\newcommand{\calX}{\mathcal{X}}
\newcommand{\calY}{\mathcal{Y}}

\renewcommand{\(}{\left(}
\renewcommand{\)}{\right)}

\newcommand{\eps}{\epsilon}

\newcommand{\bh}{\boldsymbol{h}}
\newcommand{\bm}{\boldsymbol{m}}
\newcommand{\btm}{\widetilde{\bm}}

\newcommand{\by}{\boldsymbol{y}}
\newcommand{\bp}{\boldsymbol{p}}
\newcommand{\bhp}{\widehat{\bp}}
\newcommand{\btp}{\widetilde{\bp}}

\newcommand{\tbh}{\widetilde{\boldsymbol{h}}}

\newcommand{\bP}{\mathbf{P}}

\usepackage{microtype}
\usepackage{graphicx}
\usepackage{booktabs} 
\usepackage{hyperref}
\usepackage{cleveref}

\usepackage{lipsum}

\title{On the Renyi Differential Privacy of the Shuffle Model}

\author{Antonious M. Girgis, Deepesh Data, Suhas Diggavi, Ananda Theertha Suresh, and Peter Kairouz}

\date{}
\begin{document}
\maketitle
{\allowdisplaybreaks
\begin{abstract}

The central question studied in this paper is Renyi Differential
Privacy (RDP) guarantees for general discrete local mechanisms in the
{\em shuffle} privacy model. In the shuffle model, each of the $n$
clients randomizes its response using a local differentially private
(LDP) mechanism and the untrusted server only receives a random
permutation (shuffle) of the client responses without association to
each client. The principal result in this paper is 
the \emph{first} non-trivial RDP guarantee for general discrete local
randomization mechanisms in the shuffled privacy model, and we develop 
new analysis techniques for deriving our results which could be of independent interest. In
applications, such an RDP guarantee is most useful when we use it for
composing several private interactions. We numerically demonstrate
that, for important regimes, with composition our bound yields an
improvement in privacy guarantee by a factor of $8\times$ over the
state-of-the-art approximate Differential Privacy (DP) guarantee (with
standard composition) for shuffled models. Moreover, combining with
Poisson subsampling, our result leads to at least $10\times$ improvement over
subsampled approximate DP with standard composition. 


\end{abstract}

\section{Introduction}~ \label{sec:introduction}







Differential privacy (DP)~\cite{Calibrating_DP06} gives a principled
and rigorous framework for data privacy by giving guarantees on the
information leakage for individual data points from the output of an
algorithm. Algorithmically, a standard method is to randomize the
output of an algorithm to enable such privacy. Originally DP was
studied in the \emph{centralized} context, where the privacy from
queries to a \textit{trusted server} holding the data was the
objective \cite{Calibrating_DP06}.  However, in distributed applications,
such as federated learning~\cite{Kairouz-FL-survey19}, two significant aspects need to be
accommodated: {\sf (i)} data is held locally at clients and needs to be used for
computation with an \emph{untrusted} server; and {\sf (ii)} to build
good learning models, one might need repeated interactions
(\emph{e.g.,} through distributed gradient descent).

To accommodate privacy of \emph{locally} held data, a more appropriate
notion is that of local differential privacy (LDP)
~\cite{kasiviswanathan2011can,duchi2013local}. In the LDP
framework, each (distributed) client holding local data, individually
randomizes its interactions with the (untrusted) server.%
\footnote{The mechanisms used have a long history including Randomized
  Response \cite{warner1965-RR}, but were recently studied through the lens
  of local differential privacy (LDP).} Recently, such LDP mechanisms
have been deployed by companies such as
Google~\cite{erlingsson2014rappor}, Apple~\cite{greenberg2016apple},
and Microsoft~\cite{microsoft}.  However, LDP mechanisms suffer from
poor performance in comparison with the centralized DP mechanisms,
making their applicability
limited~\cite{duchi2013local,kasiviswanathan2011can,kairouz2016discrete}.
To address this, a new privacy framework using anonymization has been
proposed in the so-called \emph{shuffled model}
\cite{erlingsson2019amplification,cheu2019distributed,balle2019privacy},
where each client sends her (randomized) interaction
message to a secure shuffler that randomly permutes all the received
messages before forwarding them to the server.\footnote{Such a
  shuffling can be enabled through anonymization techniques
  \cite{Prochlo17,erlingsson2019amplification,ESA}.} This model enables significantly better
privacy-utility performance by amplifying LDP through this mechanism.

For the second aspect, where there are repeated interactions
(\emph{e.g.,} through distributed gradient descent), one needs privacy
composition~\cite{bassily2014private}. In other words, for a given privacy budget, how many such
interactions can we accommodate. Clearly, from an optimization
viewpoint, we might need to run these interactions longer for better
models, but these also result in privacy leakage. Though the privacy
leakage can be quantified using advanced composition theorems for DP
(\emph{e.g.,} \cite{dwork2010boosting,kairouz2015composition}), these
might be loose. To address this, Abadi \emph{et al.}
\cite{abadi2016deep} developed a ``moments accountant'' framework,
which enabled a much tighter composition. This is enabled by providing
the composition privacy guarantee in terms of Renyi Differential
Privacy~\cite{mironov2017renyi}, and then mapping it back to the DP
guarantee \cite{mironov2019r}. It is known \cite{abadi2016deep} that
the moments accountant provides a significant saving in the total
privacy budget in comparison with using the strong composition
theorems~\cite{dwork2010boosting,kairouz2015composition}.  Therefore,
developing the RDP privacy guarantee can enable stronger composition
privacy results.

This leads us to the central question addressed in this paper:
\begin{quote}
\emph{Can we develop strong RDP privacy guarantees for general local mechanisms in the shuffled privacy model?}
\end{quote}

The principal result in this paper is the
\emph{first} non-trivial RDP guarantee for general discrete local
randomization mechanisms in the shuffled privacy model. In particular,
given an \emph{arbitrary discrete} local mechanism with
$\epsilon_0$-LDP guarantee, we provide an RDP guarantee for the
shuffled model, as a function of $\epsilon_0$ and the number of users
$n$. 
For example, numerically, we save a factor of
  $8\times$ compared to the state-of-the-art approximate DP guarantee
  for shuffled models in \cite{feldman2020hiding}\footnote{We used the
    open source implementation for the privacy analysis in
    \cite{feldman2020hiding} available from
    \url{https://github.com/apple/ml-shuffling-amplification}.}
  combined with strong composition, with the number of iterations
  $T=10^{5}$, LDP parameter $\epsilon_0=0.5$, and number of clients
  $n=10^{6}$; see Figure~\ref{fig:Dp_com_1} in
  Section~\ref{sec:numerics} for other such examples
  regimes. Furthermore, characterizing the RDP of the shuffled model
  enables us to compute the RDP of shuffling with Poisson sub-sampling
  by using the results in~\cite{zhu2019poission}. We show numerically
  that this approach can lead to at least $10\times$ improvement in privacy
  guarantee. This is for $T=10^{4}$, $\epsilon_0=3$, and $n=10^{6}$. The
  comparison is with applying the strong composition
  theorem~\cite{kairouz2015composition} after getting the
  state-of-the-art approximate DP of the shuffled model given
  in~\cite{feldman2020hiding} with Poisson
  sub-sampling~\cite{li2012sampling} (see
  Figure~\ref{fig:Dp_com_sampling_1} in Section~\ref{sec:numerics} for
  more such regimes). This in turn implies that we can accommodate
  at least $10\times$ more interactions for the same privacy budget in these cases. 
Moreover, our upper bounds also give several orders of magnitude improvement over
the simple RDP bound stated in
\cite[Remark~$1$]{erlingsson2019amplification} (also stated in
\eqref{eqn:rdp_erlin-1}) in several regimes (see, for example,
Figure~\ref{fig:5} in Section~\ref{sec:numerics}). We also develop a
lower bound for the RDP for the shuffled model and demonstrate
numerically that the gap is small for many parameter regimes of
interest.
  
%

In order to obtain our upper bound result, we develop new analysis techniques
which could be of independent interest. In particular, we develop a
novel RDP analysis for a special structure neighboring datasets (see
Theorem~\ref{thm:RDP_same}), in which one of the datasets has all the data points to be the same (see the definition in \eqref{datasets-same}). 
A key technical result is then to relate
the RDP of general neighboring datasets to those with special
structure (see Theorem~\ref{Thm:reduce_special_case}).
\begin{itemize}
\item For the RDP analysis of neighboring datasets with the above-mentioned special
  structure, we first observe that the output distribution of the
  shuffling mechanism is the multinomial distribution. Using this
  observation, then we show that the ratio of the distributions of the
  mechanism on special structure neighboring datasets is a
  sub-Gaussian random variable (r.v.), and we can write the Renyi
  divergence of the shuffled mechanism in terms of the moments of this
  r.v. Bounding the moments of this r.v.\ then gives an upper bound on
  the RDP for the special neighboring datasets. See the proof-sketch of
  Theorem~\ref{thm:RDP_same} in
  Section~\ref{subsubsec:proof-sketch_rdp-same} and its complete proof
  in Section~\ref{sec:special_form}.
\item We next connect the above analysis to the RDP computation for
  general neighboring datasets $\calD=(d_1,\hdots,d_n)$ and
  $\calD'=(d_1,\hdots,d_{n-1},d'_n)$. To do so, a crucial observation
  is to write the output distribution $\bp_i$ of the local randomizer
  $\calR$ on the $i$'th client's data point $d_i$ (for any
  $i\in[n-1]$) as a mixture distribution
  $\bp_i=e^{-\eps_0}\bp'_n+(1-e^{-\eps_0})\tilde{\bp}_i$ for some
  $\tilde{\bp}_i$. So, the number of clients that sample according to
  $\bp'_n$ is concentrated around $e^{-\eps_0}n$. Therefore, if we
  restrict the dataset to these clients only, the resulting datasets
  will have the special structure, and the size of that dataset will
  be concentrated around $e^{-\eps_0}n$. Finally, in order to be able
  to reduce the problem to the special case, we remove the effect of
  the clients that do not sample according to $\bp_n'$ without
  affecting the Renyi divergence. See the proof-sketch of
  Theorem~\ref{Thm:reduce_special_case} in
  Section~\ref{subsubsec:proof-sketch_reduce-special-case} and its
  complete proof in Section~\ref{sec:proof_reduce_special_case}.
\end{itemize}

\subsection*{Related Work}

We give the most relevant work related to the paper and put our contributions in the context of these works.

\paragraph{Shuffled privacy model:} As mentioned, the shuffled
model of privacy has been of significant recent interest
\cite{erlingsson2019amplification,ghazi2019power,balle2019improved,ghazi2019scalable,balle2019differentially,cheu2019distributed,balle2019privacy,balle2020private}. However,
all the existing works in 
literature~\cite{erlingsson2019amplification,balle2019privacy,feldman2020hiding}
only characterize the approximate DP of the shuffled model -- among these, 
\cite{feldman2020hiding} is the state-of-the-art, but as we show in our experiments, 
it yields weaker results when combined with composition. 
To the best of our knowledge, there is no bound on RDP of the shuffled model
in the literature except for the one mentioned briefly in a remark
in~\cite[Remark~$1$]{erlingsson2019amplification} (which is obtained by the
standard conversion results from DP to RDP) and we state it in 
\eqref{eqn:rdp_erlin-1} for comparison. However, this
bound is loose (\emph{e.g.,} see Figure \ref{fig:5}) and not useful
for conversion to approximate DP (\emph{e.g.,} see Figures
\ref{fig:DP_1}, \ref{fig:DP_3}), as well as for composition
(\emph{e.g.,} see Figure \ref{fig:Dp_com_5}). Thus, our work makes
progress on this important open question of analyzing the RDP of the
shuffled model. Both~\cite{ESA} and~\cite{girgis2021shuffled-aistats}
used advanced composition to analyze privacy of shuffled models in
federated learning; our results could be adapted to enhance their
privacy guarantees.

\paragraph{Renyi differential privacy:} The work of Abadi
\emph{et al.} \cite{abadi2016deep} provided a methodology to get
stronger composition results. Inherently, this used Renyi divergence,
and was later formalized in \cite{mironov2017renyi} which defined
Renyi differential privacy (RDP). RDP presents a unified definition
for several kinds of privacy notions including pure differential
privacy ($\epsilon$-DP), approximate differential privacy
($\left(\epsilon,\delta\right)$-DP), and concentrated differential
privacy (CDP)~\cite{bun2016concentrated,CDP_DworkR16}. As mentioned earlier, RDP enables a
stronger result for composition, through the ``moment accounting''
idea.  Similarly, several
works~\cite{mironov2019r,wang2019subsampled,zhu2019poission} have
shown that analyzing the RDP of subsampled mechanisms provides a
tighter bound on the total privacy loss than the bound that can be
obtained using the standard strong composition theorems. However, to
the best of our knowledge, RDP analysis of the shuffled model and its
use for composition in the shuffled model has not been studied. In
this paper, we analyze the RDP of the shuffled model, where we can
bound the approximate DP of a sequence of shuffled models using the
transformation from RDP to approximate
DP~\cite{abadi2016deep,wang2019subsampled,canonne2020discrete,asoodeh2021three}. We
show that our RDP analysis provides a better bound on the total
privacy loss of composition than that can be obtained using the
standard strong composition theorems (see Section \ref{sec:numerics}).

\paragraph{Discrete mechanisms:} Many of the works in DP use 
specific randomization mechanisms, adding noise using the Laplace or
Gaussian distributions. However, in many situations the data is
inherently discrete (\emph{e.g.,} see \cite{canonne2020discrete} and
references therein) or compression causes it to be so (\emph{e.g.,}
see \cite{girgis2021shuffled-aistats,Kairouz_discrete_FL-SA21} and references therein). It is therefore
of interest to directly analyze privacy of discrete randomization
mechanisms. Such discrete mechanisms have been studied extensively in
shuffled models~\cite{ghazi2019power,balle2019differentially}, but for approximate DP. 
To the best of our knowledge, RDP
for general discrete mechanisms in the shuffled privacy framework is
new to our work.
%
%

\subsection*{Paper Organization} 

The paper is organized as follows. In
Section~\ref{sec:prelims_prob-form}, we give some preliminary
definitions and results from the literature and also formulate our
problem. In Section~\ref{sec:main_results}, we present our main
results (two upper bounds and one lower bound on RDP), along with a
proof sketch of the first upper bound. We also describe the two main
ingredients in its proof -- first is the reduction of computing RDP
for the arbitrary pairs of neighboring datasets to computing RDP for
the special pairs of neighboring datasets, and the second is computing
RDP for the special pairs of neighboring datasets.  In
Section~\ref{sec:numerics}, we present several numerical results to
demonstrate the advantages of these bounds of the
state-of-the-art. The rest of the sections are devoted to the full
proofs of our main results:
Section~\ref{sec:proof_reduce_special_case} shows the reduction of our
general problem to the special case; Section~\ref{sec:special_form}
proves the RDP for the special case; Section~\ref{sec:general_case}
proves both our upper bounds; and Section~\ref{sec:lower-bound} proves
our lower bound. In Section~\ref{sec:conclusion}, we conclude with a
short discussion. Omitted details from the proofs are provided in the
appendices.


\section{Preliminaries and Problem Formulation}\label{sec:prelims_prob-form}
We give different privacy definitions that we use in Section~\ref{subsec:priv-defn}, some existing results on RDP to DP conversion and RDP composition in Section~\ref{subsec:RDP-DP}, and give our problem formulation in Section~\ref{subsec:prob-form}.
\subsection{Privacy Definitions}\label{subsec:priv-defn}
In this subsection, we define different privacy notions that we will use in this paper: local differential privacy (LDP), central different privacy (DP), and Renyi differential privacy (RDP).
\begin{defn}[Local Differential Privacy - LDP~\cite{kasiviswanathan2011can}]~\label{defn:LDPdef}
For $\epsilon_0\geq0$, a randomized mechanism $\calR:\calX\to\calY$ is said to be $\eps_0$-local differentially private (in short, $\eps_{0}$-LDP), if for every pair of inputs $d,d'\in\calX$, we have 
\begin{equation}~\label{ldp-def}
\Pr[\calR(d)\in \calS] \leq e^{\eps_0}\Pr[\calR(d')\in \calS], \qquad \forall \calS\subset\calY.
\end{equation}
\end{defn}

Let $\calD=\lbrace d_1,\ldots,d_n\rbrace$ denote a dataset comprising $n$ points from $\calX$. We say that two datasets $\calD=\lbrace d_1,\ldots,d_n\rbrace$ and $\calD^{\prime}=\lbrace d_1^{\prime},\ldots,d_n^{\prime}\rbrace$ are neighboring (and denoted by $\calD\sim\calD'$) if they differ in one data point, i.e., there exists an $i\in[n]$ such that $d_i\neq d'_i$ and for every $j\in[n],j\neq i$, we have $d_j=d'_j$.
\begin{defn}[Central Differential Privacy - DP \cite{Calibrating_DP06,dwork2014algorithmic}]\label{defn:central-DP}
For $\epsilon,\delta\geq0$, a randomized mechanism $\calM:\calX^n\to\calY$ is said to be $(\epsilon,\delta)$-differentially private (in short, $(\epsilon,\delta)$-DP), if for all neighboring datasets $\calD,\calD^{\prime}\in\calX^{n}$ and every subset $\calS\subseteq \calY$, we have
\begin{equation}~\label{dp_def}
\Pr\left[\calM(\calD)\in\calS\right]\leq e^{\eps_0}\Pr\left[\calM(\calD^{\prime})\in\calS\right]+\delta.
\end{equation}
\end{defn}

\begin{defn}[$(\lambda,\epsilon)$-RDP (Renyi Differential Privacy)~\cite{mironov2017renyi}]\label{defn:RDP}
A randomized mechanism $\calM:\calX^n\to\calY$ is said to have $\epsilon$-Renyi differential privacy of order $\lambda\in(1,\infty)$ (in short, $(\lambda,\epsilon(\lambda))$-RDP), if for any neighboring datasets $\calD$, $\calD'\in\calX^n$, the Renyi divergence between $\calM(\calD)$ and $\calM(\calD')$ is upper-bounded by $\eps(\lambda)$, i.e.,
\begin{equation}
D_{\lambda}(\calM(\calD)||\calM(\calD'))=\frac{1}{\lambda-1}\log\left(\mathbb{E}_{\theta\sim\calM(\calD')}\left[\left(\frac{\calM(\calD)(\theta)}{\calM(\calD')(\theta)}\right)^{\lambda}\right]\right)\leq \epsilon(\lambda),
\end{equation}
where $\calM(\calD)(\theta)$ denotes the probability that $\calM$ on input $\calD$ generates the output $\theta$. For convenience, instead of $\eps(\lambda)$ being an upper bound, we define it as $\eps(\lambda)=\sup_{\calD\sim\calD'}D_{\lambda}(\calM(\calD)||\calM(\calD'))$.
\end{defn}
Our objective in this paper is to characterize the Renyi differential privacy of a shuffling mechanism $\calM$ (see Section~\ref{subsec:prob-form} for details) for different values of order $\lambda$. 

\subsection{RDP to DP Conversion and RDP Composition}\label{subsec:RDP-DP}
In this subsection, we state some preliminary results from literature
that we will use. Though our main objective in this paper is to derive
RDP guarantees of a shuffling mechanism, we also give the
central privacy guarantees of that mechanism. For that purpose, we use
the following result for converting the RDP guarantees of a mechanism
to its DP guarantees. To the best of our knowledge, this result gives
the best conversion.\footnote{An optimal conversion from RDP to
  approximate DP was studied in \cite{asoodeh2021three}; however, we observed numerically, that it does not give better performance as compared to the conversion presented above.} 
\begin{lemma}[From RDP to DP~\cite{canonne2020discrete,Borja_HypTest-RDP20}]\label{lem:RDP_DP} 
Suppose for any $\lambda>1$, a mechanism $\calM$ is $\left(\lambda,\epsilon\left(\lambda\right)\right)$-RDP. Then, the mechanism $\calM$ is $\left(\epsilon,\delta\right)$-DP, where $\epsilon,\delta$ are define below: 
\begin{equation*}
\begin{aligned}
\text{For a given }\delta\in(0,1) &: & \epsilon &= \min_{\lambda} \epsilon\left(\lambda\right)+\frac{\log\left(1/\delta\right)+\left(\lambda-1\right)\log\left(1-1/\lambda\right)-\log\left(\lambda\right)}{\lambda-1}\\
\text{For a given }\epsilon>0 &: & \delta &= \min_{\lambda}\frac{\exp\left(\left(\lambda-1\right)\left(\epsilon\left(\lambda\right)-\epsilon\right)\right)}{\lambda-1}\left(1-\frac{1}{\lambda}\right)^{\lambda}. 
\end{aligned}
\end{equation*}
\end{lemma}
As mentioned in Section~\ref{sec:introduction}, the main strength of RDP in comparison to other privacy notions comes from composition. The following result states that if we adaptively compose two RDP mechanisms with the same order, their privacy parameters add up in the resulting mechanism.
\begin{lemma}[Adaptive composition of RDP~{\cite[Proposition~1]{mironov2017renyi}}]\label{lemm:compostion_rdp} 
For any $\lambda>1$, let $\calM_1:\calX\to \calY_1$ be a $(\lambda,\epsilon_1(\lambda))$-RDP mechanism and $\calM_2:\calY_1\times \calX\to \calY$ be a $(\lambda,\epsilon_2(\lambda))$-RDP mechanism. Then, the mechanism defined by $(\calM_1,\calM_2)$ satisfies $(\lambda,\epsilon_1(\lambda)+\epsilon_2(\lambda))$-RDP. 
\end{lemma}
By combining the above two lemmas, we can give the DP privacy guarantee of a sequence of $T$ adaptive RDP mechanisms. Let $\calM_1\left(\calI_1,\calD\right),\ldots,\calM_T\left(\calI_T,\calD\right)$ be a sequence of $T$ adaptive mechanisms, where $\calI_t$ denotes the auxiliary input to the $t$'th mechanism, which may depend on the previous mechanisms' outputs and the auxiliary inputs $\{(\calI_i,\calM_i(\calI_i,\calD)):i<t\}$, and the $t$'th mechanism $\calM_t$ is $\left(\lambda,\epsilon_t\left(\lambda\right)\right)$-RDP. Thus, from Lemma~\ref{lemm:compostion_rdp}, we get that the mechanism defined by $\left(\calM_1,\ldots,M_T\right)$ is $\left(\lambda,\sum_{t=1}^{T}\epsilon_t\left(\lambda\right)\right)$-RDP, and then from Lemma~\ref{lem:RDP_DP} that the mechanism $\left(\calM_1,\ldots,M_T\right)$ satisfies $\left(\epsilon,\delta\right)$-DP, where $\epsilon,\delta$ are given below: 
\begin{equation*}
\begin{aligned}
\text{For a given }\delta\in(0,1) &: & \epsilon &= \min_{\lambda} \sum_{t=1}^{T}\epsilon_t\left(\lambda\right)+\frac{\log\left(1/\delta\right)+\left(\lambda-1\right)\log\left(1-1/\lambda\right)-\log\left(\lambda\right)}{\lambda-1}\\
\text{For a given }\epsilon>0 &: & \delta &= \min_{\lambda}\frac{\exp\left(\left(\lambda-1\right)\left(\sum_{t=1}^{T}\epsilon_t\left(\lambda\right)-\epsilon\right)\right)}{\lambda-1}\left(1-\frac{1}{\lambda}\right)^{\lambda}. 
\end{aligned}
\end{equation*} 

\subsection{Problem Formulation}\label{subsec:prob-form}
\begin{figure}[t]
\centering \includegraphics[scale=0.5]{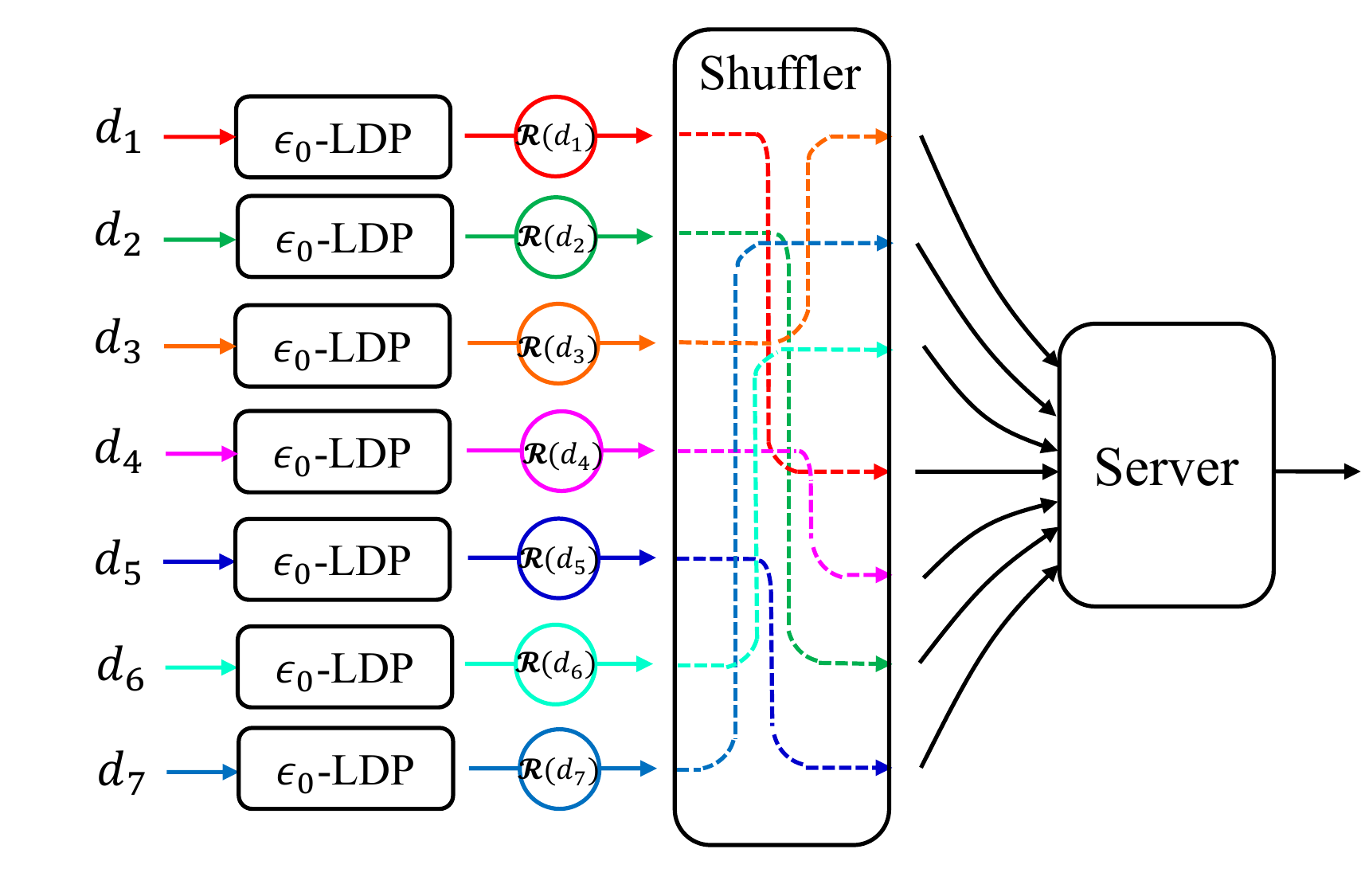}
\caption{In the shuffled model, each client sends the output of the local randomizer $\calR$ to a secure shuffler that randomly permutes the clients' messages before passing them to the server.}
\label{fig:setup-figure}
\end{figure}
Let $\calD=\left(d_1,\ldots,d_n\right)$ be a dataset consisting of $n$ data points, where $d_i$ is a data point at the $i$'th client that takes values from a set $\calX$. Let $\calR:\calX\to\calY$ be a local randomizer that satisfies the following two properties:
\begin{enumerate}
\item $\calR$ is an $\epsilon_0$-LDP mechanism (see Definition~\ref{defn:LDPdef}). 
\item  The range of $\calR$ is a discrete set, i.e., the output of $\calR$ takes values in a discrete set $\left[B\right]=\left\{1,\ldots,B\right\}$ for some $B\in\mathbb{N}:=\{1,2,3,\hdots\}$. Here, $[B]$ could be the whole of $\mathbb{N}$. 
\end{enumerate}
Client $i$ applies $\calR$ on $d_i$ (each client uses independent randomness for computing $\calR(d_i)$) and sends $\calR(d_i)$ to the shuffler, who shuffles the received $n$ inputs and outputs the result; see Figure~\ref{fig:setup-figure}. To formalize this, let $\calH_n:\calY^n\to\calY^n$ denote the shuffling operation that takes $n$ inputs and outputs their uniformly random permutation. We define the shuffling mechanism as
\begin{equation}\label{shuffle-mech}
\calM\left(\calD\right) := \calH_{n}\left(\calR\left(d_1\right),\ldots,\calR\left(d_n\right)\right).
\end{equation}

Our goal is to characterize the Renyi differential privacy of $\calM$. 

Since the output of $\calM$ is a random permutation of the $n$ outputs of $\calR$, the server cannot associate the $n$ messages to the clients; and the only information it can use from the messages is the histogram, i.e., the number of messages that give any particular output in $[B]$. 
We define a set $\calA_{B}^{n}$ as follows
\begin{equation}\label{histogram-set}
\calA_B^{n}=\bigg\lbrace \bh=\left(h_1,\ldots,h_B\right):\sum_{j=1}^{B}h_j=n\bigg\rbrace,
\end{equation}
to denote the set of all possible histograms of the output of the shuffler with $n$ inputs. Therefore, we can assume, without loss of generality (w.l.o.g.), that the output of $\calM$ is a distribution over $\calA_B^{n}$ for input dataset $\calD$ of $n$ data points.

The notation used throughout the paper is given in Table~\ref{Tab_summary}.



%
\begin{table}[t!]
\centering
\begin{tabular}{ | c | c |   }
\hline \hline
Symbol & Description \\ \hline \hline
$[B]$ & $\{1,2,\hdots,B\}$ for any $B\in\bbN$ \\ \hline
$\eps_0$ & LDP parameter (see Definition~\ref{defn:LDPdef}) \\ \hline
$(\eps,\delta)$ & Approximate DP parameters (see Definition~\ref{defn:central-DP}) \\ \hline
$(\lambda,\eps(\lambda))$ & RDP parameters (see Definition~\ref{defn:RDP}) \\ \hline
$\calR:\calX\to[B]$ & A discrete $\eps_0$-LDP mechanism at clients for mapping \\
&  their data points to elements in $[B]$ \\ \hline
$\bp=(p_1,\hdots,p_B)$ & The output distribution of $\calR$ when the data point is $d$ \\ \hline
$\bp'=(p'_1,\hdots,p'_B)$ & The output distribution of $\calR$ when the data point is $d'$ \\ \hline
$\bp_i=(p_{i1},\hdots,p_{iB})$ for $i\in[n]$ & The output distribution of $\calR$ when the data point is $d_i$ \\ \hline
$\bp'_n=(p'_{n1},\hdots,p'_{nB})$ & The output distribution of $\calR$ when the data point is $d'_n$ \\ \hline
$\calP$ & A collection of $n$ distribution $\{\bp_1,\hdots,\bp_n\}$ \\ \hline
$\calP_{-i}$ & A collection of $(n-1)$ distribution $\calP\setminus\{\bp_i\}$ \\ \hline
& A collection of $n$ distribution, where clients in the set $\calC$ map \\ 
$\calP_{\calC}$, where $\calC\subseteq[n-1]$ & according to $\bp_n'$, clients in the set $[n-1]\setminus\calC$ map according \\ 
& to $\tilde{\bp}_i$ (see \eqref{eq:mixture-dist}), and client $n$ maps according to $\bp_n$ (see \eqref{eq:defn_P_C}-\eqref{eq:defn_hatP}) \\ \hline
$\calA_B^n$ & A set of histograms with $B$ bins and $n$ elements (see \eqref{histogram-set}) \\ \hline
$\bh$ & $\bh=(h_1,\hdots,h_B)$ with $\sum_{i=1}^Bh_i=n$ is an element of $\calA_B^n$  \\ \hline
$\calM(\calD)$ & The shuffled mechanism $\calM$ on the dataset $\calD\in\calX^n$; \\
& $\calM(\calD)$ is a distribution over $\calA_B^n$ (see \eqref{shuffle-mech}) \\ \hline
$F(\calP)$ & Distribution over $\calA_B^n$ when client $i$ maps its data point \\
& according to the distribution $\bp_i$ (see \eqref{general-distribution}) \\ \hline \hline
\end{tabular}
\caption{Notation used throughout the paper}\label{Tab_summary}
\end{table}

\section{Main Results}\label{sec:main_results}


This section is dedicated to presenting the main results of this paper, along with their implications and comparisons with related work. We state two different upper bounds on the RDP of the shuffle model in Section~\ref{subsec:upper-bounds} and a lower bound in Section~\ref{subsec:lower-bound}. We present a detailed proof-sketch of our first upper bound in Section~\ref{subsec:proof-sketch_general-case}, along with all the main ingredients required to prove the upper bound.

\subsection{Upper Bounds}\label{subsec:upper-bounds}
In this subsection, we will provide two upper bounds.
\begin{theorem}[Upper Bound 1]\label{thm:general_case} 
For any $n\in\mathbb{N}$, $\epsilon_0\geq 0$,and any integer $\lambda\geq 2$, the RDP of the shuffle model is upper-bounded by
\begin{equation}\label{eqn:1st_bound}
\epsilon(\lambda) \leq \frac{1}{\lambda-1}\log\left(1+\binom{\lambda}{2} \frac{\left(e^{\epsilon_0}-1\right)^2}{\overline{n}e^{\epsilon_0}}+\sum_{i=3}^{\lambda} \binom{\lambda}{i} i\Gamma\left(i/2\right) \left(\frac{\left(e^{2\epsilon_0}-1\right)^2}{2e^{2\epsilon_0}\overline{n}}\right)^{i/2}+e^{\epsilon_0\lambda-\frac{n-1}{8e^{\epsilon_0}}}\right),
\end{equation}
where $\overline{n}=\floor{\frac{n-1}{2e^{\epsilon_0}}}+1$ and $\Gamma\left(z\right)=\int_{0}^{\infty}x^{z-1}e^{-x}dx$ is the gamma function. 
\end{theorem}
We give a proof sketch of Theorem~\ref{thm:general_case} in Section~\ref{subsec:proof-sketch_general-case} and provide its complete proof in Section~\ref{subsec:proof_general_case}.
%
%

When $n,\eps_0,\lambda$ satisfy a certain condition, we can simplify the bound in \eqref{eqn:1st_bound}, and the simplified bound is stated in the following corollary.
\begin{corollary}[Simplified Upper Bound 1]\label{corol:simplified_general_case} 
For any $n\in\mathbb{N}$, $\epsilon_0\geq 0$, and any integer $\lambda\geq 2$ that satisfy $\lambda^4e^{5\eps_0}<\frac{n}{9}$, we can simplify the bound in \eqref{eqn:1st_bound} to the following: 
\begin{equation}\label{eq:simplified_1st-bound}
\epsilon(\lambda) \leq \frac{1}{\lambda-1}\log\left(1+\binom{\lambda}{2}\frac{4\left(e^{\epsilon_0}-1\right)^2}{n} \right).
\end{equation}
\end{corollary}
We prove Corollary~\ref{corol:simplified_general_case} in Appendix~\ref{app:simplied_1st-bound-proof}.

Note that the upper bounds in Theorem~\ref{thm:general_case} and Corollary~\ref{corol:simplified_general_case} hold for any $\eps_0$-LDP mechanism. 
\begin{remark}\label{remark:simplified_1st-bound}
Note that any $\lambda,\eps_0,n$ that satisfy $\lambda^4e^{5\eps_0}<\frac{n}{9}$ lead to the bound in \eqref{eq:simplified_1st-bound}. For example, we can take $\eps_0=c\ln n$ and $\lambda<\frac{n^{(1-5c)/4}}{2}$ for any $c<\frac{1}{5}$, and it will satisfy the condition. In particular, taking $\eps_0=\frac{1}{25}\ln n$ and $\lambda < \frac{n^{1/5}}{2}$ will also give the bound in \eqref{eq:simplified_1st-bound}.
\end{remark}
\begin{remark}[Generalization to real orders $\lambda$] Theorem~\ref{thm:general_case} provides an upper bound on the RDP of the shuffled model for only integer orders $\lambda\geq 2$. However, the result can be generalized to real orders $\lambda$ using convexity of the function $\left(\lambda-1\right)\epsilon\left(\lambda\right)$ as follows. From~\cite[Corollary~$2$]{HarremoesRenyiKL14}, the function $\left(\lambda-1\right)D_{\lambda}\left(\mathbf{P}||\mathbf{Q}\right)$ is convex in $\lambda$ for given two distributions $\mathbf{P}$ and $\mathbf{Q}$. Thus, for any real order $\lambda> 1$, we can bound the RDP of the shuffled model by
\begin{equation}
\epsilon\left(\lambda\right)\leq\frac{a(\floor{\lambda}-1)\epsilon\left(\floor{\lambda}\right)+\left(1-a\right)(\ceil{\lambda}-1)\epsilon\left(\ceil{\lambda}\right)}{\lambda-1},
\end{equation}
where $a=\ceil{\lambda}-\lambda$, since $\lambda=a\floor{\lambda}+(1-a)\ceil{\lambda}$ for any real $\lambda$. Here, $\floor{\lambda}$ and $\ceil{\lambda}$ respectively denote the largest integer smaller than or equal to $\lambda$ and the smallest integer bigger than or equal to $\lambda$.
\end{remark}
%
In the following theorem, we also present another bound on RDP that readily holds for all $\lambda\geq1$.
\begin{theorem}[Upper Bound 2]\label{thm:general_case_2} 
For any $n\in\mathbb{N}$, $\epsilon_0\geq 0$, and any $\lambda\geq 1$ (including the non-integral $\lambda$), the RDP of the shuffle model is upper-bounded by
\begin{equation}\label{eqn:2nd_bound}
\epsilon(\lambda) \leq \frac{1}{\lambda-1}\log\left(e^{\lambda^2\frac{\left(e^{\epsilon_0}-1\right)^2}{\overline{n}}}+e^{\epsilon_0\lambda-\frac{n-1}{8e^{\epsilon_0}}}\right),
\end{equation} 
where $\overline{n}=\floor{\frac{n-1}{2e^{\epsilon_0}}}+1$.
\end{theorem}
We prove Theorem~\ref{thm:general_case_2} in Section~\ref{subsec:proof_general_case_2}.
\begin{remark}[Improved Upper Bounds -- Saving a Factor of $2$]\label{remark:save-factor-2}
The exponential term $e^{\epsilon_0\lambda-\frac{n-1}{8e^{\epsilon_0}}}$ in both the upper bounds stated in \eqref{eqn:1st_bound} and \eqref{eqn:2nd_bound} come from the Chernoff bound, where we naively choose the factor $\gamma=1/2$ instead of optimizing it; see the proof of Theorem~\ref{thm:general_case} in Section~\ref{subsec:proof_general_case}. If we instead had optimized $\gamma$ and chosen it to be, for example, $\gamma=\sqrt{\frac{2\epsilon_0e^{\epsilon_0}}{\sqrt{n}\log(n)}}$ (which goes to $0$ when, say, $\epsilon_0\leq\frac{1}{4}\log(n)$), we would have asymptotically saved a multiplicative factor of $2$ in the leading term in both upper bounds, because in this case we have $\overline{n}=\floor{(1-\gamma)\frac{n-1}{e^{\eps_0}}}+1\to\floor{\frac{n-1}{e^{\eps_0}}}+1$ as $n\to\infty$. We chose to evaluate our bound with $\gamma=1/2$ because of two reasons: first, it gives a simpler expression to compute; and the second, the evaluated bound does not give good results (as compared to the ones with $\gamma=1/2$) for the parameter ranges of interest.
\end{remark}
\begin{remark}[Difference in Upper Bounds]
Since the quadratic term in $\lambda$ inside the $\log$ in \eqref{eqn:2nd_bound} has an extra multiplicative factor of $e^{\eps_0}$ in comparison with the corresponding term in \eqref{eqn:1st_bound}, our  first upper bound presented in Theorem~\ref{thm:general_case} is better than our second upper bound presented in Theorem~\ref{thm:general_case_2} for all parameter ranges of interest; see also Figure~\ref{fig:RDP_comp} in Section~\ref{sec:numerics}. However, the expression in \eqref{eqn:2nd_bound} is much cleaner to state as well as to compute as compared to that in \eqref{eqn:1st_bound}. As we will see later, the techniques required to prove both upper bounds are different.
\end{remark}
\begin{remark}[Potentially Better Upper Bounds for Specific Mechanisms]\label{rem:Pot-Impr}
  Since both our upper bounds are worse-case bounds that hold for {\em
    all} $\eps_0$-LDP mechanisms, it is possible that for specific
  mechanisms, we may be able to exploit their structure for
  potentially better bounds. See Remark \ref{rem:tightenBnd} on this
  just after \eqref{ratio_mu}.
\end{remark}
The upper bounds on the RDP of the shuffled model presented
in~\eqref{eqn:1st_bound} and \eqref{eqn:2nd_bound} are general and
hold for any discrete $\epsilon_0$-LDP mechanism. Furthermore, these
bounds are in closed form expressions that can be easily
implemented. To the best of our knowledge, there is no bound on RDP of
the shuffled model in literature except for the one given
in~\cite[Remark~$1$]{erlingsson2019amplification}, which we provide
below\footnote{As mentioned in Section \ref{sec:introduction}, this was
  obtained by the standard conversion results from DP to
  RDP, which could be loose.} in \eqref{eqn:rdp_erlin-1}. For the LDP
parameter $\epsilon_0$ and number of clients $n$, they showed that for
any $\lambda>1$, the shuffled mechanism $\calM$ is
$\left(\lambda,\epsilon\left(\lambda\right)\right)$-RDP, where
\begin{equation}\label{eqn:rdp_erlin-1}
\epsilon\left(\lambda\right)=\lambda\frac{2e^{4\epsilon_0}\left(e^{\epsilon_0}-1\right)^2}{n}.
\end{equation}
In Section~\ref{sec:numerics}, we evaluate numerically the performance of both our bounds (from Theorems~\ref{thm:general_case} and \ref{thm:general_case_2}) against the above bound in~\eqref{eqn:rdp_erlin-1}. We demonstrate that both our bounds outperform the above bound in all cases; and in particular, the gap is significant when $\epsilon_0>1$ -- 
note that the bound in~\cite{erlingsson2019amplification} is worse than our simplified bound  given in Corollary~\ref{corol:simplified_general_case} by a multiplicative factor of $e^{4\epsilon_0}$. 

\subsection{Lower Bound}\label{subsec:lower-bound}
In this subsection, we provide a lower bound on the RDP for any integer order $\lambda$ satisfying $\lambda\geq2$.
\begin{theorem}[Lower Bound]\label{thm:lower_bound}
For any $n\in\mathbb{N}$, $\epsilon_0\geq 0$, and any integer $\lambda\geq 2$, the RDP of the shuffle model is lower-bounded by:
\begin{equation}\label{eq:lower-bound}
\epsilon\left(\lambda\right) 
\geq \frac{1}{\lambda-1}\log\left(1+\binom{\lambda}{2}\frac{\left(e^{\epsilon_0}-1\right)^2}{ne^{\epsilon_0}}+\sum_{i=3}^{\lambda} \binom{\lambda}{i}  \left(\frac{\left(e^{2\epsilon_0}-1\right)}{ne^{\epsilon_0}}\right)^{i}\mathbb{E}\left[\left(k-\frac{n}{e^{\epsilon_0}+1}\right)^{i}\right]\right), 
\end{equation}
where expectation is taken w.r.t.\ the binomial random variable $k\sim\text{Bin}\left(n,p\right)$ with parameter $p=\frac{1}{e^{\epsilon_0}+1}$. 
\end{theorem}
We prove Theorem~\ref{thm:lower_bound} in Section~\ref{sec:lower-bound}.

When $i$ is an even integer, then the expectation term in \eqref{eq:lower-bound} is positive.
When $i\geq3$ is an odd integer, then using the convexity of function $f(x)=x^i$, it follows from the Jensen's inequality (i.e., $\bbE f(X) \geq f(\bbE X)$) and $\bbE[k]=\frac{n}{e^{\eps_0}+1}$, that $\mathbb{E}\left[\left(k-\frac{n}{e^{\epsilon_0}+1}\right)^{i}\right] \geq \left(\mathbb{E}\left[k-\frac{n}{e^{\epsilon_0}+1}\right]\right)^{i}=0$. Using these observations, we can safely ignore the summation term from \eqref{eq:lower-bound} and obtain the following simplified lower bound.
\begin{corollary}[Simplified Lower Bound]\label{corol:simplified_lower-bound}
For any $n\in\mathbb{N}$, $\epsilon_0\geq 0$, and any integer $\lambda\geq 2$, the RDP of the shuffle model is lower-bounded by:
\begin{equation}\label{eq:simplified_lower-bound}
\epsilon\left(\lambda\right) 
\geq \frac{1}{\lambda-1}\log\left(1+\binom{\lambda}{2}\frac{\left(e^{\epsilon_0}-1\right)^2}{ne^{\epsilon_0}}\right).
\end{equation}
\end{corollary}

\begin{remark}[Upper and Lower Bound Proofs]\label{remark:ub-lb}
Both our upper bounds stated in Theorems~\ref{thm:general_case} and \ref{thm:general_case_2} hold for any $\eps_0$-LDP mechanism. In other words, they are the worse case privacy bounds, in the sense that there is no $\eps_0$-LDP mechanism for which the associated shuffle model gives a higher RDP parameter than those stated in \eqref{eqn:1st_bound} and \eqref{eqn:2nd_bound}. Therefore, the lower bound that we derive should serve as the lower bound on the RDP privacy parameter of the mechanism that achieves the largest privacy bound (i.e., worst privacy).

We prove our lower bound result (stated in Theorem~\ref{thm:lower_bound}) by showing that a specific mechanism (in particular, the binary Randomized response (RR)) on a specific pair of neighboring datasets yields the RDP privacy parameter stated in the RHS of \eqref{eq:lower-bound}. This implies that RDP privacy bound (which is the supremum over all neighboring datasets) of binary RR for the shuffle model is at least the bound stated in \eqref{eq:lower-bound}, which in turn implies that the lower bound (which is the tightest bound for any $\eps_0$-LDP mechanism) is also at least that.
\end{remark}

\begin{remark}[Gap in Upper and Lower Bounds]\label{remark:gap-up-lb}
When comparing our simplified upper and lower bounds from Corollaries~\ref{corol:simplified_general_case} and \ref{corol:simplified_lower-bound}, respectively, we observe that when $\lambda^4e^{5\eps_0}<\frac{n}{9}$, our upper and lower bounds are away by a multiplicative factor of $4e^{\epsilon_0}$. In our generic upper bound~\eqref{eqn:1st_bound}, note that when $n$ is large, only the term corresponding to $\lambda^2$ matters, and with our improved upper bound (which saves a factor of $2$ in that term asymptotically -- see Remark~\ref{remark:save-factor-2}), the upper and lower bounds are away by the factor of $e^{\eps_0}$, which tends to $1$ as $\eps_0\to0$. Thus, in the regime of large $n$ and small $\eps_0$, our upper and lower bounds coincide. 
Without any constraints on $n,\eps_0$, we believe that our lower bound is tight. Closing this gap by showing a tighter upper bound is an interesting and important open problem. 
\end{remark}

\subsection{Proof Sketch of Theorem~\ref{thm:general_case}}\label{subsec:proof-sketch_general-case}
The proof has two main steps. In the first step, we reduce the problem of deriving RDP for arbitrary neighboring datasets to the problem of deriving RDP for specific neighboring datasets, $\calD,\calD'$, where all elements in $\calD$ are the same and $\calD'$ differs from $\calD$ in one entry. In the second step, we derive RDP for the special neighboring datasets. Details follow:

The specific neighboring datasets to which we reduce our general problem has the following form:
\begin{align}\label{datasets-same}
\calD_{\same}^{m} = \left\lbrace (\calD_{m},\calD'_{m}): \calD_{m}= (d,\ldots,d,d)\in\calX^{m}, \,\,\,\,\calD'_{m}=(d,\ldots,d,d')\in\calX^{m}, \text{ where } d,d'\in\calX \right\rbrace.
\end{align}

Consider arbitrary neighboring datasets $\calD=\left(d_1,\ldots,d_n\right)\in\calX^{n}$ and $\calD'=\left(d_1,\ldots,d_{n-1},d_n'\right)\in\calX^{n}$. For any $m\in\left\{0,\ldots,n-1\right\}$, define new neighboring datasets $\calD_{m+1}^{(n)}=\left(d'_n,\ldots,d'_n,d_{n}\right)\in \calX^{m+1}$ and $\calD_{m+1}'^{(n)}=\left(d'_n,\ldots,d'_n,d_n'\right)\in \calX^{m+1}$, each having $(m+1)$ elements. Observe that $\left(\calD_{m+1}'^{(n)},\calD_{m+1}^{(n)}\right)\in\calD^{m+1}_{\same}$. The first step of our proof is summarized in the following theorem.
\begin{theorem}[Reduction to the Special Case]\label{Thm:reduce_special_case} 
Let $q=\frac{1}{e^{\epsilon_0}}$ and $m\sim\emph{Bin}\left(n-1,q\right)$ be a binomial random variable. We have: 
\begin{align}
\mathbb{E}_{\bh\sim\calM(\calD')}\left[\left(\frac{\calM(\calD)(\bh)}{\calM(\calD')(\bh)}\right)^{\lambda}\right] &\leq \mathbb{E}_{m\sim\emph{Bin}\left(n-1,q\right)}\left[\mathbb{E}_{\bh\sim\calM(\calD_{m+1}'^{(n)})}\left[\left(\frac{\calM(\calD_{m+1}^{(n)})(\bh)}{\calM(\calD_{m+1}'^{(n)})(\bh)}\right)^{\lambda}\right]\right]. \label{eq:reduce_special-case_bound}
\end{align}
\end{theorem} 
We give a proof-sketch of Theorem~\ref{Thm:reduce_special_case} in Section~\ref{subsubsec:proof-sketch_reduce-special-case} and provide its complete proof in Section~\ref{sec:proof_reduce_special_case}.

We know (by Chernoff bound) that the binomial random variable is concentrated around its mean, which implies that the terms in the RHS of \eqref{eq:reduce_special-case_bound} that correspond to $m<(1-\gamma)q(n-1)$ (we will take $\gamma=1/2$) will contribute in a negligible amount. Then we show in Lemma~\ref{lem:E_m-decreasing} (on page~\pageref{lem:E_m-decreasing}) that $E_{m}:=\mathbb{E}_{\bh\sim\calM(\calD_{m+1}'^{(n)})}\left[\left(\frac{\calM(\calD_{m+1}^{(n)})(\bh)}{\calM(\calD_{m+1}'^{(n)})(\bh)}\right)^{\lambda}\right]$ is a non-increasing function of $m$. These observation together imply that the RHS in \eqref{eq:reduce_special-case_bound} is approximately equal to $E_{(1-\gamma)q(n-1)}$.

Since $E_m$ is precisely what is required to bound the RDP for the specific neighboring datasets, we have reduced the problem of computing RDP for arbitrary neighboring datasets to the problem of computing RDP for specific neighboring datasets. The second step of the proof bounds $E_{(1-\gamma)q(n-1)}$, which follows from the result below that holds for any $m\in\bbN$.
\begin{theorem}[RDP for the Special Case]\label{thm:RDP_same} 
Let $m\in\bbN$ be arbitrary. 
For any integer $\lambda\geq 2$, we have
\begin{equation}\label{RDP_same_bound1_sketch}
\sup_{(\calD_m,\calD'_m)\in\calD_{\emph{same}}^m}\mathbb{E}_{\bh\sim\calM(\calD_{m})}\left[\left(\frac{\calM(\calD'_{m})(\bh)}{\calM(\calD_{m})(\bh)}\right)^{\lambda}\right]\leq 1+\binom{\lambda}{2} \frac{\left(e^{\epsilon_0}-1\right)^2}{m e^{\epsilon_0}} + \sum_{i=3}^{\lambda} \binom{\lambda}{i} i\Gamma(i/2) \Bigg(\frac{\left(e^{2\epsilon_0}-1\right)^2}{2m e^{2\epsilon_0}}\Bigg)^{i/2}
\end{equation}
\end{theorem}
We give a proof-sketch of Theorem~\ref{thm:RDP_same} in Section~\ref{subsubsec:proof-sketch_rdp-same} and provide its complete proof in Section~\ref{sec:special_form}.

Substituting $m=(1-\gamma)q(n-1)+1$ in \eqref{RDP_same_bound1_sketch} yields the bound in Theorem~\ref{thm:general_case}.

\subsubsection{Proof Sketch of Theorem~\ref{Thm:reduce_special_case}}\label{subsubsec:proof-sketch_reduce-special-case}
%
%

For $i\in[n]$, let $\bp_i$ denote the distribution of the $\eps_0$-LDP mechanism $\calR$ when the input data point is $d_i$, and $\bp'_n$ denote the distribution of $\calR$ when the input data point is $d'_n$. The main idea of the proof is the observation that each distribution $\bp_i$ can be written as the following mixture distribution:  
$\bp_i = \frac{1}{e^{\epsilon_0}}\bp'_n+\left(1-\frac{1}{e^{\epsilon_0}}\right)\tilde{\bp}_i$,
where $\tilde{\bp}_i$ is a certain distribution associated with $\bp_i$. 
So, instead of client $i\in[n-1]$ mapping its data point $d_i$ according to $\bp_i$, we can view it as the client $i$ maps $d_i$ according to $\bp_n'$ with probability $\frac{1}{e^{\epsilon_0}}$ and according to $\tilde{\bp}_i$ with probability $(1-\frac{1}{e^{\epsilon_0}})$. Thus the number of clients 
that sample from the distribution $\bp'_n$ follows a binomial distribution Bin($n-1,q$) with parameter $q=\frac{1}{e^{\epsilon_0}}$. 
This allows us to write the distribution of $\calM$ when clients map their data points according to $\bp_1,\hdots,\bp_n,\bp'_n$ as a convex combination of the distribution of $\calM$ when clients map their data points according to $\tilde{\bp}_1,\hdots,\tilde{\bp}_{n-1},\bp_n,\bp'_n$; see Lemma~\ref{lem:convex-combinations}. 
Then using a joint convexity argument (see Lemma~\ref{lem:joint-con_renyi-exp}), we write the Renyi divergence between the original pair of distributions of $\calM$ in terms of the same convex combination of the Renyi divergence between the resulting pairs of distributions of $\calM$ as in Lemma~\ref{lem:convex-combinations}.
Using a monotonicity argument (see Lemma~\ref{lem:cvx_rdp}), we can remove the effect of clients that do not sample from the distribution $\bp'_n$ without decreasing the Renyi divergence.
By this chain of arguments, we have reduced the problem to the one involving the computation of Renyi divergence only for the special form of neighboring datasets, which proves Theorem~\ref{Thm:reduce_special_case}. Details can be found in Section~\ref{sec:proof_reduce_special_case}.

\subsubsection{Proof Sketch of Theorem~\ref{thm:RDP_same}}\label{subsubsec:proof-sketch_rdp-same}
Consider any pair of special neighboring datasets $(\calD_m,\calD'_m)\in\calD_{\same}^m$ for any $m\in\bbN$. 
Using the polynomial expansion, we get
\begin{align}
\mathbb{E}_{\bh\sim\calM(\calD_{m})}\left[\left(\frac{\calM(\calD'_{m})(\bh)}{\calM(\calD_{m})(\bh)}\right)^{\lambda}\right] &= \sum_{i=0}^{\lambda}\binom{\lambda}{i}\bbE_{\bh\sim\calM(\calD_{m})}\left[\left(\frac{\calM(\calD'_{m})(\bh)}{\calM(\calD_{m})(\bh)}-1\right)^{i}\right]. \label{Proof-same_Taylor-exp}
\end{align}
Let $X:\calA_B^m\to\bbR$ denote a random variable (r.v.) associated with the distribution $\calM(\calD_m)$, 
and for every $\bh\in\calA_B^m$, is defined as $X(\bh)=m\(\frac{\calM(\calD'_{m})(\bh)}{\calM(\calD_{m})(\bh)}-1\)$. With this, we can rewrite \eqref{Proof-same_Taylor-exp} in terms of the moments of $X$. 
Then we show that $X$ is a sub-Gaussian r.v.\ that has zero-mean and bounded variance. Using the sub-Gaussianity of $X$, we bound its higher moments (see Lemma~\ref{lemm:MomAcc_RV}). Substituting these bounds in \eqref{Proof-same_Taylor-exp} proves Theorem~\ref{thm:RDP_same}. Details can be found in Section~\ref{sec:special_form}.

\section{Numerical Results}\label{sec:numerics}

In this section, we present numerical experiments to show the performance of our bounds on the RDP of the shuffled model and its usage for getting approximate DP and composition results.

\begin{figure}[t]
    \centering 
\begin{subfigure}{0.31\textwidth}
  \includegraphics[scale=0.39]{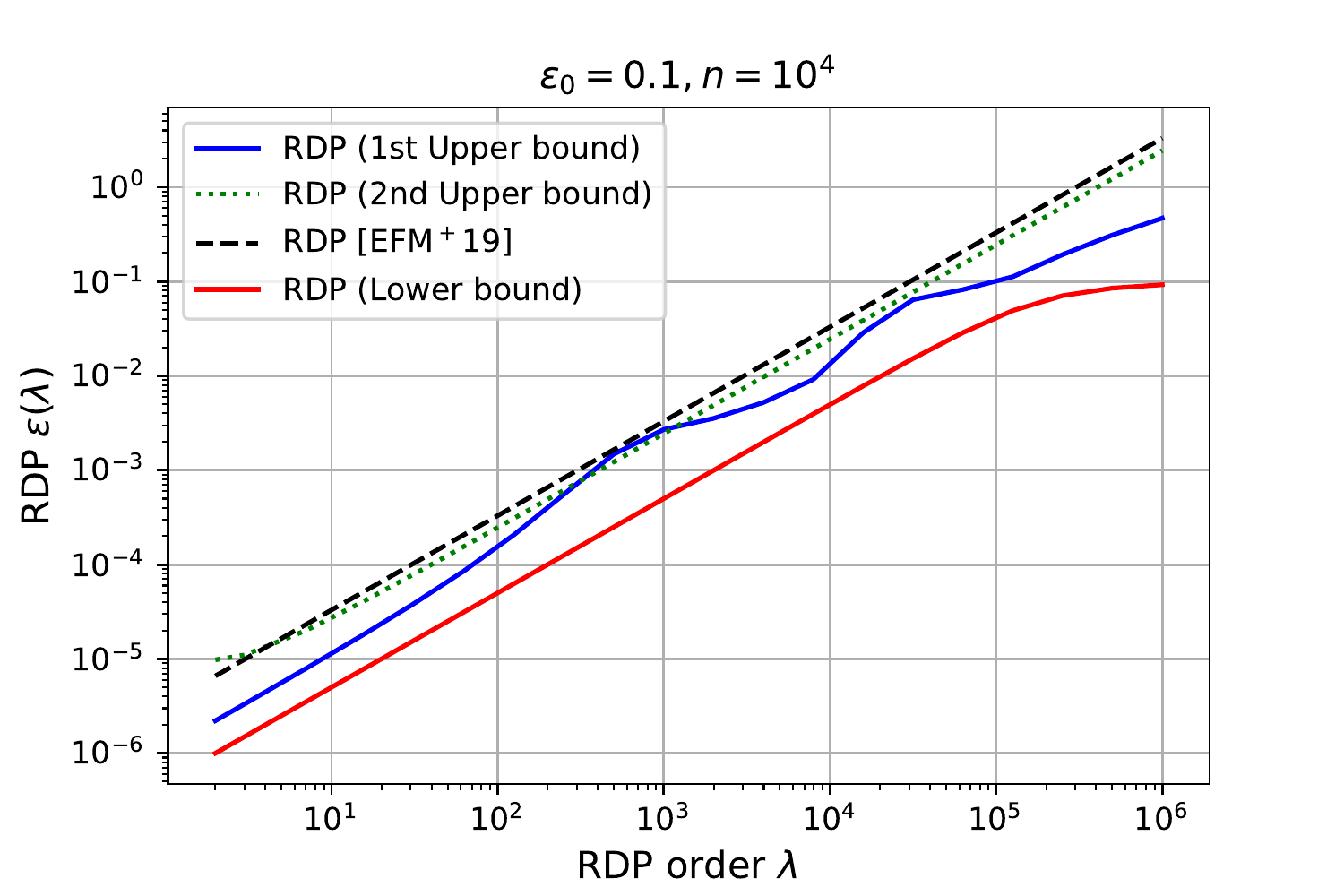}
  \caption{RDP as a function of $\lambda$ for $\epsilon_0=0.1$ and $n=10^4$ }
  \label{fig:1}
\end{subfigure}\hfil 
\begin{subfigure}{0.31\textwidth}
  \includegraphics[scale=0.39]{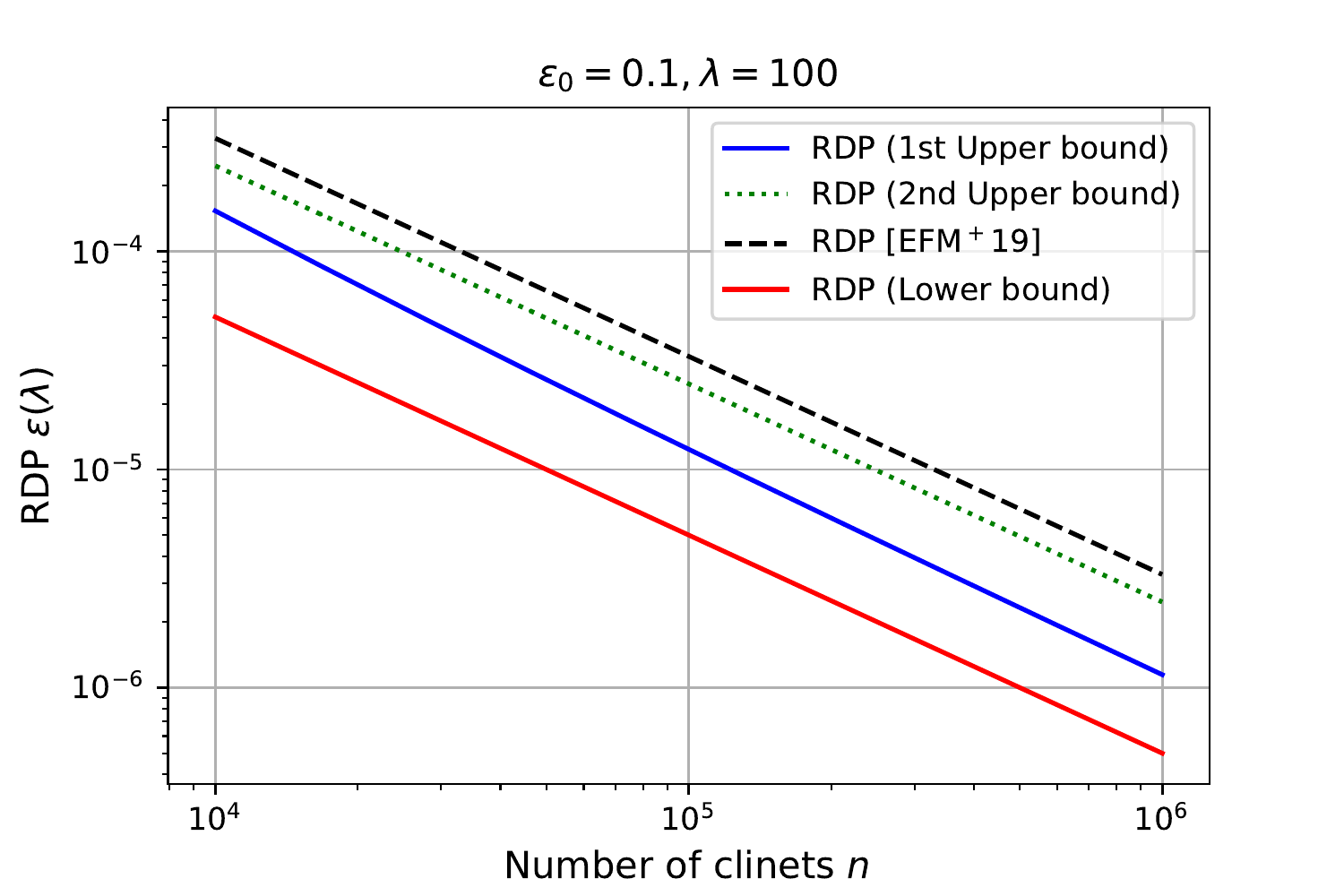}
  \caption{RDP as a function of $n$ for $\epsilon_0=0.1$ and $\lambda=100$}
  \label{fig:2}
\end{subfigure}\hfil 
\begin{subfigure}{0.31\textwidth}
  \includegraphics[scale=0.39]{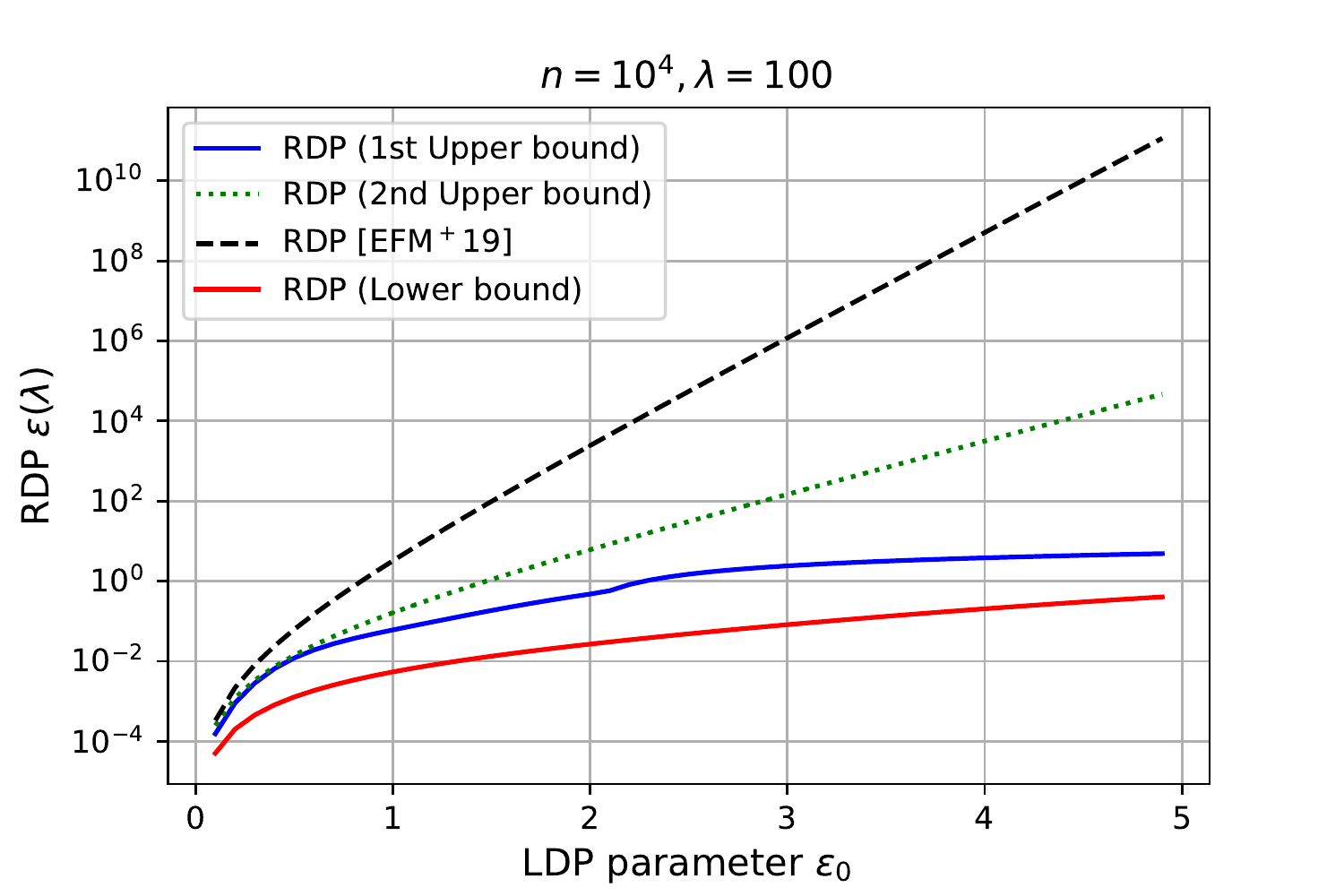}
  \caption{RDP as a function of $\epsilon_0$ for $n=10^4$ and $\lambda=100$}
  \label{fig:3}
\end{subfigure}

\medskip
\begin{subfigure}{0.31\textwidth}
  \includegraphics[scale=0.39]{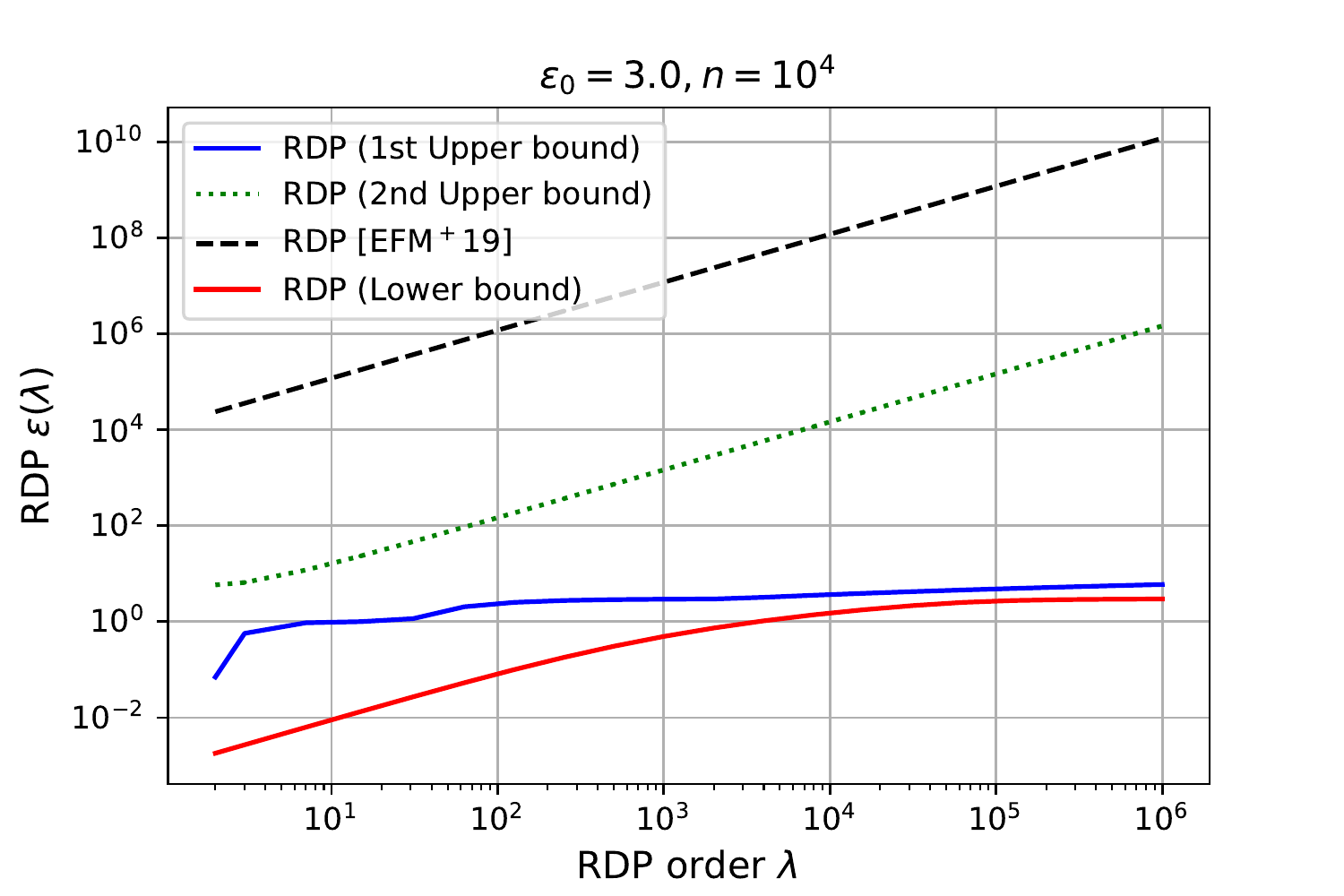}
  \caption{RDP as a function of $\lambda$ for $\epsilon_0=3$ and $n=10^4$}
  \label{fig:4}
\end{subfigure}\hfil 
\begin{subfigure}{0.31\textwidth}
  \includegraphics[scale=0.39]{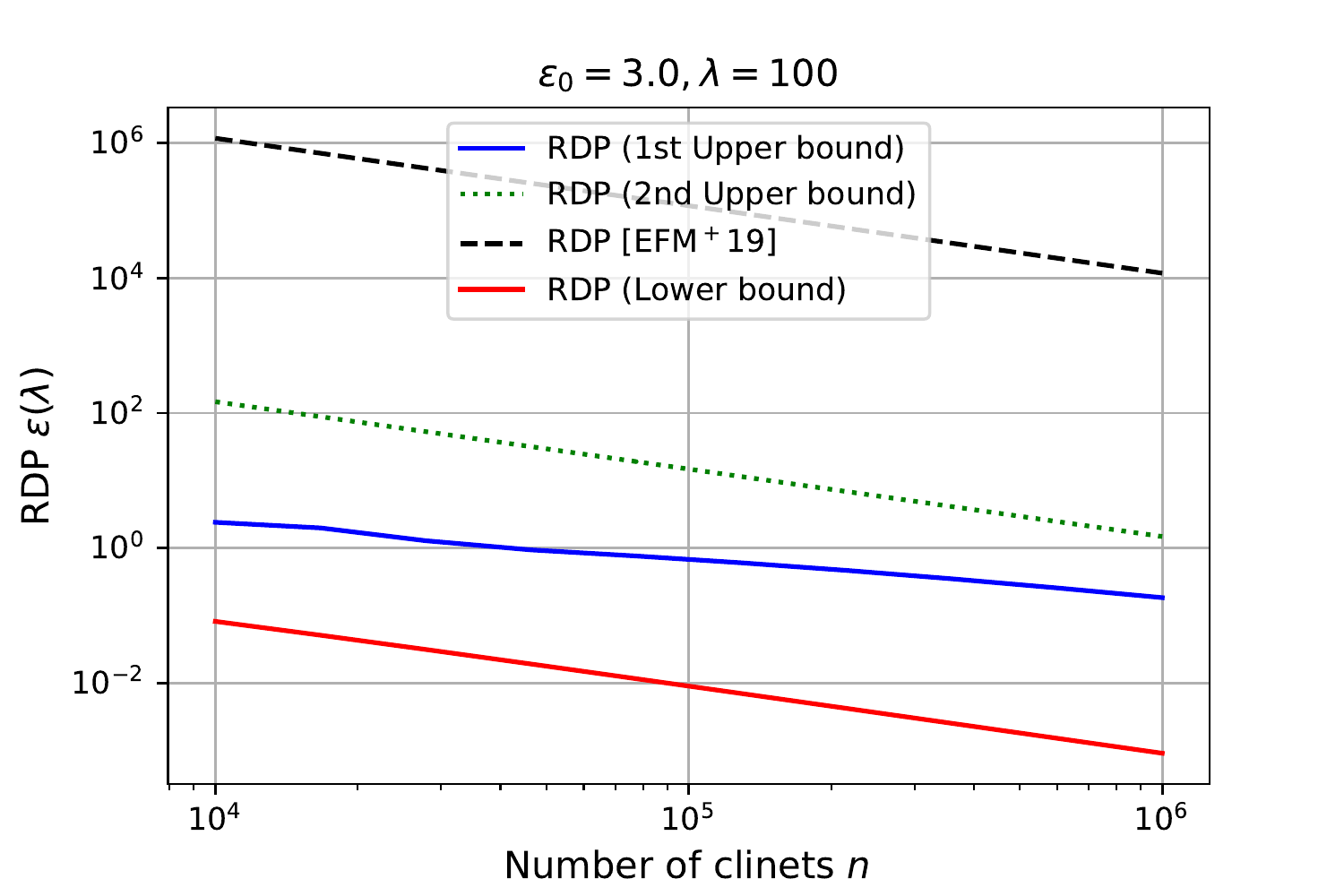}
  \caption{RDP as a function of $n$ for $\epsilon_0=3$ and $\lambda=100$}
  \label{fig:5}
\end{subfigure}\hfil 
\begin{subfigure}{0.31\textwidth}
  \includegraphics[scale=0.39]{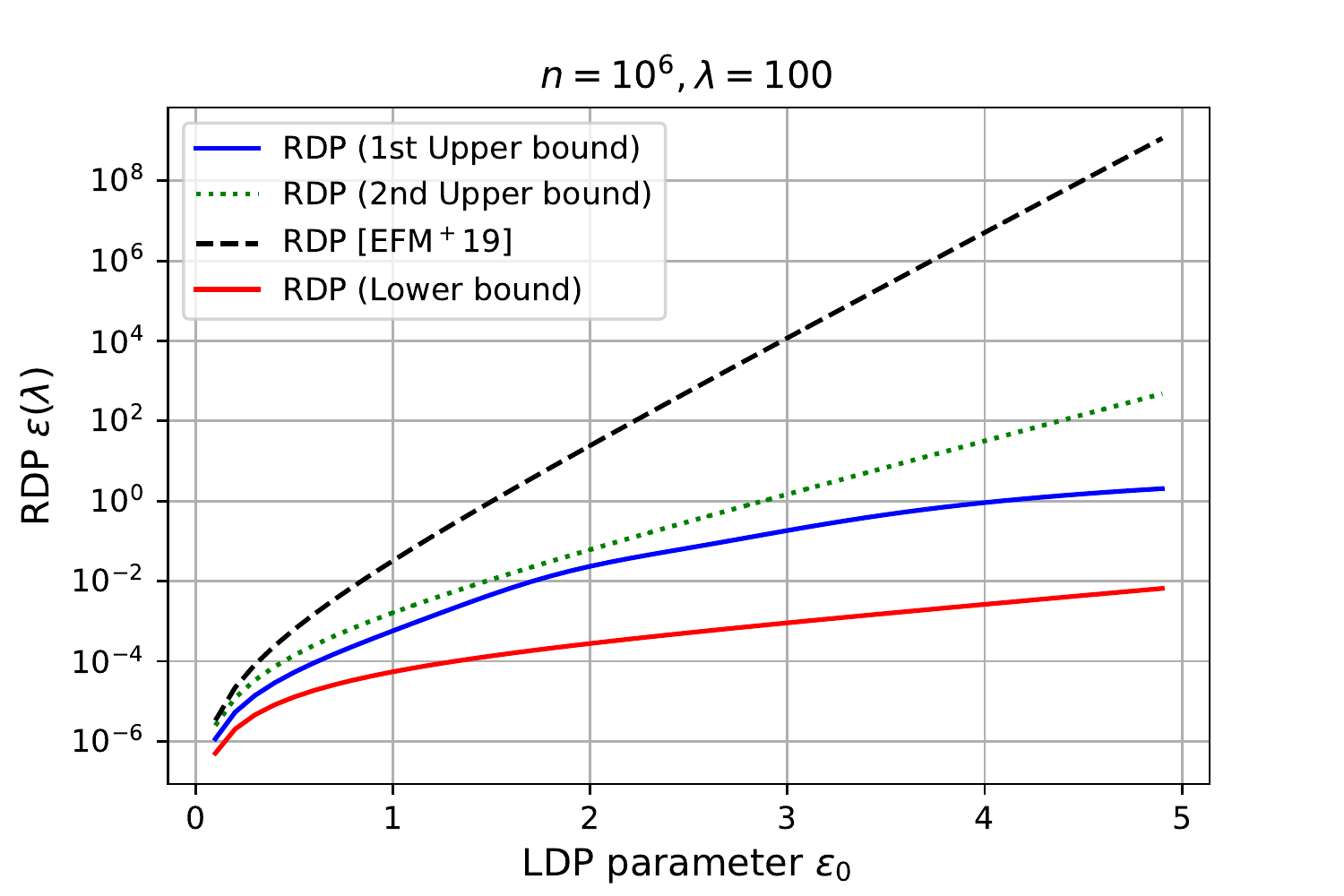}
  \caption{RDP as a function of $\epsilon_0$ for $n=10^6$ and $\lambda=100$}
  \label{fig:6}
\end{subfigure}
\caption{Comparison of several bounds on the RDP of the shuffled model: {\sf (i)} Our first upper bound~\eqref{eqn:1st_bound} in Theorem~\ref{thm:general_case}. {\sf (ii)} Our second upper bound~\eqref{eqn:2nd_bound} in Theorem~\ref{thm:general_case}. {\sf (iii)} Our lower bound proposed in Theorem~\ref{thm:lower_bound}. {\sf (iv)} The upper bound on the RDP of the shuffled model given in~\cite[Remark~$1$]{erlingsson2019amplification}.}
\label{fig:RDP_comp}
\end{figure}


\paragraph{RDP of the shuffled model:} In Figure~\ref{fig:RDP_comp}, we plot several bounds on the RDP of the
shuffled model in different regimes. In particular, we compare between
the first upper bound on the RDP given in
Theorem~\ref{thm:general_case}, the second upper bound on the RDP
given in Theorem~\ref{thm:general_case_2}, the lower bound on the RDP
given in Theorem~\ref{thm:lower_bound}, and the upper bound on the RDP
given in~\cite[Remark~$1$]{erlingsson2019amplification} and stated
in~\eqref{eqn:rdp_erlin-1}.\footnote{The results in~\cite{feldman2020hiding} are for approximate DP (not for RDP), that is why we did not compare with them in Figure~\ref{fig:RDP_comp}.} 
It is shown that our first upper
bound~\eqref{eqn:1st_bound} gives a tighter bound on the RDP in
comparison with the second bound~\eqref{eqn:2nd_bound} and the upper
bound given in~\cite{erlingsson2019amplification}. Furthermore, the
first upper bound is close to the lower bound for small values of the
LDP parameter $\epsilon_0$ and for high orders $\lambda$. In addition,
the gap between our proposed bound in Theorem~\ref{thm:general_case}
and the bound given in~\cite{erlingsson2019amplification} increases as
the LDP parameter $\epsilon_0$ increases. We also observe that the
curves of the lower and upper bounds on the RDP of the shuffled model
saturate close to $\epsilon_0$ when the order $\lambda$ approaches to
infinity.  This indicates that the pure DP of the shuffled model is
bounded below by $\epsilon_0$, an observation made in literature
\cite{erlingsson2019amplification,BalcerC20}. As can be seen
in Figures \ref{fig:4} and \ref{fig:5}, the RDP obtained by
standard approximate DP to RDP conversion in
~\cite[Remark~$1$]{erlingsson2019amplification}, can be several orders
of magnitude loose in comparison to our analysis.

%

\begin{figure}[t]
    \centering 
\begin{subfigure}{0.31\textwidth}
  \includegraphics[scale=0.39]{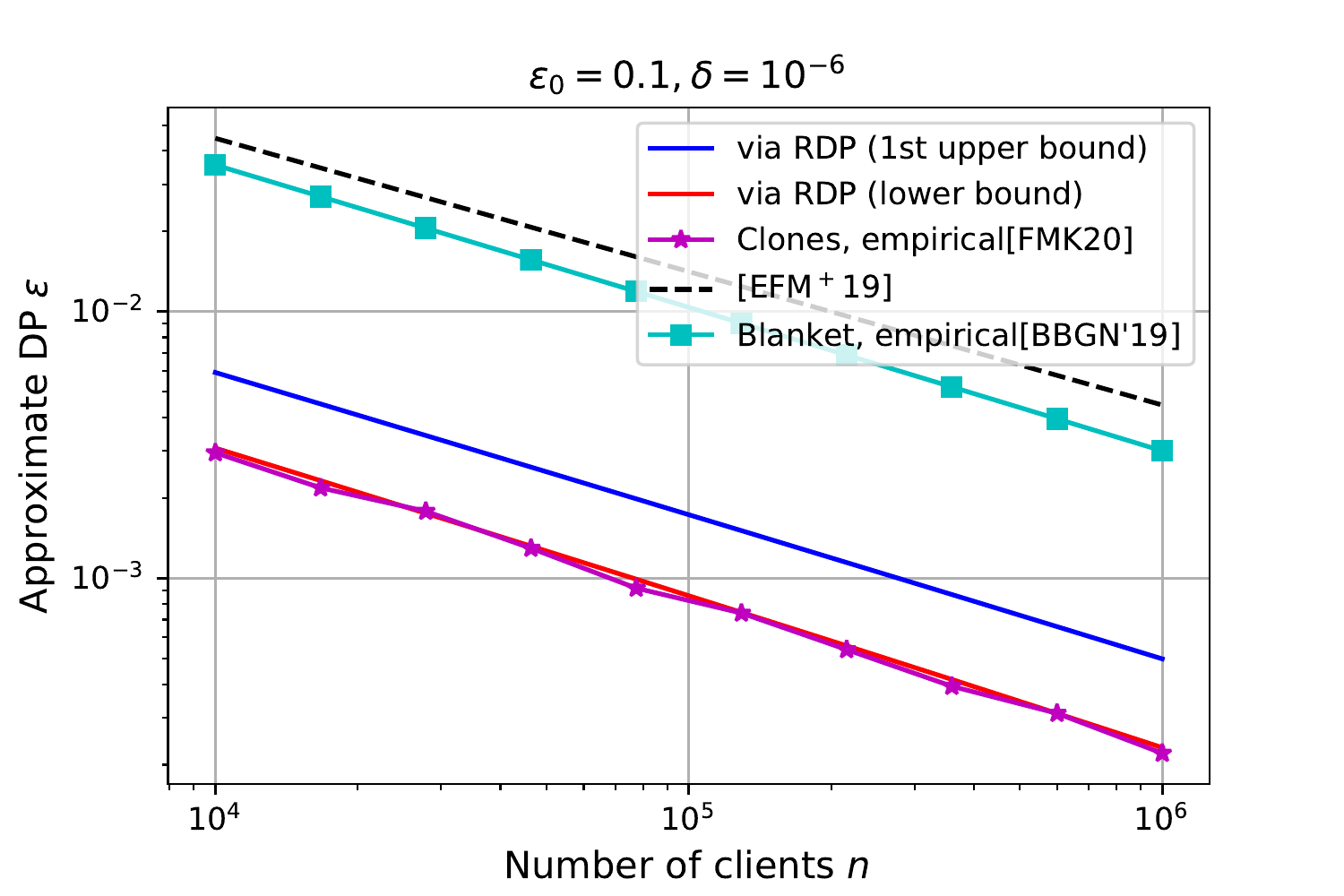}
  \caption{Approximate DP as a function of $n$ for $\epsilon_0=0.1$ and $\delta=10^{-6}$}
  \label{fig:DP_1}
\end{subfigure}\hfil 
\begin{subfigure}{0.31\textwidth}
  \includegraphics[scale=0.39]{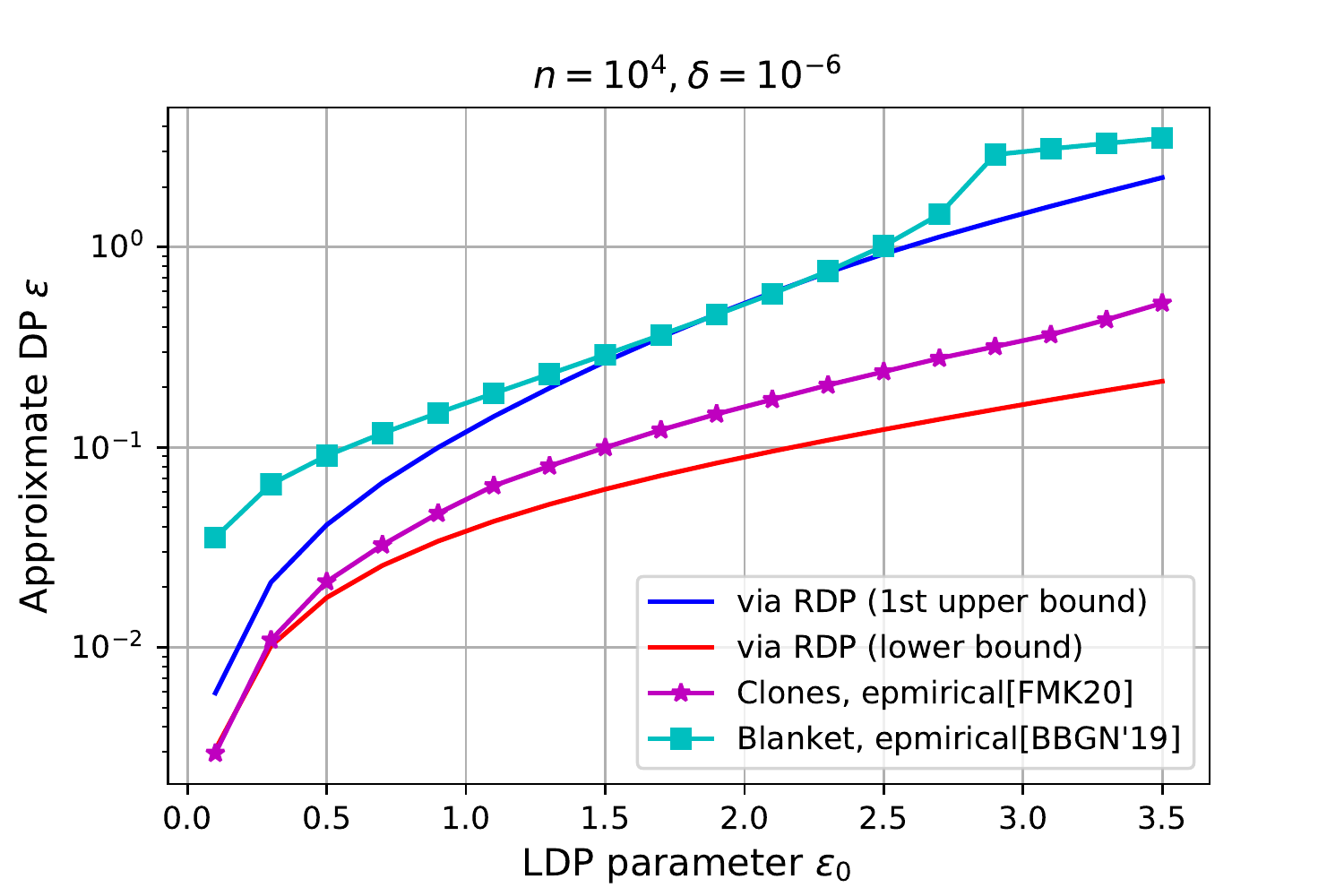}
  \caption{Approximate DP as a function of $\epsilon_0$ for $n=10^4$ and $\delta=10^{-6}$}
  \label{fig:DP_2}
\end{subfigure}

\medskip
\begin{subfigure}{0.31\textwidth}
  \includegraphics[scale=0.39]{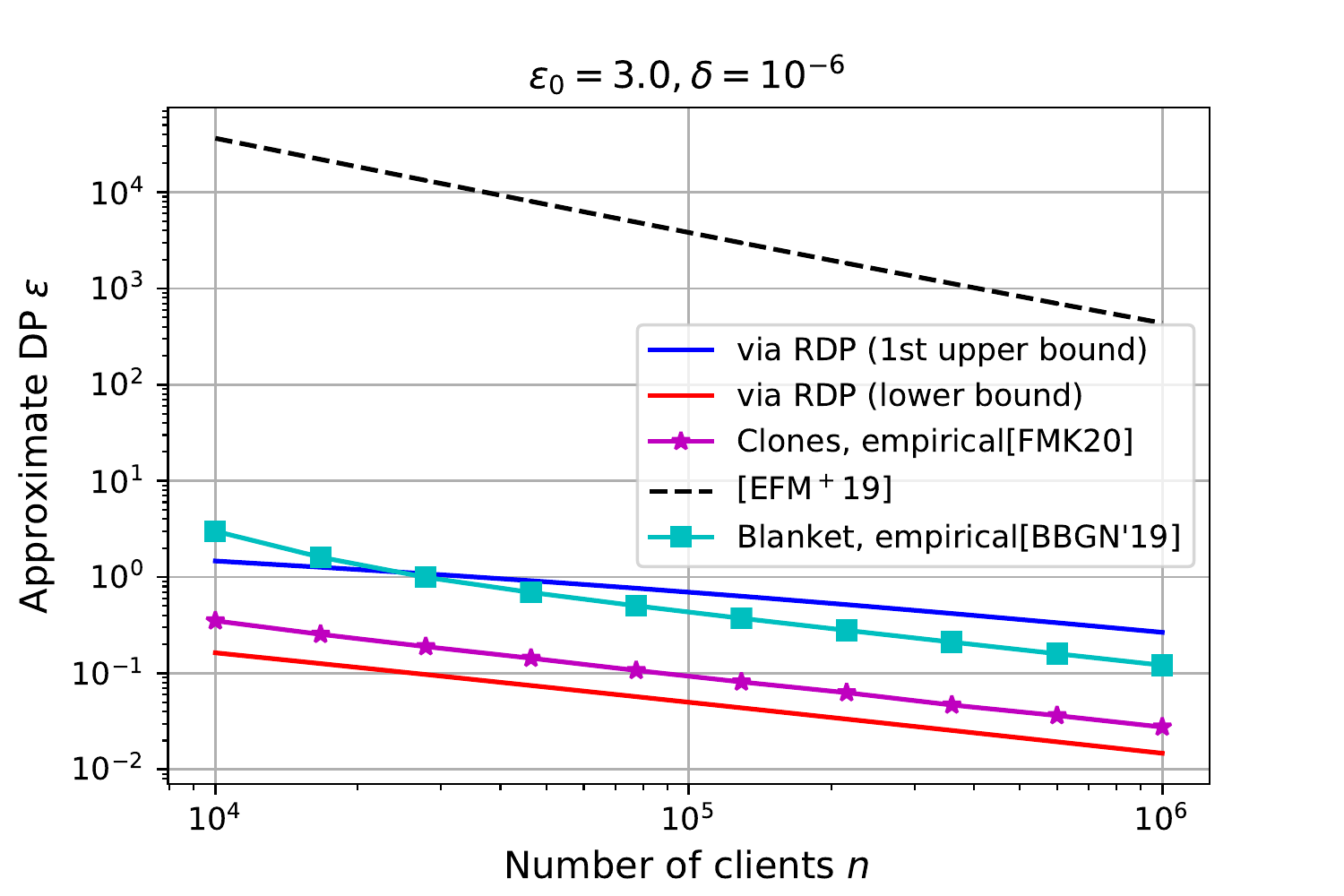}
  \caption{Approximate DP as a function of $n$ for $\epsilon_0=3$ and $\delta=10^{-6}$}
  \label{fig:DP_3}
\end{subfigure}\hfil 
\begin{subfigure}{0.31\textwidth}
  \includegraphics[scale=0.39]{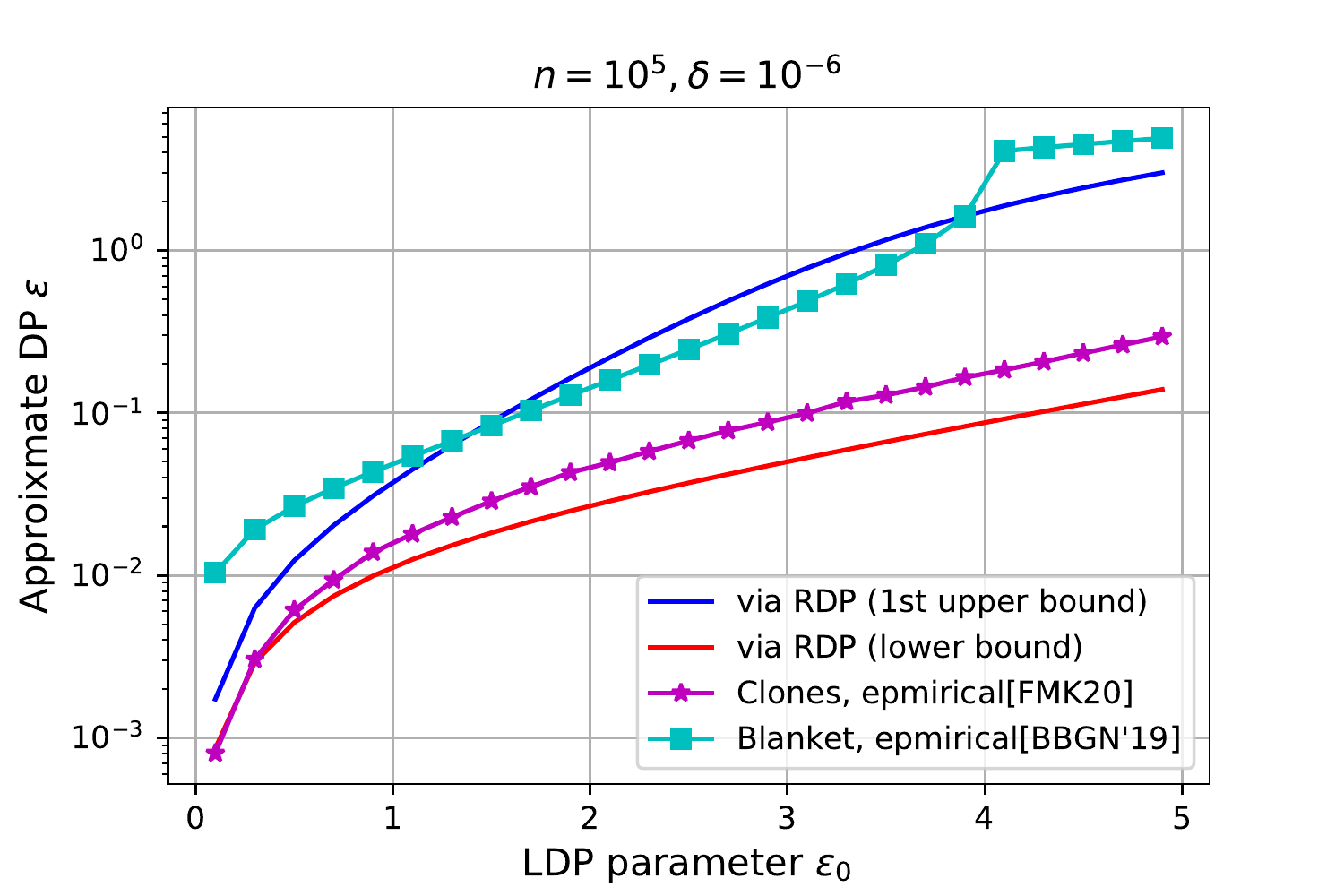}
  \caption{Approximate DP as a function of $\epsilon_0$ for $n=10^5$ and $\delta=10^{-6}$}
  \label{fig:DP_4}
\end{subfigure}
\caption{Comparison of several bounds on the Approximate $\left(\epsilon,\delta\right)$-DP of the shuffled model for $\delta=10^{-6}$: {\sf (i)} Approximate DP obtained from our first upper bound~\eqref{eqn:1st_bound} of the RDP in Theorem~\ref{thm:general_case}. {\sf (ii)} Approximate DP obtained from our lower bound on the RDP proposed in Theorem~\ref{thm:lower_bound}. {\sf (iii)} The empirical bound on the approximate DP given in~\cite{feldman2020hiding}. {\sf (iv)} The theoretical bound on the approximate DP given in~\cite{erlingsson2019amplification}. {\sf (v)} The generic bound on the approximate DP given in~\cite{balle2019privacy}.} 
\label{fig:Approximate_DP_comp}
\end{figure}
    
\paragraph{Approximate DP of the shuffled model:} Analyzing RDP of the
shuffled model provides a bound on the approximate DP of the shuffled
model from the relation between the RDP and approximate DP as shown in
Lemma~\ref{lem:RDP_DP}. In Figure~\ref{fig:Approximate_DP_comp}, we
plot several bounds on the approximate
$\left(\epsilon,\delta\right)$-DP of the shuffled model for fixed
$\delta=10^{-6}$. In Figures~\ref{fig:DP_4} and~\ref{fig:DP_2}, we do
not plot the results given in~\cite{erlingsson2019amplification},
since their bounds are quite loose and are far from the plotted range when
$\epsilon_0>1$. We can see that our analysis of the RDP of the
shuffled model provides a better bound on the approximate DP of the
shuffled model for small values of LDP parameter
$\epsilon_0<1$. However, our RDP analysis performs worse than the best
known bound given in~\cite{feldman2020hiding} when LDP parameter is
large $\epsilon_0>1.5$. This might be due to the gap between our upper
and lower bound on the RDP of the shuffled model as the lower bound
provides better performance than the bound given
in~\cite{feldman2020hiding} for all values of LDP parameter
$\epsilon_0$. Note that the main use case for converting our RDP
analysis to approximate DP is after composition rather than in the
single-shot conversion illustrated in Figure
\ref{fig:Approximate_DP_comp}.

\begin{figure}[htb]
    \centering 
\begin{subfigure}{0.31\textwidth}
  \includegraphics[scale=0.39]{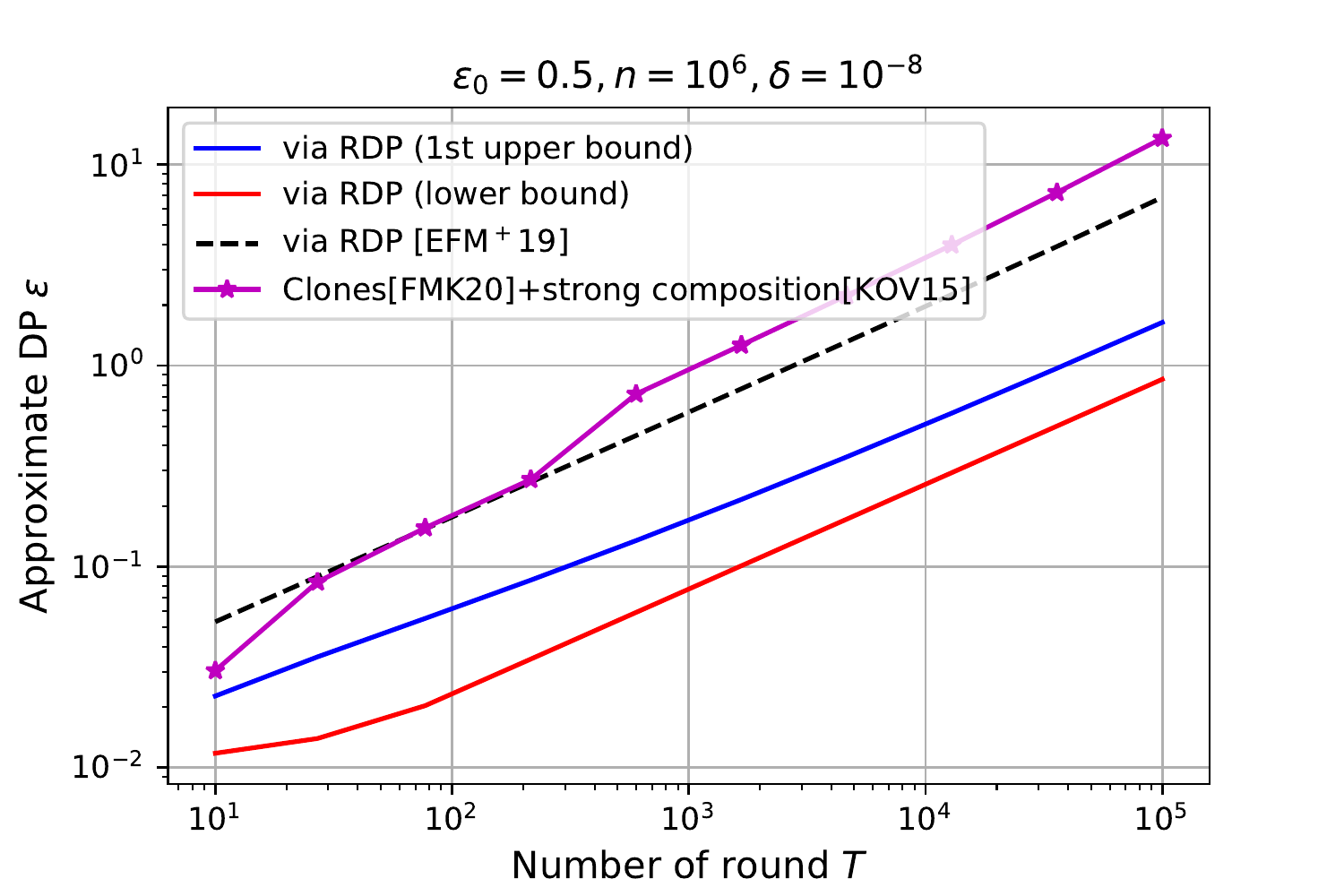}
  \caption{Approximate DP as a function of $T$ for $\epsilon_0=0.5$ and $n=10^{6}$ }
  \label{fig:Dp_com_1}
\end{subfigure}\hfil 
\begin{subfigure}{0.31\textwidth}
  \includegraphics[scale=0.39]{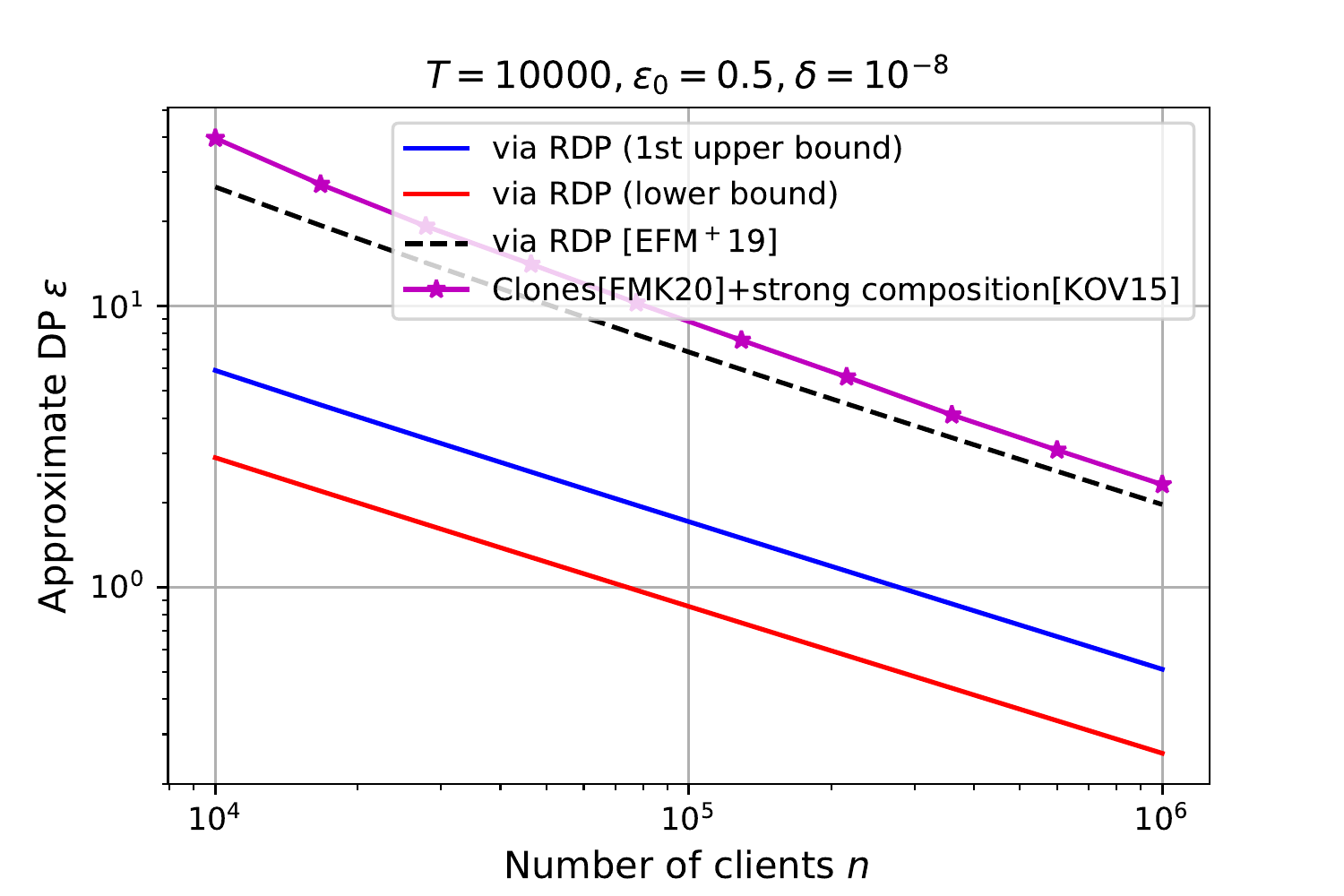}
  \caption{Approximate DP as a function of $n$ for $\epsilon_0=0.5$ and $T=10^{4}$}
  \label{fig:Dp_com_2}
\end{subfigure}\hfil 
\begin{subfigure}{0.31\textwidth}
  \includegraphics[scale=0.39]{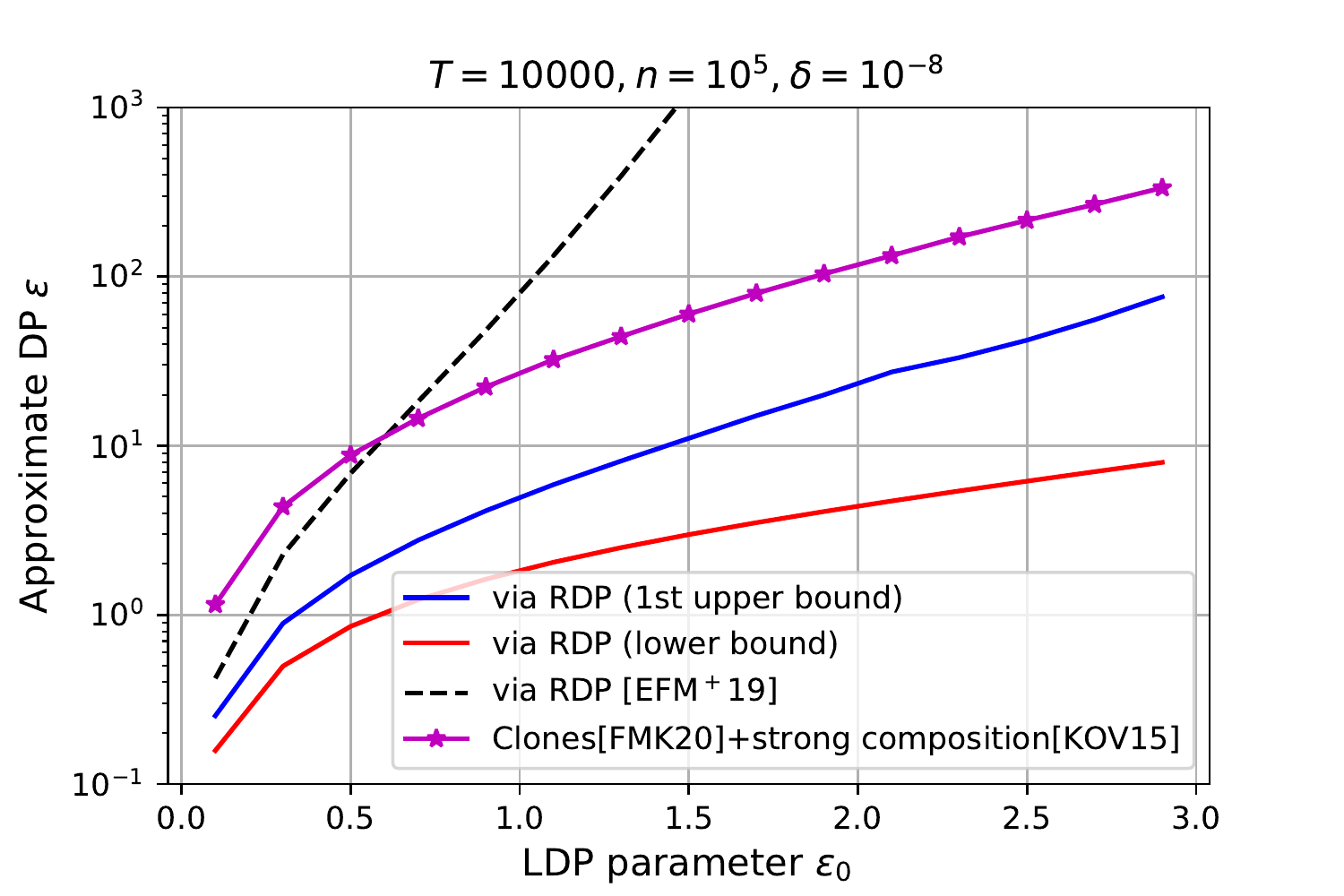}
  \caption{Approximate DP as a function of $\epsilon_0$ for $n=10^4$ and $T=10^{4}$}
  \label{fig:Dp_com_3}
\end{subfigure}

\medskip
\begin{subfigure}{0.31\textwidth}
  \includegraphics[scale=0.39]{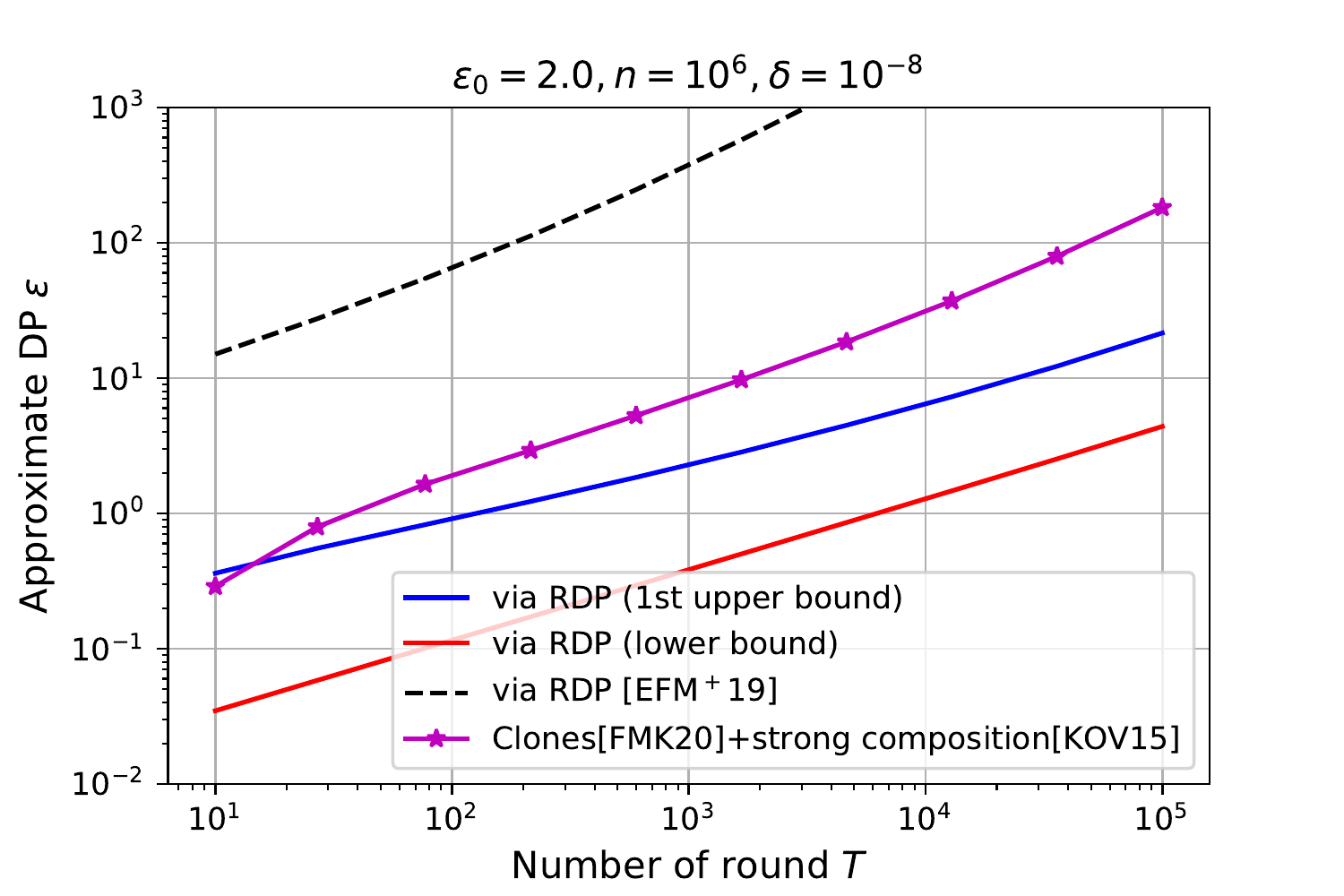}
  \caption{Approximate DP as a function of $T$ for $\epsilon_0=2$ and $n=10^6$}
  \label{fig:Dp_com_4}
\end{subfigure}\hfil 
\begin{subfigure}{0.31\textwidth}
  \includegraphics[scale=0.39]{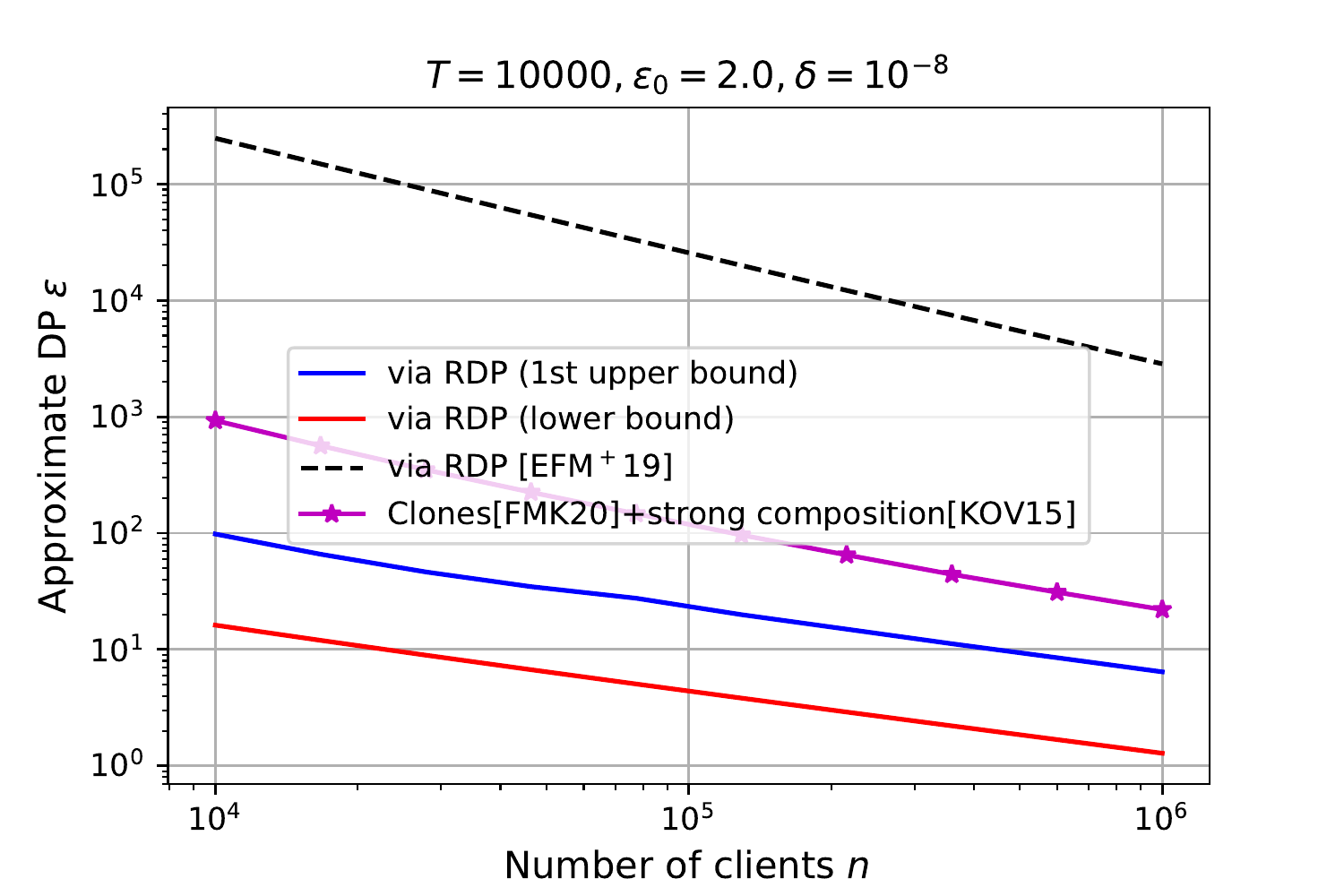}
  \caption{Approximate DP as a function of $n$ for $\epsilon_0=2$ and $T=10^{4}$}
  \label{fig:Dp_com_5}
\end{subfigure}\hfil 
\begin{subfigure}{0.31\textwidth}
  \includegraphics[scale=0.39]{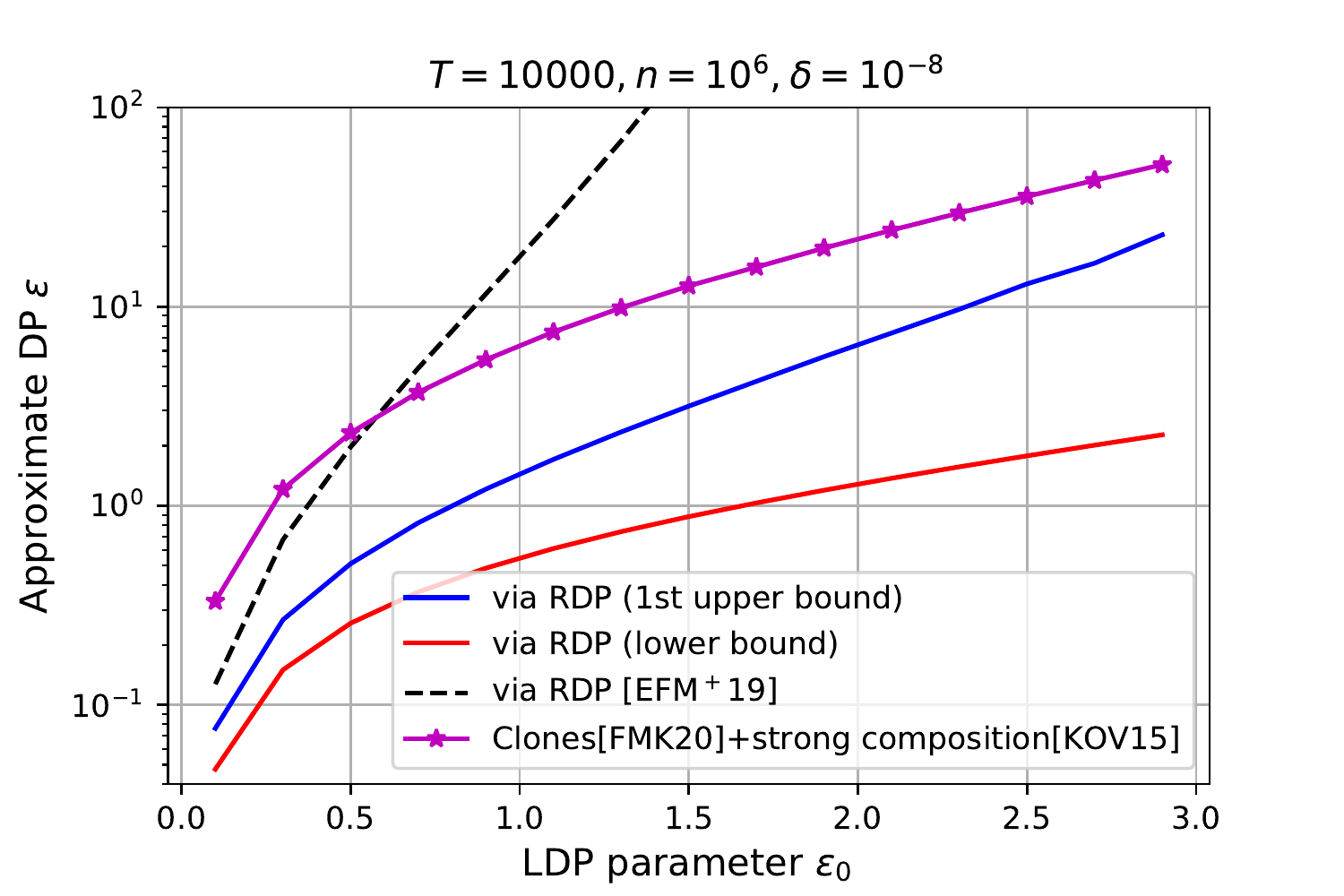}
  \caption{Approximate DP as a function of $\epsilon_0$ for $n=10^6$ and $T=10^4$}
  \label{fig:Dp_com_6}
\end{subfigure}
\caption{Comparison of several bounds on the Approximate $\left(\epsilon,\delta\right)$-DP for composition a sequence of shuffled models for $\delta=10^{-6}$: {\sf (i)} Approximate DP obtained from our first upper bound~\eqref{eqn:1st_bound} of the RDP in Theorem~\ref{thm:general_case}. {\sf (ii)} Approximate DP obtained from our lower bound on the RDP proposed in Theorem~\ref{thm:lower_bound}. {\sf (iii)} Approximate DP obtained from the upper bound on the RDP given in~\cite{erlingsson2019amplification}. {\sf (iv)} Applying the strong composition theorem~\cite{kairouz2015composition} after getting the approximate DP of the shuffled model given in~\cite{feldman2020hiding}.}
\label{fig:DP_composition}
\end{figure}
 
\paragraph{Composition of a sequence of shuffled models:} We now numerically evaluate the privacy parameters of the approximate $\left(\epsilon,\delta\right)$-DP for a composition of $T$ mechanisms $\left(\calM_1,\ldots,\calM_T\right)$, where $\calM_t$ is a shuffled mechanism for all $t\in \left[T\right]$. In Figure~\ref{fig:DP_composition}, we plot three different bounds on the overall privacy parameter $\epsilon$ for fixed $\delta=10^{-6}$ for a composition of $T$ identical shuffled models. The first bound on the overall privacy parameter $\epsilon$ is obtained as a function of $\delta$ and the number of iterations $T$ by optimizing over the RDP order $\lambda$ using our upper bound on the RDP of the shuffled model given in Theorem~\ref{thm:general_case}. The second bound is obtained by optimizing over the RDP order $\lambda$ using the upper bound on the RDP of the shuffled model given in~\cite{erlingsson2019amplification}. The third bound is obtained by first computing the privacy parameters $(\tilde{\epsilon},\tilde{\delta})$ of the shuffled model given in~\cite{feldman2020hiding}. Then, we use the strong composition theorem given in~\cite{kairouz2015composition} to obtain the overall privacy loss $\epsilon$. We observe that there is a significant saving in the overall privacy parameter $\epsilon$-DP using our bound on RDP in comparison with using the bound on DP~\cite{feldman2020hiding} with strong composition theorem~\cite{kairouz2015composition}. For example, we save a factor of $8\times$ in computing the overall privacy parameter $\epsilon$ for number of iterations $T=10^{5}$, LDP parameter $\epsilon_0=0.5$, and number of clients $n=10^{6}$. We observe that the bound given in~\cite{feldman2020hiding} with the strong composition theorem~\cite{kairouz2015composition} behaves better for small number of iterations $T<10$ and large LDP parameter $\epsilon_0=2$. However, the typical number of iterations $T$ in the standard SGD algorithm is usually larger. Therefore, this  demonstrate the significance of our RDP analysis for composition in important regimes of interest.  

\begin{figure}[htb]
    \centering 
\begin{subfigure}{0.31\textwidth}
  \includegraphics[scale=0.39]{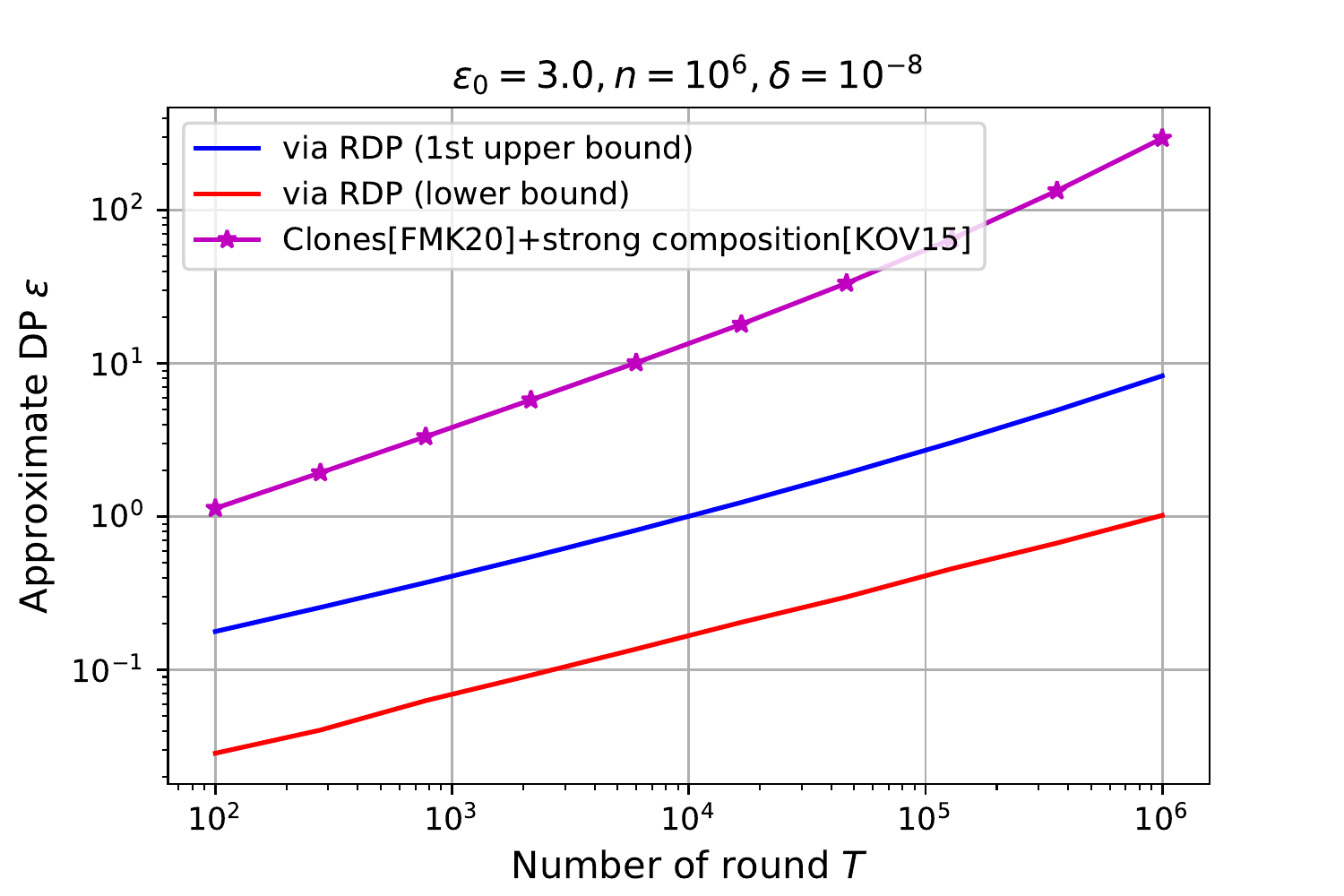}
  \caption{Approximate DP as a function of $T$ for $\epsilon_0=3$, $\gamma=0.001$ and $n=10^{6}$.}
  \label{fig:Dp_com_sampling_1}
\end{subfigure}\hfil 
\begin{subfigure}{0.31\textwidth}
  \includegraphics[scale=0.39]{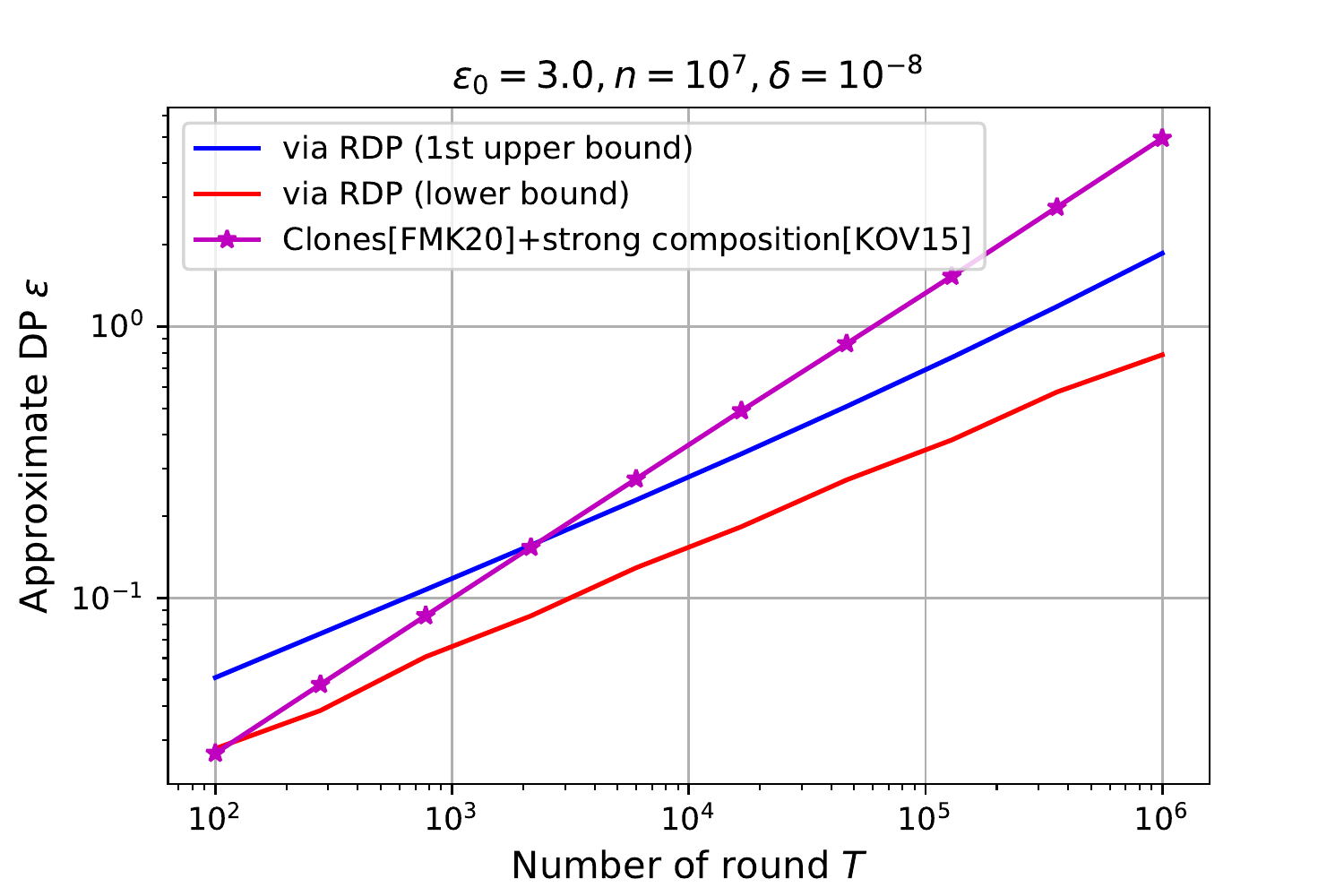}
  \caption{Approximate DP as a function of $n$ for $\epsilon_0=3$, $\gamma=0.001$ and $n=10^{7}$.}
  \label{fig:Dp_com_sampling_2}
\end{subfigure}
\caption{Comparison of several bounds on the Approximate $\left(\epsilon,\delta\right)$-DP for composition a sequence of shuffled models with Poisson sub-sampling for $\delta=10^{-8}$ and $\gamma=0.001$: {\sf (i)} Approximate DP obtained from our first upper bound~\eqref{eqn:1st_bound} on the RDP in Theorem~\ref{thm:general_case} with RDP amplification by Poisson sub-sampling~\cite{zhu2019poission}. {\sf (ii)} Approximate DP obtained from our lower bound on the RDP proposed in Theorem~\ref{thm:lower_bound} with RDP amplification by Poisson sub-sampling~\cite{zhu2019poission}. {\sf (iii)} Applying the strong composition theorem~\cite{kairouz2015composition} after getting the approximate DP of the shuffled model given in~\cite{feldman2020hiding} with Poisson sub-sampling~\cite{li2012sampling}.}
\label{fig:DP_composition_sampling}
\end{figure}

\paragraph{Privacy amplification by shuffling and Poisson sub-sampling:}
In the Differentially Private Stochastic Gradient Descent (DP-SGD),
shuffling and sampling the dataset at each iteration are important
tools to provide strong privacy
guarantee~\cite{girgis2021shuffled-jsait,ESA}. In these frameworks,
the further advantage of sampling with shuffling\footnote{In this framework we assume that the sampling and shuffling is done by a secure mechanism which is separated from the server, \emph{i.e.,} the server does not know which clients are participating.}  can be analyzed by
standard combination of approximate DP with Poisson subsampling
~\cite{li2012sampling}. The resulting approximate DP along with strong
composition theorem given in~\cite{kairouz2015composition} gives the
overall privacy loss $\epsilon$. An alternate path we use is to
combine our RDP analysis with sampling of RDP mechanisms using
~\cite{wang2019subsampled,zhu2019poission}. This enables us to get an
RDP guarantee with sampling, which we can then compose using
properties of RDP.  We can use the conversion from RDP to approximate
DP to obtain a bound on the overall privacy loss of multiple
iterations. In Figure~\ref{fig:DP_composition_sampling}, we compare
our results of amplifying the RDP of the shuffled model by Poisson
sub-sampling to the strong composition~\cite{kairouz2015composition}
after getting the approximate DP of the shuffled model given
in~\cite{feldman2020hiding} with Poisson sub-sampling given
in~\cite{li2012sampling}. We observe that we save a factor of
$11\times$ by using our RDP bound for $n=10^6$ and
$\gamma=0.001$. However, we can see that the gap between our
(lower/upper) bounds and the strong composition decreases when
$n=10^7$. This could be due to the simplistic combination of our
analysis with the RDP subsampling of \cite{zhu2019poission}.

\section{Proof of Theorem~\ref{Thm:reduce_special_case} -- Reduction to the Special Case}\label{sec:proof_reduce_special_case}




Recall that the LDP mechanism $\calR:\calX\to\calY$ has a discrete range $\calY=\left[B\right]$ for some $B\in\bbN$. 
Let $\bp_i:=(p_{i1},\ldots,p_{iB})$ and $\bp'_n:=(p'_{n1},\ldots,p'_{nB})$ denote the probability distributions over $\calY$ when the input to $\calR$ is $d_i$ and $d'_n$, respectively, where $p_{ij}=\Pr[\calR(d_i)=j]$ and $p_{nj}'=\Pr[\calR(d'_n)=j]$ for all $j\in[B]$ and $i\in\left[n\right]$. 
Let $\calP=\lbrace \bp_i:i\in\left[n\right] \rbrace$ and $\calP'=\lbrace \bp_i:i\in\left[n-1\right] \rbrace\bigcup \lbrace \bp'_n\rbrace$. 
For $i\in[n-1]$, let $\calP_{-i}=\calP\setminus\{\bp_i\}$, $\calP'_{-i}=\calP'\setminus\{\bp_i\}$, and also 
$\calP_{-n}=\calP\setminus\{\bp_n\}$, $\calP'_{-n}=\calP'\setminus\{\bp_n'\}$.
Here, $\calP,\calP'$ correspond to the datasets $\calD=\{d_1,\hdots,d_n\},\calD'=\{d_1,\hdots,d_{n-1},d'_n\}$, respectively, and for any $i\in[n]$, $\calP_{-i}$ and $\calP'_{-i}$ correspond to the datasets $\calD_{-i}=\{d_1,\hdots,d_{i-1},d_{i+1},\hdots,d_n\}$ and $\calD'_{-i}=\{d_1,\hdots,d_{i-1},d_{i+1},\hdots,d_{n-1},d'_n\}$, respectively. 
%
%

For any collection $\calP=\{\bp_1,\hdots,\bp_n\}$ of $n$ distributions, we define $F(\calP)$ to be the distribution over $\calA_B^n$ (which is the set of histograms on $B$ bins with $n$ elements as defined in \eqref{histogram-set}) that is induced when every client $i$ (independent to the other clients) samples an element from $[B]$ accordingly to the probability distribution $\bp_i$. 
Formally, for any $\bh\in\calA_B^n$, define
\begin{align}\label{mapping-possibilities}
\calU_{\bh} := \left\{ (\calU_1,\hdots,\calU_B): \calU_1,\hdots,\calU_B\subseteq[n] \text{ s.t. } \bigcup_{j=1}^B\calU_j=[n] \text{ and } |\calU_j|=h_j,\forall j\in[B] \right\}.
\end{align}
Note that for each $(\calU_1,\hdots,\calU_B)\in\calU_{\bh}$, $\calU_j$ for $j=1,\hdots,B$ denotes the identities of the clients that map to the $j$'th element in $[B]$ -- here $\calU_j$'s are disjoint for all $j\in[B]$. Note also that $|\calU_{\bh}|=\binom{n}{\bh}=\frac{n!}{h_1!h_2!\hdots h_B!}$. It is easy to verify that for any $\bh\in\calA_B^n$, $F(\calP)(\bh)$ is equal to
\begin{align}\label{general-distribution}
F(\calP)(\bh) = \sum_{(\calU_1,\hdots,\calU_B)\in\calU_{\bh}}\prod_{j=1}^B\prod_{i\in\calU_j}p_{ij}
\end{align}
Similarly, we can define $F(\calP'),F(\calP_{-i}),F(\calP'_{-i})$. 
Note that $F(\calP)$ and $F(\calP')$ are distributions over $\calA_B^n$, whereas, $F(\calP_{-i})$ and $F(\calP'_{-i})$ are distributions over $\calA_B^{n-1}$.
It is easy to see that $F(\calP)=\calM(\calD)$ and $F(\calP')=\calM(\calD')$. Similarly, $F(\calP_{-i})=\calM(\calD_{-i})$ and $F(\calP'_{-i})=\calM(\calD'_{-i})$. 

Now we are ready to prove Theorem~\ref{Thm:reduce_special_case}.

Since $\calR$ is an $\eps_0$-LDP mechanism, we have $$e^{-\eps_0}\leq \frac{p_{ij}}{p'_{nj}}\leq e^{\eps_0},\qquad \forall j\in[B]\ ,i\in[n].$$

As mentioned in Section~\ref{subsubsec:proof-sketch_reduce-special-case}, a crucial observation is that any distribution $\bp_i$ can be written as the following mixture distribution:
\begin{equation}\label{eq:mixture-dist}
\bp_i= q \bp'_n+\left(1-q\right)\tilde{\bp}_i,
\end{equation}
where $q=\frac{1}{e^{\epsilon_0}}$. The distribution $\tilde{\bp}_i=\left[\tilde{p}_{i1},\ldots,\tilde{p}_{iB}\right]$ is given by $\tilde{p}_{ij}=\frac{p_{ij}-q p'_{nj}}{1-q}$, where it is easy to verify that $\tilde{p}_{ij}\geq 0$ and $\sum_{j=1}^{B}\tilde{p}_{ij}=1$.

Now we show that since each $\bp_i=q \bp'_n+(1-q) \tilde{\bp}_i$ is a mixture distribution, we can write $F(\calP)$ and $F(\calP')$ as certain convex combinations. Before stating the result, we need some notation.

For any $\calC\subseteq[n-1]$, define two sets $\calP_{\calC},\calP'_{\calC}$, having $n$ distributions each, as follows:
\begin{align}
\calP_{\calC} &= \{\hat{\bp}_1,\hdots,\hat{\bp}_{n-1}\}\bigcup\{\bp_n\}, \label{eq:defn_P_C} \\
\calP'_{\calC} &= \{\hat{\bp}_1,\hdots,\hat{\bp}_{n-1}\}\bigcup\{\bp'_n\}, \label{eq:defn_P_C-prime}
\end{align}
where, for every $i\in[n-1]$, $\hat{\bp}_i$ is defined as follows:
\begin{equation}\label{eq:defn_hatP}
\hat{\bp}_i=
\begin{cases}
\bp_n' & \text{ if } i\in\calC, \\
\tilde{\bp}_i & \text{ if } i\in[n-1]\setminus\calC.
\end{cases}
\end{equation}
Note that $\calP_{\calC}$ and $\calP'_{\calC}$ differ only in one distribution -- $\calP_{\calC}$ contains $\bp_n$ whereas $\calP'_{\calC}$ contains $\bp'_n$. In words, if clients map their data points according to the distributions in either $\calP_{\calC}$ or $\calP'_{\calC}$ for any $\calC\subseteq[n-1]$, then for all clients $i\in\calC$, the $i$'th client maps its data point according to $\bp'_n$ (which is the distribution of $\calR$ on input $d_n'$), and for all clients $i\in[n-1]\setminus\calC$, the $i$'th client maps its data point according to $\tilde{\bp}_i$. The last client maps its data point according to $\bp_n$ or $\bp_n'$ depending on whether the set is $\calP_{\calC}$ or $\calP'_{\calC}$.


In the following lemma, we show that $F(\calP)$ and $F(\calP')$ can be written as convex combinations of $\{F(\calP_{\calC}):\calC\subseteq[n-1]\}$ and $\{F(\calP'_{\calC}):\calC\subseteq[n-1]\}$, respectively, where for any $\calC\subseteq[n-1]$, both $F(\calP_{\calC})$ and $F(\calP'_{\calC})$ can be computed analogously as in \eqref{general-distribution}.

\begin{lemma}[Mixture Interpretation]\label{lem:convex-combinations}
$F(\calP)$ and $F(\calP')$ can be written as the following convex combinations:
\begin{align}
F(\calP)&=\sum_{\calC\subseteq [n-1]} q^{|\calC|}(1-q)^{n-|\calC|-1}F(\calP_{\calC}), \label{P_mixture} \\
F(\calP')&=\sum_{\calC\subseteq [n-1]} q^{|\calC|}(1-q)^{n-|\calC|-1}F(\calP'_{\calC}), \label{P-prime_mixture}
\end{align}
where $\calP_{\calC},\calP_{\calC}'$ are defined in \eqref{eq:defn_P_C}-\eqref{eq:defn_hatP}.
\end{lemma}
We prove Lemma~\ref{lem:convex-combinations} in Appendix~\ref{app:proof_convex-combinations}.

Now, using Lemma~\ref{lem:convex-combinations}, in the following lemma we show that the Renyi divergence between $F(\calP)$ and $F(\calP')$ can be upper-bounded by a convex combination of the Renyi divergence between $F(\calP_{\calC})$ and $F(\calP'_{\calC})$ for $\calC\subseteq[n-1]$. 
\begin{lemma}[Joint Convexity]\label{lem:joint-con_renyi-exp}
For any $\lambda>1$, the function $\mathbb{E}_{\bh\sim F\left(\calP'\right)}\left[\left(\frac{F\left(\calP\right)\left(\bh\right)}{F\left(\calP'\right)\left(\bh\right)}\right)^{\lambda}\right]$ is jointly convex in $(F(\calP),F(\calP'))$, i.e.,
\begin{align}
\mathbb{E}_{\bh\sim F\left(\calP'\right)}\left[\left(\frac{F\left(\calP\right)\left(\bh\right)}{F\left(\calP'\right)\left(\bh\right)}\right)^{\lambda}\right]
&\leq \sum_{\calC\subseteq \left[n-1\right]} q^{|\calC|}\left(1-q\right)^{n-|\calC|-1} \mathbb{E}_{\bh\sim F\left(\calP'_{\calC}\right)}\left[\left(\frac{F\left(\calP_{\calC}\right)\left(\bh\right)}{F\left(\calP'_{\calC}\right)\left(\bh\right)}\right)^{\lambda}\right]. \label{eqn:rdp_bound}
\end{align}
\end{lemma}
We prove Lemma~\ref{lem:joint-con_renyi-exp} in Appendix~\ref{app:proof_joint-con_renyi-exp}.

For any $\calC\subseteq[n-1]$, let $\widetilde{\calP}_{[n-1]\setminus\calC}=\{\tilde{\bp}_i:i\in\left[n-1\right]\setminus \calC\}$. With this notation, note that $\calP_{\calC}\setminus\widetilde{\calP}_{[n-1]\setminus\calC}=\{\bp'_n,\hdots,\bp'_n\} \bigcup \{\bp_n\}$ and $\calP'_{\calC}\setminus\widetilde{\calP}_{[n-1]\setminus\calC}=\{\bp'_n,\hdots,\bp'_n\} \bigcup \{\bp'_n\}$ is a pair of specific neighboring distributions, each containing $|\calC|+1$ distributions. In other words, if we define $\calD_{|\calC|+1}^{(n)}=\left(d'_n,\ldots,d'_n,d_{n}\right)$ and $\calD_{|\calC|+1}'^{(n)}=\left(d'_n,\ldots,d'_n,d_n'\right)$, each having $(|\calC|+1)$ data points (note that $(\calD_{|\calC|+1}'^{(n)},\calD_{|\calC|+1}^{(n)})\in\calD_{\same}^{|\calC|+1}$), then the mechanisms $\calM(\calD_{|\calC|+1}^{(n)})$ and $\calM(\calD_{|\calC|+1}'^{(n)})$ will have distributions $F(\calP_{\calC}\setminus\widetilde{\calP}_{[n-1]\setminus\calC})$ and $F(\calP'_{\calC}\setminus\widetilde{\calP}_{[n-1]\setminus\calC})$, respectively.

Now, since $(\calD_{|\calC|+1}'^{(n)},\calD_{|\calC|+1}^{(n)})\in\calD_{\same}^{|\calC|+1}$, if we remove the effect of distributions in $\widetilde{\calP}_{[n-1]\setminus\calC}$ in the RHS of \eqref{eqn:rdp_bound}, we would be able to bound the RHS of \eqref{eqn:rdp_bound} using the RDP for the special neighboring datasets in $\calD_{\same}^{|\calC|+1}$. This is precisely what we will do in the following lemma and the subsequent corollary, where we will eliminate the distributions in $\widetilde{\calP}_{[n-1]\setminus\calC}$ in the RHS \eqref{eqn:rdp_bound}.

The following lemma actually holds for arbitrary pairs $(\calP,\calP')$ of neighboring distributions $\calP=\{\bp_1,\hdots,\bp_n\}$ and $\calP=\{\bp_1,\hdots,\bp_{n-1},\bp'_n\}$, where we show that $\mathbb{E}_{\bh\sim F\left(\calP'\right)}\left[\left(\frac{F\left(\calP\right)\left(\bh\right)}{F\left(\calP'\right)\left(\bh\right)}\right)^{\lambda}\right]$ does not decrease when we eliminate a distribution $\bp_i$ (i.e., remove the data point $d_i$ from the datasets) for any $i\in[n-1]$. We need this general statement as it will be required in the proof of Theorem~\ref{thm:general_case} later.
\begin{lemma}[Monotonicity]\label{lem:cvx_rdp} 
For any $i\in\left[n-1\right]$, we have
\begin{equation}\label{eq:cvx_rdp} 
\mathbb{E}_{\bh\sim F\left(\calP'\right)}\left[\left(\frac{F\left(\calP\right)\left(\bh\right)}{F\left(\calP'\right)\left(\bh\right)}\right)^{\lambda}\right]\leq \mathbb{E}_{\bh\sim F\left(\calP'_{-i}\right)}\left[\left(\frac{F\left(\calP_{-i}\right)\left(\bh\right)}{F\left(\calP'_{-i}\right)\left(\bh\right)}\right)^{\lambda}\right],
\end{equation}
where, for $i\in[n-1]$, $\calP_{-i}=\calP\setminus\{\bp_i\}$ and $\calP'_{-i}=\calP'\setminus\{\bp_i\}$. 
Note that in the LHS of \eqref{eq:cvx_rdp}, $F(\calP),F(\calP')$ are distributions over $\calA_B^n$, whereas, in the RHS, $F(\calP_{-i}),F(\calP'_{-i})$ for any $i\in[n-1]$ are distributions over $\calA_B^{n-1}$.
\end{lemma}
We prove Lemma~\ref{lem:cvx_rdp} in Appendix~\ref{app:proof_cvx_rdp}.

Note that Lemma~\ref{lem:cvx_rdp} is a general statement that holds for arbitrary pairs $(\calP,\calP')$ of neighboring distributions. For our purpose, we apply Lemma~\ref{lem:cvx_rdp} with $(\calP_{\calC},\calP'_{\calC})$ for any $\calC\subseteq[n-1]$ and then eliminate the distributions in $\widetilde{\calP}_{[n-1]\setminus\calC}$ one by one. The result is stated in the following corollary.
\begin{corollary}\label{corol:cvx_rdp_repeat}
Consider any $m\in\{0,1,\hdots,n-1\}$.
Let $\calD_{m+1}^{(n)}=\left(d'_n,\ldots,d'_n,d_n\right)$ and $\calD_{m+1}'^{(n)}=\left(d'_n,\ldots,d'_n\right)$.
Then, for any $\calC\in\binom{[n-1]}{m}$ (i.e., $\calC\subseteq[n-1]$ such that $|\calC|=m$), we have
\begin{align}
\bbE_{\bh\sim F(\calP'_{\calC})}\left[\left(\frac{F(\calP_{\calC})(\bh)}{F(\calP'_{\calC})(\bh)}\right)^{\lambda}\right] &\leq \bbE_{\bh\sim \calM(\calD_{m+1}'^{(n)})}\left[\left(\frac{\calM(\calD_{m+1}^{(n)})(\bh)}{\calM(\calD_{m+1}'^{(n)})(\bh)}\right)^{\lambda}\right]. \label{eqn:bound_C}
\end{align}
\end{corollary}
We prove Corollary~\ref{corol:cvx_rdp_repeat} in Appendix~\ref{app:proof_cvx_rdp_repeat}.

Substituting from~\eqref{eqn:bound_C} into~\eqref{eqn:rdp_bound} and noting that for every $\bh\in\calA_B^{n}$, $F(\calP)(\bh)$ and $F(\calP')(\bh)$ are distributionally equal to $\calM(\calD)(\bh)$ and $\calM(\calD')(\bh)$, respectively, we get
\begin{align*}
\mathbb{E}_{\bh\sim\calM\left(\calD'\right)}\left[\left(\frac{\calM\left(\calD\right)\left(\bh\right)}{\calM\left(\calD'\right)\left(\bh\right)}\right)^{\lambda}\right] &\stackrel{\text{(a)}}{\leq} \sum_{m=0}^{n-1}\sum_{\calC\in\binom{[n-1]}{m}} q^{m}\left(1-q\right)^{n-m-1} \mathbb{E}_{\bh\sim F\left(\calP'_{\calC}\right)}\left[\left(\frac{F\left(\calP_{\calC}\right)\left(\bh\right)}{F\left(\calP'_{\calC}\right)\left(\bh\right)}\right)^{\lambda}\right] \\
&\stackrel{\text{(b)}}{\leq} \sum_{m=0}^{n-1}\sum_{\calC\in\binom{[n-1]}{m}} q^{m}\left(1-q\right)^{n-m-1} \bbE_{\bh\sim \calM(\calD_{m+1}'^{(n)})}\left[\left(\frac{\calM(\calD_{m+1}^{(n)})(\bh)}{\calM(\calD_{m+1}'^{(n)})(\bh)}\right)^{\lambda}\right] \\
&\stackrel{\text{(c)}}{=} \sum_{m=0}^{n-1} \binom{n-1}{m} q^{m}\left(1-q\right)^{n-m-1} \bbE_{\bh\sim \calM(\calD_{m+1}'^{(n)})}\left[\left(\frac{\calM(\calD_{m+1}^{(n)})(\bh)}{\calM(\calD_{m+1}'^{(n)})(\bh)}\right)^{\lambda}\right] \\
&= \bbE_{m\sim \text{Bin}(n-1,q)}\left[ \bbE_{\bh\sim \calM(\calD_{m+1}'^{(n)})}\left[\left(\frac{\calM(\calD_{m+1}^{(n)})(\bh)}{\calM(\calD_{m+1}'^{(n)})(\bh)}\right)^{\lambda}\right]\right].
\end{align*}
The inequality (a) is the same as \eqref{eqn:rdp_bound} -- just writing it differently.
In (b) we used \eqref{eqn:bound_C} and in (c) we used the fact that number of $m$-sized subsets of $[n-1]$ is equal to $\binom{n-1}{m}$.

This completes the proof of Theorem~\ref{Thm:reduce_special_case}.

\section{Proof of Theorem~\ref{thm:RDP_same} -- RDP for the Special Form}\label{sec:special_form}

Fix an arbitrary $m\in\bbN$ and consider any pair of neighboring datasets $(\calD_{m},\calD'_{m})\in\calD_{\same}^{m}$. Let $\calD_m=(d,\hdots,d)\in\calX^m$ and $\calD'_m=(d,\hdots,d,d')\in\calX^m$. 
Let $\bp=(p_1,\ldots,p_B)$ and $\bp'=(p'_1,\ldots,p'_B)$ be the probability distributions of the discrete $\eps_0$-LDP mechanism $\calR:\calX\to\calY=[B]$ when its inputs are $d$ and $d'$, respectively, where $p_j=\Pr[\calR(d)=j]$ and $p_j'=\Pr[\calR(d')=j]$ for all $j\in[B]$. Since $\calR$ is $\eps_0$-LDP, we have 
\begin{equation}\label{LDP_constraints}
e^{-\eps_0}\leq \frac{p_j}{p'_j}\leq e^{\eps_0},\qquad \forall j\in[B].
\end{equation}
%


Since $\calM$ is a shuffled mechanism, it induces a distribution on $\calA_B^m$ for any input dataset. 
So, for any $\bh\in\calA_B^m$, $\calM(\calD_m)(\bh)$ and $\calM(\calD'_m)(\bh)$ are equal to the probabilities of seeing $\bh$ when the inputs to $\calM$ are $\calD_m$ and $\calD'_m$, respectively.
Thus, for a given histogram $\bh=(h_1,\hdots,h_B)\in\calA^m_B$ with $m$ elements and $B$ bins, we have
\begin{align}
\calM(\calD_m)\left(\bh\right)&=\textsl{MN}\left(m,\bp,\bh\right)={m \choose \bh}  \prod_{j=1}^{B} p_j^{h_j} \label{defn-mu_0},
\end{align}
where $\textsl{MN}\left(m,\bp,\bh\right)$ denotes the Multinomial distribution with ${m \choose \bh}=\frac{m!}{h_1!\cdots h_B!}$.
Note that \eqref{defn-mu_0} can be obtained as a special case of the general distribution in \eqref{general-distribution} by putting $\bp_j=\bp$ for each client $j$.

For $\calM(\calD'_m)$, note that the last client (independent of the other clients) maps its input data point $d'$ to the $j$'th bin with probability $p'_j$, and the remaining $(m-1)$ clients' mappings induce a distribution on $\calA_B^{m-1}$. Thus, $\calM(\calD'_m)(\bh)$ for any $\bh\in\calA_B^m$ can be written as
\begin{align}
\calM(\calD'_m)(\bh)&=\sum_{j=1}^{B} p'_j\textsl{MN}\left(m-1,\bp,\tbh_j\right), \label{defn-mu_1}
\end{align}
where 
$\tbh_j=\left(h_1,\ldots,h_{j-1},h_j-1,h_{j+1},\ldots,h_B\right)\in\calA_B^{m-1}$. 
Note that similar to \eqref{defn-mu_0}, \eqref{defn-mu_1} can also be obtained from \eqref{general-distribution} as a special case.

Using the polynomial expansion $(1+x)^n=\sum_{k=0}^{n}\binom{n}{k}x^k$, we have: 
\begin{align}
\mathbb{E}_{\bh\sim\calM(\calD_{m})}\left[\left(\frac{\calM(\calD'_{m})(\bh)}{\calM(\calD_{m})(\bh)}\right)^{\lambda}\right] 
&=\bbE_{\bh\sim\calM(\calD_m)}\left[\left(1+\frac{\calM(\calD'_m)(\bh)}{\calM(\calD_m)(\bh)}-1\right)^{\lambda}\right] \notag \\
&= \sum_{i=0}^{\lambda}\binom{\lambda}{i}\bbE_{\bh\sim\calM(\calD_m)}\left[\left(\frac{\calM(\calD'_m)(\bh)}{\calM(\calD_m)(\bh)}-1\right)^{i}\right]. \label{eqn:MomAcc_Taylor-exp}
\end{align}

Let $X:\calA_B^m\to\bbR$ be a random variable associated with the distribution $\calM(\calD_m)$ on $\calA_B^m$, and for any $\bh\in\calA_B^{m}$, define $X(\bh):=m\(\frac{\calM(\calD'_m)(\bh)}{\calM(\calD_m)(\bh)}-1\)$. 
Substituting this in \eqref{eqn:MomAcc_Taylor-exp} gives:
\begin{align}
\mathbb{E}_{\bh\sim\calM(\calD_{m})}\left[\left(\frac{\calM(\calD'_{m})(\bh)}{\calM(\calD_{m})(\bh)}\right)^{\lambda}\right] = 1 + \sum_{i=1}^{\lambda}\binom{\lambda}{i}\frac{1}{m^{i}} \bbE_{\bh\sim\calM(\calD_m)}\left[\left(X(\bh)\right)^{i}\right]. \label{eqn:MomAcc_Taylor-exp_X}
\end{align}
The RHS of \eqref{eqn:MomAcc_Taylor-exp_X} is in terms of the moments of $X$, which we bound in the following lemma. Before stating the lemma, first we simplify the expression for $X(\bh)$ by computing the ratio $\frac{\calM(\calD'_m)(\bh)}{\calM(\calD_m)(\bh)}$ for any $\bh\in\calA_B^m$:
\begin{equation}\label{ratio_mu}
\frac{\calM(\calD'_m)(\bh)}{\calM(\calD_m)(\bh)}
=\sum_{j=1}^{B}p'_j\frac{\textsl{MN}\,(m-1,\bp,\tbh_j)}{\textsl{MN}\,(m,\bp,\bh)}
=\sum_{j=1}^{B}\frac{p'_j}{p_j}\frac{h_j}{m}.
\end{equation}
Thus, we get 
$X(\bh)=m\(\frac{\calM(\calD'_m)(\bh)}{\calM(\calD_m)(\bh)} -1\) = \left(\sum_{j=1}^{B}\frac{p'_j}{p_j}h_j\right)-m$. 

\begin{remark}\label{rem:tightenBnd}  As mentioned in Remark \ref{rem:Pot-Impr},
we could tighten our upper bounds for specific mechanisms.  As shown in
\eqref{eqn:MomAcc_Taylor-exp_X} above, the Renyi divergence of a mechanism
between two neighboring datasets can be written in terms of the
moments of a r.v.\ $X$, which is defined as the ratio of distributions
of the mechanism on these two neighboring datasets. However, since our
goal is to bound RDP for all $\eps_0$-LDP mechanisms, we prove the
worse-case bound on the moments of $X$ that hold for all mechanisms;
see \eqref{ith-moment-Xh} in Lemma~\ref{lemm:MomAcc_RV} for bound on
the $i\geq3$'rd moments of $X$ and \eqref{eq:sup_lemma_bound} in
Lemma~\ref{lemma_bound_sup} for bound on the variance of $X$.

\end{remark}

\begin{lemma}\label{lemm:MomAcc_RV}
The random variable $X$ has the following properties:
\begin{enumerate}
\item $X$ has zero mean, i.e., $\mathbb{E}_{\bh\sim\calM(\calD_m)}\left[X(\bh)\right]=0$.
\item The variance of $X$ is equal to 
\begin{equation}\label{eqn:variance}
\mathbb{E}_{\bh\sim\calM(\calD_m)}\left[X(\bh)^{2}\right]=m\left(\sum_{j=1}^{B}\frac{p_j'^2}{p_j} -1\right).
\end{equation}
\item For $i\geq3$, the $i$'th moment of $X$ is bounded by 
\begin{equation}\label{ith-moment-Xh}
\mathbb{E}_{\bh\sim\calM(\calD_m)}\left[(X(\bh))^{i}\right] \leq \mathbb{E}_{\bh\sim\calM(\calD_m)}\left[|X(\bh)|^{i}\right]\leq i \Gamma\left(i/2\right)\left(2m\nu^2\right)^{i/2},
\end{equation}
where $\nu^2=\frac{\left(e^{\epsilon_0}-e^{-\epsilon_0}\right)^2}{4}$ and $\Gamma\left(z\right)=\int_{0}^{\infty}x^{z-1}e^{-x}dx$ is the Gamma function. 
\end{enumerate}
\end{lemma}
A proof of Lemma~\ref{lemm:MomAcc_RV} is presented in Appendix~\ref{app:supp_MomAcc_rv}. 

Substituting the bounds from Lemma~\ref{lemm:MomAcc_RV} into \eqref{eqn:MomAcc_Taylor-exp_X}, we get
\begin{align}
\mathbb{E}_{\bh\sim\calM\left(\calD_{m}\right)}\left[\left(\frac{\calM\left(\calD'_{m}\right)\left(\bh\right)}{\calM\left(\calD_{m}\right)\left(\bh\right)}\right)^{\lambda}\right] &\leq 1+\binom{\lambda}{2} \frac{1}{m}\left(\sum_{j=1}^{B}\frac{p_j'^2}{p_j} -1\right) + \sum_{i=3}^{\lambda} \binom{\lambda}{i} i\Gamma\left(i/2\right) \left(\frac{\left(e^{\epsilon_0}-e^{-\epsilon_0}\right)^2}{2m}\right)^{i/2} \label{special-proof_interim1}
\end{align}
Note that $p_1,\hdots,p_m,p_1',\hdots,p_m'$ are defined for the fixed pair of datasets $(\calD_m,\calD'_m)\in\calD_{\same}^m$ that we started with. So, the term containing $\left(\sum_{j=1}^{B}\frac{p_j'^2}{p_j} -1\right)$ in the RHS of \eqref{special-proof_interim1} depends on $(\calD_m,\calD'_m)$, and that is the only term in \eqref{special-proof_interim1} that depends on $(\calD_m,\calD'_m)$. Since Theorem~\ref{thm:RDP_same} requires us to bound \eqref{special-proof_interim1} for any pair of neighboring datasets $(\calD_m,\calD'_m)\in\calD_{\same}^m$, so, in order to prove Theorem~\ref{thm:RDP_same}, we need to compute $\sup_{(\calD_m,\calD'_m)\in\calD_{\same}^m}\left(\sum_{j=1}^{B}\frac{p_j'^2}{p_j} -1\right)$. We bound this in the following.

Define a set $\calT_{\eps_0}$ consisting of all pairs of $B$-dimensional probability vectors satisfying the $\eps_0$-LDP constraints as follows:
\begin{equation}\label{setT-pairs}
\calT_{\eps_0}=\left\lbrace(\bp,\bp')\in\bbR^B\times\bbR^B : p_j,p'_j\geq0, \forall j\in[B], \sum_{j=1}^Bp_j = \sum_{j=1}^Bp'_j = 1, \text{ and } e^{-\eps_0}\leq \frac{p'_j}{p_j}\leq e^{\eps_0}, \forall j\in[B]\right\rbrace
\end{equation}
Note that $\calT_{\eps_0}$ contains {\em all} pairs of the output probability distributions $(\bp,\bp')$ of {\em all} $\eps_0$-LDP mechanisms $\calR$ on {\em all} neighboring data points $d,d'\in\calX$.
Since any $(\calD_{m},\calD'_{m})\in\calD_{\same}^{m}$ generates a pair of probability distributions $(\bp,\bp')\in\calT_{\eps_0}$ (because $\calD_m=(d,\hdots,d)$ and $\calD_m'=(d,\hdots,d,d')$ together contain only two distinct data points $d,d'$), we have
\begin{align}\label{eq:variance-inequality}
\sup_{(\calD_{m},\calD'_{m})\in\calD_{\same}^{m}} \(\sum_{j=1}^{B}\frac{p_j'^2}{p_j}-1\) \leq \sup_{(\bp,\bp')\in\calT_{\eps_0}} \(\sum_{j=1}^{B}\frac{p_j'^2}{p_j}-1\).
\end{align}
In the following lemma, we bounds the RHS of \eqref{eq:variance-inequality}.
\begin{lemma}\label{lemma_bound_sup} 
We have the following bound:
\begin{equation}\label{eq:sup_lemma_bound}
\sup_{(\bp,\bp')\in\calT_{\eps_0}}\(\sum_{j=1}^{B}\frac{p_j'^2}{p_j}-1\) = \frac{\left(e^{\eps_0}-1\right)^2}{e^{\eps_0}}.
\end{equation}
\end{lemma} 
We prove Lemma~\ref{lemma_bound_sup} in Appendix~\ref{app:supp_MomAcc_bound_sup}.

Taking supremum over $(\calD_{m},\calD'_{m})\in\calD_{\same}^{m}$ in \eqref{special-proof_interim1} and then using \eqref{eq:variance-inequality} and \eqref{eq:sup_lemma_bound}, we get
\begin{equation*} 
\sup_{(\calD_{m},\calD'_{m})\in\calD_{\same}^{m}}\mathbb{E}_{\bh\sim\calM\left(\calD_{m}\right)}\left[\left(\frac{\calM\left(\calD'_{m}\right)\left(\bh\right)}{\calM\left(\calD_{m}\right)\left(\bh\right)}\right)^{\lambda}\right]\leq 1+\binom{\lambda}{2} \frac{\left(e^{\epsilon_0}-1\right)^2}{m e^{\epsilon_0}} + \sum_{i=3}^{\lambda} \binom{\lambda}{i} i\Gamma\left(i/2\right) \left(\frac{\left(e^{2\epsilon_0}-1\right)^2}{2me^{2\epsilon_0}}\right)^{i/2}
\end{equation*}
This completes the proof of Theorem~\ref{thm:RDP_same}.

%
 
\section{Proofs of Theorem~\ref{thm:general_case} and Theorem~\ref{thm:general_case_2} -- Upper Bounds}\label{sec:general_case}

We prove Theorems~\ref{thm:general_case} and \ref{thm:general_case_2} in Sections~\ref{subsec:proof_general_case} and \ref{subsec:proof_general_case_2}, respectively.
%
%
 
\subsection{Proof of Theorem~\ref{thm:general_case}}\label{subsec:proof_general_case}
Consider arbitrary neighboring datasets $\calD=\left(d_1,\ldots,d_n\right)\in\calX^{n}$ and $\calD'=\left(d_1,\ldots,d_{n-1},d_n'\right)\in\calX^{n}$. As mentioned in Section~\ref{subsec:proof-sketch_general-case}, for any $m\in\left\{0,\ldots,n-1\right\}$, we define new neighboring datasets $\calD_{m+1}^{(n)}=\left(d'_n,\ldots,d'_n,d_n\right)\in \calX^{m+1}$ and $\calD_{m+1}'^{(n)}=\left(d'_n,\ldots,d'_n,d'_{n}\right)\in \calX^{m+1}$, each having $(m+1)$ elements. Observe that $\left(\calD_{m+1}'^{(n)},\calD_{m+1}^{(n)}\right)\in\calD^{m+1}_{\same}$.

Recall from Theorem~\ref{Thm:reduce_special_case}, we have
\begin{align}
\mathbb{E}_{\bh\sim\calM(\calD')}\left[\left(\frac{\calM(\calD)(\bh)}{\calM(\calD')(\bh)}\right)^{\lambda}\right] &\leq \mathbb{E}_{m\sim\text{Bin}\left(n-1,q\right)}\left[\mathbb{E}_{\bh\sim\calM(\calD_{m+1}'^{(n)})}\left[\left(\frac{\calM(\calD_{m+1}^{(n)})(\bh)}{\calM(\calD_{m+1}'^{(n)})(\bh)}\right)^{\lambda}\right]\right] \notag \\
&= \sum_{m=0}^{n-1}\binom{n-1}{m} q^m(1-q)^{n-m-1}\left[\mathbb{E}_{\bh\sim\calM(\calD_{m+1}'^{(n)})}\left[\left(\frac{\calM(\calD_{m+1}^{(n)})(\bh)}{\calM(\calD_{m+1}'^{(n)})(\bh)}\right)^{\lambda}\right]\right] \label{eq:proof-reduce_special-case_bound}
\end{align}
For simplicity of notation, for any $m\in\{0,1,\hdots,n-1\}$, define 
\begin{align*}
q_m &:= \binom{n-1}{m} q^m(1-q)^{n-m-1} \\
E_m &:= \mathbb{E}_{\bh\sim\calM(\calD_{m+1}'^{(n)})}\left[\left(\frac{\calM(\calD_{m+1}^{(n)})(\bh)}{\calM(\calD_{m+1}'^{(n)})(\bh)}\right)^{\lambda}\right].
\end{align*} 
First we show an important property of $E_m$ that we will use in the proof.
\begin{lemma}\label{lem:E_m-decreasing}
$E_m$ is a non-increasing function of $m$, i.e., 
\begin{align}\label{eq:E_m-decreasing}
\mathbb{E}_{\bh\sim\calM(\calD_{m+1}'^{(n)})}\left[\left(\frac{\calM(\calD_{m+1}^{(n)})(\bh)}{\calM(\calD_{m+1}'^{(n)})(\bh)}\right)^{\lambda}\right] \leq \mathbb{E}_{\bh\sim\calM(\calD_{m}'^{(n)})}\left[\left(\frac{\calM(\calD_{m}^{(n)})(\bh)}{\calM(\calD_{m}^{(n)})(\bh)}\right)^{\lambda}\right],
\end{align}
where, for any $k\in\{m,m+1\}$, $\calD_{k}^{(n)}=\left(d'_n,\ldots,d'_n,d_{n}\right)$ and  $\calD_{k}'^{(n)}=\left(d'_n,\ldots,d'_n,d_n'\right)$ with $|\calD_k|=|\calD'_k|=k$.
\end{lemma}
\begin{proof}
Lemma~\ref{lem:E_m-decreasing} follows from Lemma~\ref{lem:cvx_rdp} in a straightforward manner, as, unlike Lemma~\ref{lem:E_m-decreasing}, in Lemma~\ref{lem:cvx_rdp} we consider arbitrary pairs of neighboring datasets.
\end{proof}
Continuing from \eqref{eq:proof-reduce_special-case_bound}, and using concentration properties of the Binomial r.v.\ together with Lemma~\ref{lem:E_m-decreasing}, we get
\begin{align}
\mathbb{E}_{\bh\sim\calM(\calD')}\left[\left(\frac{\calM(\calD)(\bh)}{\calM(\calD')(\bh)}\right)^{\lambda}\right]&\leq \sum_{m=0}^{n-1}q_mE_m \notag \\ 
&=\sum_{m<\floor{(1-\gamma)q(n-1)}} q_m E_m+ \sum_{m\geq\floor{(1-\gamma)q(n-1)}} q_m E_m\notag \\
&\stackrel{\text{(a)}}{\leq} E_{0}\sum_{m<\floor{(1-\gamma)q(n-1)}}q_m+\sum_{m\geq\floor{(1-\gamma)q(n-1)}} q_m E_m \notag \\
&\stackrel{\text{(b)}}{\leq}  E_{0}e^{-\frac{q(n-1)\gamma^2}{2}}+\sum_{m\geq\floor{(1-\gamma)q(n-1)}} q_m E_m \notag \\
&\stackrel{\text{(c)}}{\leq} e^{\epsilon_0\lambda}e^{-\frac{q(n-1)\gamma^2}{2}}+\sum_{m\geq\floor{(1-\gamma)q(n-1)}} q_m E_m \notag \\
&\stackrel{\text{(d)}}{\leq} e^{\epsilon_0\lambda}e^{-\frac{q(n-1)\gamma^2}{2}}+ E_{(1-\gamma)q(n-1)}. \label{proof_main-result_interim2}
\end{align}
Here, steps (a) and (d) follow from the fact that $E_m$ is a non-increasing function of $m$ (see Lemma~\ref{lem:E_m-decreasing}). Step (b) follows from the Chernoff bound. In step (c), we used that $\calM(d_n)=\calR(d_n)$ and $\calM(d'_n)=\calR(d'_n)$, which together imply that $E_0=\mathbb{E}\left[\left(\frac{\calM(d_n)}{\calM(d'_n)}\right)^{\lambda}\right] = \mathbb{E}\left[\left(\frac{\calR(d_n)}{\calR(d'_n)}\right)^{\lambda}\right] \leq e^{\eps_0\lambda}$, where the inequality follows because $\calR$ is an $\eps_0$-LDP mechanism. 

Note that we have already bounded $E_m$ for all $m$ in Theorem~\ref{thm:RDP_same}. For convenience, we state the bound below. 
\begin{equation*} 
E_{m-1} \leq \sup_{(\calD_m,\calD'_m)\in\calD_{\text{same}}^m}\mathbb{E}_{\bh\sim\calM(\calD_{m})}\left[\left(\frac{\calM(\calD'_{m})(\bh)}{\calM(\calD_{m})(\bh)}\right)^{\lambda}\right]\leq 1+\binom{\lambda}{2} \frac{\left(e^{\epsilon_0}-1\right)^2}{m e^{\epsilon_0}} + \sum_{i=3}^{\lambda} \binom{\lambda}{i} i\Gamma(i/2) \Bigg(\frac{\left(e^{2\epsilon_0}-1\right)^2}{2m e^{2\epsilon_0}}\Bigg)^{i/2}
\end{equation*}
By setting $\gamma=\frac{1}{2}$ and $\overline{n}=\floor{(1-\gamma)q(n-1)}+1=\floor{\frac{n-1}{2e^{\epsilon_0}}}+1$, we get
\begin{align}
\mathbb{E}_{\bh\sim\calM\left(\calD'\right)}\left[\left(\frac{\calM\left(\calD\right)\left(\bh\right)}{\calM\left(\calD'\right)\left(\bh\right)}\right)^{\lambda}\right] &\leq E_{\overline{n}-1} + e^{\epsilon_0\lambda-\frac{n-1}{8e^{\epsilon_0}}} \label{proof_main-result_interim3} \\
&\leq 1+\binom{\lambda}{2}\frac{\left(e^{\epsilon_0}-1\right)^2}{\overline{n}e^{\epsilon_0}}+\sum_{i=3}^{\lambda} \binom{\lambda}{i} i\Gamma\left(i/2\right) \left(\frac{\left(e^{2\epsilon_0}-1\right)^2}{2\overline{n}e^{2\epsilon_0}}\right)^{i/2}+e^{\epsilon_0\lambda-\frac{n-1}{8e^{\epsilon_0}}}. \notag 
\end{align}
Since the above bound holds for arbitrary pairs of neighboring datasets $\calD$ and $\calD'$, 
this completes the proof of Theorem~\ref{thm:general_case}.

\subsection{Proof of Theorem~\ref{thm:general_case_2}}\label{subsec:proof_general_case_2}
The proof of Theorem~\ref{thm:general_case_2} follows the same steps as that of the proof of Theorem~\ref{thm:general_case} that we outlined in Section~\ref{subsec:proof-sketch_general-case} and also gave formally in Section~\ref{subsec:proof_general_case}, except for the following change. Instead of using Theorem~\ref{thm:RDP_same} for bounding the RDP for specific neighboring datasets, we will use the following theorem.
\begin{theorem}\label{thm:RDP_same_2nd-bound} 
Let $m\in\bbN$ be arbitrary. For 
any $\lambda\geq 2$ (including the non-integral $\lambda$), we have 
\begin{equation}\label{RDP_same_bound2}
\sup_{(\calD_{m},\calD'_{m})\in\calD_{\emph{same}}^{m}}\mathbb{E}_{\bh\sim\calM\left(\calD_{m}\right)}\left[\left(\frac{\calM\left(\calD'_{m}\right)\left(\bh\right)}{\calM\left(\calD_{m}\right)\left(\bh\right)}\right)^{\lambda}\right]\leq \exp\left(\lambda^2 \frac{\left(e^{\epsilon_0}-1\right)^2}{m}\right).
\end{equation}
\end{theorem}
First we show how Theorem~\ref{thm:general_case_2} follows from Theorem~\ref{thm:RDP_same_2nd-bound}, and then we prove Theorem~\ref{thm:RDP_same_2nd-bound}. \\

Note that Theorem~\ref{thm:RDP_same_2nd-bound} implies that $E_{m-1} \leq \exp\left(\lambda^2 \frac{\left(e^{\epsilon_0}-1\right)^2}{m}\right)$ holds for every integer $m\geq2$.
Substituting this in \eqref{proof_main-result_interim3} (by putting $m=\overline{n}=\floor{\frac{n-1}{2e^{\epsilon_0}}}+1$), we get
\begin{equation*}
\mathbb{E}_{\bh\sim\calM\left(\calD'\right)}\left[\left(\frac{\calM\left(\calD\right)\left(\bh\right)}{\calM\left(\calD'\right)\left(\bh\right)}\right)^{\lambda}\right]\leq e^{\lambda^2\frac{\left(e^{\epsilon_0}-1\right)^2}{\overline{n}}}+e^{\epsilon_0\lambda-\frac{n-1}{8e^{\epsilon_0}}}.
\end{equation*}
This proves Theorem~\ref{thm:general_case_2}.

%

\begin{proof}[Proof of Theorem~\ref{thm:RDP_same_2nd-bound}]
Fix an arbitrary $m\in\bbN$. Let $(\calD_{m},\calD'_{m})\in\calD_{\same}^{m}$ and $\bp=(p_1,\ldots,p_B),\bp'=(p'_1,\ldots,p'_B)$ be the same as defined in the proof of Theorem~\ref{thm:RDP_same} in Section~\ref{sec:special_form}.

\begin{align}
\mathbb{E}_{\bh\sim\calM(\calD_{m})}\left[\left(\frac{\calM\left(\calD'_{m}\right)\left(\bh\right)}{\calM\left(\calD_{m}\right)\left(\bh\right)}\right)^{\lambda}\right] &= \mathbb{E}_{\bh\sim\calM(\calD_{m})}\left[\left(\sum_{j=1}^{B}\frac{p'_j}{p_j}\frac{h_j}{m}\right)^{\lambda}\right] \tag{from \eqref{ratio_mu}} \\
&=\mathbb{E}_{\bh\sim\calM(\calD_{m})}\left[\left(1+\sum_{j=1}^{B}\frac{p'_j}{p_j}\frac{h_j}{m}-1\right)^{\lambda}\right] \notag \\
&\leq \mathbb{E}_{\bh\sim\calM(\calD_{m})}\left[ \exp\left(\lambda\left(\sum_{j=1}^{B}\frac{p'_j}{p_j}\frac{h_j}{m}-1\right)\right)\right], \label{eqn:main-bound}
\end{align}
where the last inequality follows from $1+x \leq e^{x}$.

In \eqref{eqn:main-bound}, $\bh$ is distributed according to $\calM(\calD_{m})=\calH_m(\calR(d),\hdots,\calR(d))$, where $\calH_m$ denotes the shuffling operation on $m$ elements and range of $\calR$ is equal to $[B]$. Since all the $m$ data points are identical, and all clients use independent randomness for computing $\calR(d)$, we can assume, w.l.o.g., that $\calM(\calD_m)$ is a collection of $m$ i.i.d.\ random variables $X_1,\hdots,X_m$, where $\Pr\left[X_i=j\right]=p_j$ for $j\in\left[B\right]$. Thus, we have (in the following, note that $\bh=(h_1,\hdots,h_B)$ is a r.v.)
\begin{equation}~\label{eqn:ratio-i.i.d.}
\frac{1}{m}\sum_{j=1}^{B}\frac{p'_j}{p_j}h_j = \frac{1}{m}\sum_{j=1}^{B}\frac{p'_j}{p_j}\sum_{i=1}^{m}\mathbbm{1}_{\{X_i=j\}} = \frac{1}{m}\sum_{i=1}^{m}\sum_{j=1}^{B}\frac{p'_j}{p_j}\mathbbm{1}_{\{X_i=j\}} = \frac{1}{m}\sum_{i=1}^{m}\frac{p'_{X_i}}{p_{X_i}},
\end{equation} 
where $\mathbbm{1}_{\{\cdot\}}$ denotes the indicator r.v.
Substituting from~\eqref{eqn:ratio-i.i.d.} into~\eqref{eqn:main-bound}, we get
\begin{align}
\mathbb{E}_{\bh\sim\calM(\calD_{m})}\left[ \exp\left(\lambda\left(\sum_{j=1}^{B}\frac{p'_j}{p_j}\frac{h_j}{m}-1\right)\right)\right] &= \mathbb{E}_{X_1,\hdots,X_m}\left[ \exp\left(\frac{\lambda}{m}\sum_{i=1}^{m}\left(\frac{p'_{X_i}}{p_{X_i}}-1\right)\right)\right] \notag \\
&= \prod_{i=1}^m\mathbb{E}_{X_i}\left[ \exp\left(\frac{\lambda}{m}\left(\frac{p'_{X_i}}{p_{X_i}}-1\right)\right)\right] \notag \\
&= \(\mathbb{E}_{X\sim\bp}\left[ e^{\frac{\lambda}{m}\left(\frac{p'_{X}}{p_{X}}-1\right)}\right]\)^m \label{eqn:i.i.d-bound} 
\end{align}
where $\bp=\left[p_1,\ldots,p_B\right]$. From Taylor expansion of $e^x=1+\sum_{k=1}^{\infty}\frac{x^k}{k!}$, we get
\begin{align}
\mathbb{E}_{X\sim\bp}\left[ e^{\frac{\lambda}{m}\left(\frac{p'_{X}}{p_{X}}-1\right)}\right] &= 1+\sum_{k=1}^{\infty}\frac{1}{k!}\mathbb{E}_{X\sim\bp}\left[\(\frac{\lambda}{m}\left(\frac{p'_{X}}{p_{X}}-1\right)\)^k\right] \notag \\
&= 1+\sum_{k=1}^{\infty}\frac{1}{k!}\sum_{j=1}^B p_j \(\frac{\lambda}{m}\left(\frac{p'_j}{p_j}-1\right)\)^k \notag \\
&= 1+\sum_{k=2}^{\infty}\frac{1}{k!}\sum_{j=1}^B p_j \(\frac{\lambda}{m}\left(\frac{p'_j}{p_j}-1\right)\)^k \notag \\
&\leq 1+\sum_{k=2}^{\infty}\frac{1}{k!}\sum_{j=1}^B p_j \(\frac{\lambda(e^{\eps_0}-1)}{m}\)^k \tag{since $\frac{p_j'}{p_j}\leq e^{\eps_0}, \forall j\in[B]$} \\
&= 1+\sum_{k=1}^{\infty}\frac{1}{k!}\(\frac{\lambda(e^{\eps_0}-1)}{m}\)^k - \frac{\lambda(e^{\eps_0}-1)}{m} \notag \\
&= e^{\frac{\lambda \left(e^{\epsilon_0}-1\right)}{m}}-\frac{\lambda \left(e^{\epsilon_0}-1\right)}{m}. \label{eqn:exp-bound}
\end{align}
Substituting from~\eqref{eqn:exp-bound} into~\eqref{eqn:i.i.d-bound}, we get
\begin{align*}
\mathbb{E}_{\bh\sim\calM\left(\calD_{m}\right)}\left[\left(\frac{\calM\left(\calD'_{m}\right)\left(\bh\right)}{\calM\left(\calD_{m}\right)\left(\bh\right)}\right)^{\lambda}\right]&\leq \left(e^{\frac{\lambda \left(e^{\epsilon_0}-1\right)}{m}}-\frac{\lambda \left(e^{\epsilon_0}-1\right)}{m}\right)^{m}\\
&=e^{\lambda \left(e^{\epsilon_0}-1\right)}\left[1-\frac{\lambda \left(e^{\epsilon_0}-1\right)}{m}e^{\frac{-\lambda \left(e^{\epsilon_0}-1\right)}{m}}\right]^{m} \\
&\leq e^{\lambda \left(e^{\epsilon_0}-1\right)} e^{-\lambda \left(e^{\epsilon_0}-1\right)e^{\frac{-\lambda \left(e^{\epsilon_0}-1\right)}{m}}} \tag{since $1-x\leq e^{-x}$} \\
&=e^{\lambda \left(e^{\epsilon_0}-1\right)\left[1-e^{\frac{-\lambda \left(e^{\epsilon_0}-1\right)}{m}}\right]}\\
&\leq e^{\frac{\lambda^2 \left(e^{\epsilon_0}-1\right)^2}{m}} \tag{since $1-e^{-x} \leq x$}.
\end{align*}
This completes the proof of Theorem~\ref{thm:RDP_same_2nd-bound}.
\end{proof}

\section{Proof of Theorem~\ref{thm:lower_bound} -- Lower Bound}\label{sec:lower-bound}


Consider the binary case, where each data point $d$ can take a value from $\calX=\lbrace 0,1\rbrace$. Let the local randomizer $\calR$ be the binary randomized response (2RR) mechanism, where $\Pr\left[\calR\left(d\right)=d\right]=\frac{e^{\epsilon_0}}{e^{\epsilon_0}+1}$ for $d\in\calX$. It is easy to verify that $\calR$ is an $\eps_0$-LDP mechanism. For simplicity, let $p=\frac{1}{e^{\epsilon_0}+1}$. Consider two neighboring datasets $\calD,\ \calD' \in\{0,1\}^n$, where $\calD=\left(0,\ldots,0,0\right)$ and $\calD'=\left(0,\ldots,0,1\right)$. Let $k\in\left\{0,\ldots,n\right\}$ denote the number of ones in the output of the shuffler. As argued in Section~\ref{subsec:prob-form} on page~\pageref{histogram-set}, since the output of the shuffled mechanism $\calM$ can be thought of as the distribution of the number of ones in the output, we have that $k\sim\calM(\calD)$ is distributed as a Binomial random variable Bin$(n,p)$. Thus, we have
\begin{align*}
\calM(\calD)(k)&= \binom{n}{k} p^{k} (1-p)^{n-k} \\ 
\calM(\calD')(k)&= (1-p) \binom{n-1}{k-1} p^{k-1} (1-p)^{n-k}+ p\binom{n-1}{k} p^{k} (1-p)^{n-k-1}.
\end{align*}
It will be useful to compute $\frac{\calM(\calD)(k)}{\calM(\calD')(k)}-1$ for the calculations later.
\begin{align}
\frac{\calM(\calD')(k)}{\calM(\calD)(k)} -1 &= \frac{(1-p) \binom{n-1}{k-1} p^{k-1} (1-p)^{n-k} + p\binom{n-1}{k} p^{k} (1-p)^{n-k-1}}{\binom{n}{k} p^{k} (1-p)^{n-k}} -1 \notag \\
&= \frac{k}{n}\frac{(1-p)}{p} + \frac{(n-k)}{n}\frac{p}{(1-p)} -1 \notag \\
&= \frac{k}{n} e^{\eps_0} + \frac{(n-k)}{n} e^{-\eps_0} -1 \notag \\
&= \frac{k}{n}\(e^{\eps_0}-e^{-\eps_0}\) + e^{-\eps_0} -1 \notag \\
&= \frac{k}{n}\(\frac{e^{2\eps_0}-1}{e^{\eps_0}}\) - \(\frac{e^{\eps_0} -1}{e^{\eps_0}}\) \notag \\
&= \(\frac{e^{2\eps_0}-1}{ne^{\eps_0}}\)\(k - \frac{n}{e^{\eps_0}+1}\) \label{eq:compute_ratio_lb}
\end{align}
Thus, we have that 
\begin{align*}
\mathbb{E}_{k\sim\calM(\calD)}\left[\left(\frac{\calM(\calD')(k)}{\calM(\calD)(k)}\right)^{\lambda}\right]&=\mathbb{E}\left[\left(1+\frac{\calM(\calD')(k)}{\calM(\calD)(k)}-1\right)^{\lambda}\right]\\
&\stackrel{\text{(a)}}{=} 1+\sum_{i=1}^{\lambda} \binom{\lambda}{i} \mathbb{E}\left[\left(\frac{\calM(\calD')(k)}{\calM(\calD)(k)}-1\right)^{i}\right] \\
&\stackrel{\text{(b)}}{=} 1+\sum_{i=2}^{\lambda} \binom{\lambda}{i} \mathbb{E}\left[\left(\frac{\calM(\calD')(k)}{\calM(\calD)(k)}-1\right)^{i}\right] \\
&=1+\sum_{i=2}^{\lambda} \binom{\lambda}{i} \left(\frac{\left(e^{2\epsilon_0}-1\right)}{ne^{\epsilon_0}}\right)^{i} \mathbb{E}\left[\left(k-\frac{n}{e^{\epsilon_0}+1}\right)^{i}\right] \tag{from \eqref{eq:compute_ratio_lb}} \\
&\stackrel{\text{(c)}}{=} 1+\binom{\lambda}{2}\frac{\left(e^{\epsilon_0}-1\right)^2}{ne^{\epsilon_0}}+\sum_{i=3}^{\lambda} \binom{\lambda}{i} \left(\frac{\left(e^{2\epsilon_0}-1\right)}{ne^{\epsilon_0}}\right)^{i} \mathbb{E}\left[\left(k-\frac{n}{e^{\epsilon_0}+1}\right)^{i}\right].
\end{align*}
Here, step (a) from the polynomial expansion $(1+x)^n=\sum_{k=0}^{n}\binom{n}{k}x^k$, 
step (b) follows because the term corresponding to $i=1$ is zero (i.e., $\mathbb{E}_{k\sim\calM(\calD)}\left[\left(\frac{\calM(\calD')(k)}{\calM(\calD)(k)}-1\right)\right]=0$),
and step (c) from the from the fact that $\mathbb{E}_{k\sim\calM(\calD)}\left[\left(k-\frac{n}{e^{\epsilon_0}+1}\right)^{2}\right]=np(1-p)=\frac{ne^{\eps_0}}{(e^{\eps_0}+1)^2}$, which is equal to the variance of the Binomial random variable. This completes the proof of Theorem~\ref{thm:lower_bound}.
\section{Conclusion}\label{sec:conclusion}

The analysis of the RDP for the shuffle model presented in this paper
was based on some new analysis techniques that may be of independent
interest. The utility of these bounds were also demonstrated
numerically, where we saw that in important regimes of interest, we
get $8\times$ improvement over the state-of-the-art without sampling
and at least $10\times$ improvement with sampling (see Section
\ref{sec:numerics} for more details).

A simple extension of the results would be to work with local
approximate DP guarantees instead of pure LDP. This can be seen by
using the tight conversion between approximate DP and pure DP given in
\cite{feldman2020hiding}. However, there are several open problems of
interest. Perhaps the most important one is mentioned in Remark~\ref{remark:gap-up-lb}.
There is a multiplicative gap of the order
$e^{\eps_0}$ in our upper and lower bounds, and closing this gap is an
important open problem. We believe that our lower bound is tight
(at least for the first order term) and the upper bound is
loose. Showing this or getting a tighter upper bound may require new
proof techniques. A second question could be how to get an overall RDP
guarantee if we are given local RDP guarantees instead of local LDP
guarantees.


\bibliographystyle{alpha}
\bibliography{RDPRefs}

\appendix

\section{Proof of Corollary~\ref{corol:simplified_general_case}}\label{app:simplied_1st-bound-proof}
In this section, we prove the simplified bound (stated in \eqref{eq:simplified_1st-bound}) on the RDP of the shuffle model, provided that $\lambda,\eps_0,n$ satisfy a certain condition. In particular, we will show that 
if $\lambda,\eps_0,n$ satisfy $\lambda^4e^{5\eps_0}<\frac{n}{9}$, then 
\begin{equation}\label{eq:simplified_1st-bound_proof}
\epsilon(\lambda) \leq \frac{1}{\lambda-1}\log\left(1+\frac{\lambda^2\left(e^{\epsilon_0}-1\right)^2}{\overline{n}e^{\epsilon_0}} \right),
\end{equation}
where $\overline{n}=\frac{n-1}{2e^{\epsilon_0}}+1$.
In order to show \eqref{eq:simplified_1st-bound_proof}, it suffices to prove the following (using which in \eqref{eqn:1st_bound} will yield \eqref{eq:simplified_1st-bound}):
\begin{align}\label{proof_simplied_1st-bound-1}
\sum_{i=3}^{\lambda} \binom{\lambda}{i} i\Gamma\left(i/2\right) \left(\frac{\left(e^{2\epsilon_0}-1\right)^2}{2e^{2\epsilon_0}\overline{n}}\right)^{i/2}+e^{\epsilon_0\lambda-\frac{n-1}{8e^{\epsilon_0}}} \leq \binom{\lambda}{2} \frac{\left(e^{\epsilon_0}-1\right)^2}{\overline{n}e^{\epsilon_0}}.
\end{align}
First notice that $\binom{\lambda}{i} i\Gamma\left(i/2\right) \leq \lambda^i$ (see Claim~\ref{claim:gamma_ineq} on page~\pageref{claim:gamma_ineq}).
In order to show \eqref{proof_simplied_1st-bound-1}, it suffices to show 
\begin{align}\label{proof_simplied_1st-bound-2}
\sum_{i=3}^{\lambda} \left(\frac{\lambda\left(e^{2\epsilon_0}-1\right)}{(2e^{2\epsilon_0}\overline{n})^{1/2}}\right)^i + e^{\epsilon_0\lambda-\frac{n-1}{8e^{\epsilon_0}}} \leq \binom{\lambda}{2} \frac{\left(e^{\epsilon_0}-1\right)^2}{\overline{n}e^{\epsilon_0}}.
\end{align}
Note that there are $(\lambda-2)$ terms inside the summation. If we show that each of those terms is smaller than 1 (which would imply that the term corresponding to $i=3$ is the largest one), then the summation is at most $(\lambda-2)$ times the term with $i=3$. Further, if the additional exponential term in the LHS is upper-bounded by the term with $i=3$, then we can prove \eqref{proof_simplied_1st-bound-2} by showing that $(\lambda-1)$ times the term with $i=3$ is upper-bounded by the RHS. These arguments are summarized in the following set of three inequalities:
\begin{align}
\frac{\lambda\left(e^{2\epsilon_0}-1\right)}{(2e^{2\epsilon_0}\overline{n})^{1/2}} &< 1 \label{proof_simplied_1st-bound-3} \\
e^{\epsilon_0\lambda-\frac{n-1}{8e^{\epsilon_0}}} &\leq \left(\frac{\lambda\left(e^{2\epsilon_0}-1\right)}{(2e^{2\epsilon_0}\overline{n})^{1/2}}\right)^3 \label{proof_simplied_1st-bound-4} \\
(\lambda-1)\left(\frac{\lambda\left(e^{2\epsilon_0}-1\right)}{(2e^{2\epsilon_0}\overline{n})^{1/2}}\right)^3 &\leq \binom{\lambda}{2} \frac{\left(e^{\epsilon_0}-1\right)^2}{\overline{n}e^{\epsilon_0}} \label{proof_simplied_1st-bound-5}
\end{align}
In the rest of this proof, we will derive the condition on $\eps_0,\lambda,n$ such that \eqref{proof_simplied_1st-bound-5} is satisfied. As we see later, the values of $\eps_0,\lambda$ thus obtained will automatically satisfy \eqref{proof_simplied_1st-bound-3} and \eqref{proof_simplied_1st-bound-4}.

By cancelling same terms from both sides of \eqref{proof_simplied_1st-bound-5}, we get

\begin{align}
\frac{\lambda^2\left(e^{2\epsilon_0}-1\right)^3}{(2e^{\epsilon_0}\overline{n})^{3/2}e^{3\eps_0/2}} &\leq  \frac{\left(e^{\epsilon_0}-1\right)^2}{2\overline{n}e^{\epsilon_0}} \notag \\
\iff \lambda^2(e^{2\epsilon_0}-1)(e^{\epsilon_0}+1)^2 &\leq \sqrt{2\overline{n}e^{\epsilon_0}}e^{3\eps_0/2} \label{proof_simplied_1st-bound-6}
\end{align}
For the LHS and the RHS, we respectively have
\begin{align}
(e^{2\epsilon_0}-1)(e^{\epsilon_0}+1)^2 &= (e^{2\epsilon_0}-1)(e^{2\epsilon_0}+2e^{\eps_0}+1) \leq e^{4\eps_0}+2e^{3\eps_0} \leq 3e^{4\eps_0} \label{proof_simplied_1st-bound-7} \\
2\overline{n}e^{\epsilon_0} &= n-1+2e^{\eps_0}\geq n. \label{proof_simplied_1st-bound-8} 
\end{align} 
Therefore, in order to show \eqref{proof_simplied_1st-bound-6}, it suffices to show
$3\lambda^2 e^{4\eps_0} \leq \sqrt{ne^{3\eps_0}}$, 
which is equivalent to $\lambda^4e^{5\eps_0}<\frac{n}{9}$. Thus, we have shown that $\lambda^4e^{5\eps_0}<\frac{n}{9}$ implies \eqref{proof_simplied_1st-bound-5}. \\

Now we show that when $\lambda^4e^{5\eps_0}<\frac{n}{9}$, \eqref{proof_simplied_1st-bound-3} and \eqref{proof_simplied_1st-bound-4} are automatically satisfied:
\begin{enumerate}
\item {Proof of \eqref{proof_simplied_1st-bound-3}:} 
\begin{align*}
\frac{\lambda\left(e^{2\epsilon_0}-1\right)}{\sqrt{2e^{2\epsilon_0}\overline{n}}} &\leq \frac{\lambda e^{2\epsilon_0}}{\sqrt{2e^{\epsilon_0}\overline{n}}} \leq \sqrt{\frac{\lambda^4 e^{5\epsilon_0}}{2e^{\epsilon_0}\overline{n}}} \leq \sqrt{\frac{n/9}{n}} < 1.
\end{align*}
In the second inequality we used $\lambda\geq1$ and in the penultimate inequality we used $2e^{\epsilon_0}\overline{n} \geq n$ from \eqref{proof_simplied_1st-bound-8}.
\item {Proof of \eqref{proof_simplied_1st-bound-4}:} For this, first we upper-bound the LHS and lower-bound the RHS, and then note that the upper-bound is smaller than the lower-bound. For the upper-bound on $\exp(\eps_0\lambda - \frac{n-1}{8e^{\eps_0}})$, note that $\eps_0\lambda \leq e^{5\eps_0/4}\lambda = \(e^{5\eps_0}\lambda^4\)^{1/4} < \(\frac{n}{9}\)^{1/4}= \frac{n^{1/4}}{\sqrt{3}}$. Also note that $e^{\eps_0}\leq e^{5\eps_0/4}\lambda < \frac{n^{1/4}}{\sqrt{3}}$, which implies $\frac{n-1}{8e^{\eps_0}}=\frac{\sqrt{3}}{8}\frac{n-1}{n^{1/4}}\geq \frac{\sqrt{3}}{16}n^{3/4}$. Substituting these bounds in the exponent of $\exp(\eps_0\lambda - \frac{n-1}{8e^{\eps_0}})$, we get:
\begin{align}\label{proof_simplied_1st-bound-10}
\exp\(\eps_0\lambda - \frac{n-1}{8e^{\eps_0}}\) &\leq \exp\(\frac{n^{1/4}}{\sqrt{3}} - \frac{\sqrt{3}}{16}n^{3/4}\) 
= \exp\(-n^{3/4}\(\frac{\sqrt{3}}{16} - \frac{1}{\sqrt{3n}}\)\) \leq \exp\(-c'n^{3/4}\),
\end{align}
where $c'>0$ is a constant even for small values of $n$. For example, for $n=100$, we get $c'\geq\frac{1}{20}$. 

For the lower-bound on $\left(\frac{\lambda\left(e^{2\epsilon_0}-1\right)}{(2e^{2\epsilon_0}\overline{n})^{1/2}}\right)^3$, note that $2e^{\eps_0}\overline{n}=n-1+2e^{\eps_0}\leq n-1+2\(\frac{n}{9}\)^{1/5}\leq 2n$, where $e^{\eps_0}\leq \(\frac{n}{9}\)^{1/5}$ follows from $e^{5\eps_0}\leq \lambda^4e^{5\eps_0}<\frac{n}{9}$. Now we show the lower bound:
\begin{align}
\frac{\lambda^3(e^{2\epsilon_0}-1)^3}{(2e^{2\epsilon_0}\overline{n})^{3/2}} &\geq \frac{(e^{\epsilon_0}-1)^3(e^{\epsilon_0}+1)^3}{(2e^{\eps_0}\overline{n})^{3/2}e^{3\eps_0/2}} \geq \frac{(e^{\epsilon_0}-1)^3e^{3\epsilon_0}}{(2n)^{3/2}e^{3\eps_0/2}} \geq \frac{(e^{\epsilon_0}-1)^3}{(2n)^{3/2}} \geq \frac{\eps_0^3}{(2n)^{3/2}} \label{proof_simplied_1st-bound-11}
\end{align}
Note that the upper-bound on $\exp(\eps_0\lambda - \frac{n-1}{8e^{\eps_0}})$ is exponentially small in $n^{3/4}$, whereas, the lower-bound on $\frac{\lambda^3(e^{2\epsilon_0}-1)^3}{(2e^{2\epsilon_0}\overline{n})^{3/2}}$ is inverse-polynomial in $n$. So, for sufficiently large $n$, \eqref{proof_simplied_1st-bound-4} will be satisfied.
\end{enumerate}
This completes the proof of Corollary~\ref{corol:simplified_general_case}

\begin{claim}[An Inequality for the Gamma Function]\label{claim:gamma_ineq}
For any $\lambda\in\bbN$ and $k\geq3$, we have $\binom{\lambda}{k}k\Gamma(k/2)\leq  \lambda^k$.
\end{claim}
\begin{proof}
Note that for any $\lambda\in\bbN$ and $k\leq\lambda$, we have
$\binom{\lambda}{k} = \frac{\lambda(\lambda-1)(\lambda-2)\hdots(\lambda-k+1)}{k!}$.

We show the claim separately for the cases when $k$ is an even integer or not.
\begin{enumerate}
\item {\it When $k$ is an even integer:}
Since for any integer $n\in\bbN$, $\Gamma(n)=(n-1)!$, so when $k$ is an even integer, we have 
\begin{align*}
\binom{\lambda}{k}k\Gamma(k/2) &= \frac{\lambda(\lambda-1)(\lambda-2)\hdots(\lambda-k+1)}{k!}\times k\times (\frac{k}{2}-1)! \\
&\leq \lambda(\lambda-1)(\lambda-2)\hdots(\lambda-k+1) \\
&\leq \lambda^k.
\end{align*}
\item {\it When $k$ is an odd integer:} Note that for any integer $n\in\bbN$, we have 
$\Gamma\(n+\frac{1}{2}\) = \frac{(2n)!}{4^nn!}\sqrt{\pi}$; see \cite{gamma-wiki}. Let $k=2a+1$. Then
\begin{align*}
\binom{\lambda}{k}k\Gamma(k/2) &= \binom{\lambda}{k}k\Gamma(a+\frac{1}{2}) \\
&= \frac{\lambda(\lambda-1)(\lambda-2)\hdots(\lambda-k+1)}{k!}\times k\times \frac{(2a)!}{4^aa!}\sqrt{\pi} \\
&= \lambda(\lambda-1)(\lambda-2)\hdots(\lambda-k+1)\frac{\sqrt{\pi}}{4^aa!} \\
&\stackrel{\text{(a)}}{\leq} \lambda(\lambda-1)(\lambda-2)\hdots(\lambda-k+1) \\
&\leq \lambda^k
\end{align*}
where (a) follows because $\frac{\sqrt{\pi}}{4^aa!}\leq 1$ when $a\geq1 \iff k\geq3$.
\end{enumerate}
This proves Claim~\ref{claim:gamma_ineq}.
\end{proof}

\section{Omitted Details from Section~\ref{sec:proof_reduce_special_case}}\label{app:reduce_special_case}
\subsection{Proof of Lemma~\ref{lem:convex-combinations}}\label{app:proof_convex-combinations}
\begin{lemma*}[Restating Lemma~\ref{lem:convex-combinations}]
$F(\calP)$ and $F(\calP')$ can be written as the following convex combinations:
\begin{align}
F(\calP)&=\sum_{\calC\subseteq [n-1]} q^{|\calC|}(1-q)^{n-|\calC|-1}F(\calP_{\calC}), \label{app:P_mixture} \\
F(\calP')&=\sum_{\calC\subseteq [n-1]} q^{|\calC|}(1-q)^{n-|\calC|-1}F(\calP'_{\calC}), \label{app:P-prime_mixture}
\end{align}
where $\calP_{\calC},\calP_{\calC}'$ are defined in \eqref{eq:defn_P_C}-\eqref{eq:defn_hatP}.
\end{lemma*}
\begin{proof}
We only show \eqref{app:P_mixture}; \eqref{app:P-prime_mixture} can be shown similarly.

For convenience, for any $\calC\subseteq[n-1]$, define 
\begin{align*}
\calP'_{|\calC|,n} &= \{\bp'_n,\hdots,\bp'_n\} \text{ with } |\calP'_{|\calC|,n}|=|\calC|, \\
\widetilde{\calP}_{[n-1]\setminus\calC} &= \{\tilde{\bp}_i:i\in\left[n-1\right]\setminus \calC\}.
\end{align*}
With these notations, we can write $\calP_{\calC} = \calP'_{|\calC|,n} \bigcup \widetilde{\calP}_{[n-1]\setminus\calC} \bigcup \{\bp_n\}$ and $\calP'_{\calC} = \calP'_{|\calC|,n} \bigcup \widetilde{\calP}_{[n-1]\setminus\calC} \bigcup \{\bp'_n\}$.

Note that $\bp_i=q \bp'_n+(1-q) \tilde{\bp}_i$ for all $i\in[n-1]$. 
For any $i\in[n-1]$, define the following random variable $\bhp_i$:
\[
\bhp_i = 
\begin{cases}
\bp'_n & \text{ w.p. } q, \\
\btp_i & \text{ w.p. } 1-q. \\
\end{cases}
\]
Note that $\bbE[\bhp_i]=\bp_i$.

For any subset $\calC\subseteq[n-1]$, define an event $\calE_{\calC}:=\{\bhp_i=\bp'_n \text{ for } i\in\calC \text{ and } \bhp_i=\btp_i \text{ for } i\in[n-1]\setminus\calC\}$. Since $\bhp_1,\hdots,\bhp_{n-1}$ are independent random variables, we have $\Pr[\calE_{\calC}]=q^{|\calC|}(1-q)^{n-|\calC|-1}$.

Consider an arbitrary $\bh\in\calA_B^n$. 
Define a random variable $U(\calP)$ over $\calA_B^n$ whose distribution is equal to $F(\calP)$.
\begin{align}
F(\calP)(\bh) &= \Pr[U(\calP)=\bh] \notag \\
&= \Pr[U(\bp_1,\hdots,\bp_{n-1},\bp_n)=\bh] \notag \\
&= \Pr\big[U\(\bbE[\bhp_1],\hdots,\bbE[\bhp_{n-1}],\bp_n\)=\bh\big] \notag \\
&= \sum_{\calC\subseteq[n-1]}\Pr[\calE_{\calC}]\Pr\big[U\(\bbE[\bhp_1],\hdots,\bbE[\bhp_{n-1}],\bp_n\)=\bh \ | \ \calE_{\calC}\big] \notag \\
&\stackrel{\text{(e)}}{=} \sum_{\calC\subseteq[n-1]}\Pr[\calE_{\calC}]\Pr\Big[U\Big(\calP'_{|\calC|,n} \bigcup \widetilde{\calP}_{[n-1]\setminus\calC} \bigcup \{\bp_n\}\Big)=\bh\Big] \notag \\
&= \sum_{\calC\subseteq[n-1]}\Pr[\calE_{\calC}]\Pr\big[U(\calP_{\calC})=\bh\big] \notag \\
&= \sum_{\calC\subseteq[n-1]}q^{|\calC|}(1-q)^{n-|\calC|-1}\Pr\big[U(\calP_{\calC})=\bh\big], \notag \\
&= \sum_{\calC\subseteq[n-1]}q^{|\calC|}(1-q)^{n-|\calC|-1}F(\calP_{\calC})(\bh)
\end{align}
where, $\calP'_{|\calC|,n}$ and $\widetilde{\calP}_{[n-1]\setminus\calC}$ in the RHS of (e) are defined in the statement of the claim.

Since the above calculation holds for every $\bh\in\calA_B^n$, we have proved \eqref{app:P_mixture}.
\end{proof}

\subsection{Proof of Lemma~\ref{lem:joint-con_renyi-exp}}\label{app:proof_joint-con_renyi-exp}
\begin{lemma*}[Restating Lemma~\ref{lem:joint-con_renyi-exp}]
For any $\lambda>1$, the function $\mathbb{E}_{\bh\sim F\left(\calP'\right)}\left[\left(\frac{F\left(\calP\right)\left(\bh\right)}{F\left(\calP'\right)\left(\bh\right)}\right)^{\lambda}\right]$ is jointly convex in $(F(\calP),F(\calP'))$, i.e.,
\begin{align}
\mathbb{E}_{\bh\sim F\left(\calP'\right)}\left[\left(\frac{F\left(\calP\right)\left(\bh\right)}{F\left(\calP'\right)\left(\bh\right)}\right)^{\lambda}\right]
&\leq \sum_{\calC\subseteq \left[n-1\right]} q^{|\calC|}\left(1-q\right)^{n-|\calC|-1} \mathbb{E}_{\bh\sim F\left(\calP'_{\calC}\right)}\left[\left(\frac{F\left(\calP_{\calC}\right)\left(\bh\right)}{F\left(\calP'_{\calC}\right)\left(\bh\right)}\right)^{\lambda}\right]. \label{app:rdp_bound}
\end{align}
\end{lemma*}
%
\begin{proof}
For simplicity of notation, let $P=F(\calP)$ and $Q=F(\calP')$.
Note that $\bbE_Q\left[\left(\frac{P}{Q}\right)^{\lambda}\right]=\int P^{\lambda}Q^{1-\lambda}d\mu$, which is also called the Hellinger integral. 
In order to prove the lemma, it suffices to show that $\int P^{\lambda}Q^{1-\lambda}d\mu$ is jointly convex in $(P,Q)$, i.e., if $P_{\alpha}=\alpha P_0 + (1-\alpha)P_1$ and $Q_{\alpha}=\alpha Q_0 + (1-\alpha)Q_1$ for some $\alpha\in[0,1]$, then the following holds
\begin{align}
\int P_{\alpha}^{\lambda}Q_{\alpha}^{1-\lambda}d\mu \leq \alpha\int P_0^{\lambda}Q_0^{1-\lambda}d\mu + (1-\alpha)\int P_1^{\lambda}Q_1^{1-\lambda}d\mu. \label{hellinger_joint-convexity}
\end{align}
Proof of \eqref{hellinger_joint-convexity} is implicit in the proof of \cite[Theorem 13]{HarremoesRenyiKL14}. However, for completeness, we prove \eqref{hellinger_joint-convexity} in Lemma~\ref{lem:hellinger_convex} below.

Since $P=F(\calP)$ and $Q=F(\calP')$ are convex combinations of $P_{\calC}=F(\calP_{\calC})$ and $Q_{\calC}=F(\calP'_{\calC})$, respectively, with same coefficients, repeated application of \eqref{hellinger_joint-convexity} implies \eqref{app:rdp_bound}.
\end{proof}

\begin{lemma}\label{lem:hellinger_convex}
For $\lambda\geq 1$, the Hellinger integral
$\int P^{\lambda}Q^{1-\lambda}d\mu$ is jointly convex in $(P,Q)$, i.e., if $P_{\alpha}=\alpha P_0 + (1-\alpha)P_1$ and $Q_{\alpha}=\alpha Q_0 + (1-\alpha)Q_1$ for some $\alpha\in[0,1]$, then we have 
\begin{align}
\int P_{\alpha}^{\lambda}Q_{\alpha}^{1-\lambda}d\mu \leq \alpha\int P_0^{\lambda}Q_0^{1-\lambda}d\mu + (1-\alpha)\int P_1^{\lambda}Q_1^{1-\lambda}d\mu. \label{app_hellinger_joint-convexity}
\end{align}
\end{lemma}
\begin{proof}
Let $f(x)=x^{\lambda}$. It is easy to show that for any $\lambda\geq1$, $f(x)$ is a convex function when $x>0$. 
This implies that for any point $\omega\in\Omega$ in the sample space, we have
\begin{align*}
f\left(\frac{P_{\alpha}(\omega)}{Q_{\alpha}(\omega)}\right) &= f\left(\frac{\alpha P_{0}(\omega)}{Q_{\alpha}(\omega)} + \frac{(1-\alpha)P_{1}(\omega)}{Q_{\alpha}(\omega)} \right) \\
&= f\left(\frac{\alpha Q_{0}(\omega)}{Q_{\alpha}(\omega)} \frac{P_{0}(\omega)}{Q_{0}(\omega)} + \frac{(1-\alpha)Q_{1}(\omega)}{Q_{\alpha}(\omega)} \frac{P_{1}(\omega)}{Q_{1}(\omega)} \right) \\
&\leq \frac{\alpha Q_{0}(\omega)}{Q_{\alpha}(\omega)} f\left(\frac{P_{0}(\omega)}{Q_{0}(\omega)}\right) + \frac{(1-\alpha)Q_{1}(\omega)}{Q_{\alpha}(\omega)} f\left(\frac{P_{1}(\omega)}{Q_{1}(\omega)} \right), 
\end{align*}
where the last inequality follows from the convexity of $f(x)$. By multiplying both sides with $Q_{\alpha}(\omega)$ and substituting the definition of $f(x)=x^{\lambda}$, we get
\begin{equation*}
P_{\alpha}^{\lambda}(\omega)Q_{\alpha}^{1-\lambda}(\omega) \leq \alpha P_0^{\lambda}(\omega)Q_0^{1-\lambda}(\omega) + (1-\alpha) P_1^{\lambda}(\omega)Q_1^{1-\lambda}(\omega).
\end{equation*}
By integrating this equality, we get~\eqref{app_hellinger_joint-convexity}.
\end{proof}

\subsection{Proof of Lemma~\ref{lem:cvx_rdp}}\label{app:proof_cvx_rdp}
\begin{lemma*}[Restating Lemma~\ref{lem:cvx_rdp}]
For any $i\in\left[n-1\right]$, we have
\begin{equation*} 
\mathbb{E}_{\bh\sim F\left(\calP'\right)}\left[\left(\frac{F\left(\calP\right)\left(\bh\right)}{F\left(\calP'\right)\left(\bh\right)}\right)^{\lambda}\right]\leq \mathbb{E}_{\bh\sim F\left(\calP'_{-i}\right)}\left[\left(\frac{F\left(\calP_{-i}\right)\left(\bh\right)}{F\left(\calP'_{-i}\right)\left(\bh\right)}\right)^{\lambda}\right],
\end{equation*}
where, for $i\in[n-1]$, $\calP_{-i}=\calP\setminus\{\bp_i\}$ and $\calP'_{-i}=\calP'\setminus\{\bp_i\}$. 
Note that in the LHS, $F(\calP),F(\calP')$ are distributions over $\calA_B^n$, whereas, in the RHS, $F(\calP_{-i}),F(\calP'_{-i})$ for any $i\in[n-1]$ are distributions over $\calA_B^{n-1}$.
\end{lemma*}
\begin{proof}
First we show that $\bbE_{\bh\sim F(\calP')}\left[\left(\frac{F(\calP)(\bh)}{F(\calP')(\bh)}\right)^{\lambda}\right]$ is convex in $\bp_i$ for any $i\in[n-1]$. 

Note that due to the independence of $\calR$ on different data points, for any $\bh=(h_1,\hdots,h_B)\in\calA^n_B$, we can recursively write the distributions $F(\calP)(\bh)$ and $F(\calP')(\bh)$ (which are defined in \eqref{general-distribution}) as follows:
\begin{align}
F(\calP)(\bh) &= \sum_{j=1}^{B}p_{ij}F(\calP_{-i})(\tbh_j), \qquad \forall i\in\left[n\right] \label{decomposition_P} \\
F(\calP')(\bh)&=\sum_{j=1}^{B}p_{ij}F(\calP'_{-i})(\tbh_j)=\sum_{j=1}^{B}p'_{nj}F(\calP'_{-n})(\tbh_j), \qquad \forall i\in\left[n-1\right], \label{decomposition_P-prime}
\end{align}
where $\tbh_j=(h_1,\ldots,h_{j-1},h_j-1,h_{j+1},\ldots,h_B)$ for any $j\in[B]$. 
Here, $F(\calP_{-i})$, $F(\calP'_{-i})$ are distributions over $\calA_{B}^{n-1}.$\footnote{We assume that $F(\calP_{-i})(\tbh_j)=0$ and $F(\calP'_{-i})(\tbh_j)=0$ if $h_j-1<0$.}

Fix any $i\in\left[n-1\right]$ and also fix arbitrary $\bp_1,\ldots,\bp_{i-1},\bp_{i+1},\ldots,\bp_{n},\bp'_{n}$. Take any $\alpha\in[0,1]$, and consider $\bp^{\alpha}_i=\alpha \bp^{0}_{i} + (1-\alpha)\bp^{1}_i$.
Let $\calP_{\alpha}=(\bp_1,\ldots,\bp_{i}^{\alpha},\ldots,\bp_n)$, $\calP_{0}=(\bp_1,\ldots,\bp_{i}^{0},\ldots,\bp_n)$, and $\calP_{1}=(\bp_1,\ldots,\bp_{i}^{1},\ldots,\bp_n)$. Similarly, let $\calP'_{\alpha}=(\bp_1,\ldots,\bp_{i}^{\alpha},\ldots,\bp'_n)$, $\calP'_{0}=(\bp_1,\ldots,\bp_{i}^{0},\ldots,\bp'_n)$, and $\calP'_{1}=(\bp_1,\ldots,\bp_{i}^{1},\ldots,\bp'_n)$. With these definitions, we have $\calP_{\alpha}=\alpha\calP_0+(1-\alpha)\calP_1$. Note that $\(\calP_{\alpha}\)_{-i}=\(\calP_0\)_{-i}=\(\calP_1\)_{-i}$.

 Then, from the recursive definitions of $F\left(\calP\right)$ and $F\left(\calP'\right)$ (given in \eqref{decomposition_P} and \eqref{decomposition_P-prime}, respectively), for any $\bh\in\calA_B^n$, we get
\begin{align*}
F(\calP_{\alpha})(\bh) &= \sum_{j=1}^{B}p_{ij}^{\alpha}F\((\calP_{\alpha})_{-i}\)(\tbh_j) \\
&= \alpha\sum_{j=1}^{B}p_{ij}^0F\((\calP_{\alpha})_{-i}\)(\tbh_j) + (1-\alpha)\sum_{j=1}^{B}p_{ij}^1F\((\calP_{\alpha})_{-i}\)(\tbh_j) \tag{since $\bp^{\alpha}_i=\alpha \bp^{0}_{i} + (1-\alpha)\bp^{1}_i$} \\
&= \alpha\sum_{j=1}^{B}p_{ij}^0F\((\calP_{0})_{-i}\)(\tbh_j) + (1-\alpha)\sum_{j=1}^{B}p_{ij}^1F\((\calP_{1})_{-i}\)(\tbh_j) \tag{since $\(\calP_{\alpha}\)_{-i}=\(\calP_0\)_{-i}=\(\calP_1\)_{-i}$} \\
&= \alpha F(\calP_{0})(\bh) + (1-\alpha)F(\calP_{1})(\bh).
\end{align*}
Similarly, we can show that $F(\calP'_{\alpha})(\bh)=\alpha F(\calP'_{0})(\bh) + (1-\alpha)F(\calP'_{1})(\bh)$.

Thus we have shown that
\begin{align*}
F\left(\calP_{\alpha}\right)&=\alpha F\left(\calP_0\right)+\left(1-\alpha\right)F\left(\calP_1\right) \\ 
F\left(\calP'_{\alpha}\right)&=\alpha F\left(\calP'_0\right)+\left(1-\alpha\right)F\left(\calP'_1\right). 
\end{align*}

From Lemma~\ref{lem:hellinger_convex}, we have that $\mathbb{E}_{\bh\sim F\left(\calP'\right)}\left[\left(\frac{F\left(\calP\right)\left(\bh\right)}{F\left(\calP'\right)\left(\bh\right)}\right)^{\lambda}\right]$ is jointly convex in $F\left(\calP\right)$ and $F\left(\calP'\right)$. As a result, we get
\begin{equation}
\mathbb{E}_{\bh\sim F\left(\calP'_{\alpha}\right)}\left[\left(\frac{F\left(\calP_{\alpha}\right)\left(\bh\right)}{F\left(\calP'_{\alpha}\right)\left(\bh\right)}\right)^{\lambda}\right]\leq \alpha \mathbb{E}_{\bh\sim F\left(\calP'_0\right)}\left[\left(\frac{F\left(\calP_0\right)\left(\bh\right)}{F\left(\calP'_0\right)\left(\bh\right)}\right)^{\lambda}\right]+\left(1-\alpha\right) \mathbb{E}_{\bh\sim F\left(\calP'_1\right)}\left[\left(\frac{F\left(\calP_1\right)\left(\bh\right)}{F\left(\calP'_1\right)\left(\bh\right)}\right)^{\lambda}\right]
\end{equation} 
Thus, we have shown that $\bbE_{\bh\sim F(\calP')}\left[\left(\frac{F(\calP)(\bh)}{F(\calP')(\bh)}\right)^{\lambda}\right]$ is convex in $\bp_i$ for any $i\in[n-1]$.


Now we are ready to prove Lemma~\ref{lem:cvx_rdp}.

The LDP constraints put some restrictions on the set of values that the distribution $\bp_i$ can take; however, 
the maximum value that $\bbE_{\bh\sim F(\calP')}\left[\left(\frac{F(\calP)(\bh)}{F(\calP')(\bh)}\right)^{\lambda}\right]$ takes can only increase when we remove those constraints. We instead maximize it w.r.t.\ $\bp_i$ over the simplex $\Delta_B:=\{(p_{i1},\hdots,p_{iB}):p_{ij}\geq0\text{ for }j\in[B] \text{ and } \sum_{j=1}^Bp_{ij}=1\}$. This implies
\begin{align}\label{cvx_rdp-interim1}
\bbE_{\bh\sim F(\calP')}\left[\left(\frac{F(\calP)(\bh)}{F(\calP')(\bh)}\right)^{\lambda}\right] &\leq \max_{\bp_i\in\Delta_B} \bbE_{\bh\sim F(\calP')}\left[\left(\frac{F(\calP)(\bh)}{F(\calP')(\bh)}\right)^{\lambda}\right]
\end{align}

Substituting from \eqref{decomposition_P} and \eqref{decomposition_P-prime} into \eqref{cvx_rdp-interim1}, we get
\begin{align}\label{cvx_rdp-interim2}
\bbE_{\bh\sim F(\calP')}\left[\left(\frac{F(\calP)(\bh)}{F(\calP')(\bh)}\right)^{\lambda}\right] &\leq 
\max_{\bp_i\in\Delta_B} \bbE_{\bh\sim F(\calP')}\left[\left(\frac{\sum_{j=1}^{B}p_{ij}F(\calP_{-i})(\tbh_j)}{\sum_{j=1}^{B}p_{ij}F(\calP'_{-i})(\tbh_j)}\right)^{\lambda}\right]
\end{align}
Since maximizing a convex function over a polyhedron attains its maximum value at one of its vertices, and there are $B$ vertices in the simplex $\Delta_B$, which are of the form $p_{ij^*}=1$ for some $j^*\in[B]$ and $p_{ik}=0$ for all $k\neq j^*$, we have
\begin{align*}
\max_{\bp_i\in\Delta_B} \bbE_{\bh\sim F(\calP')}\left[\left(\frac{\sum_{j=1}^{B}p_{ij}F(\calP_{-i})(\tbh_j)}{\sum_{j=1}^{B}p_{ij}F(\calP'_{-i})(\tbh_j)}\right)^{\lambda}\right] &\stackrel{\text{(a)}}{=} \bbE_{\bh\sim F(\calP')}\left[\left(\frac{F(\calP_{-i})(\tbh_{j^*})}{F(\calP'_{-i})(\tbh_{j^*})}\right)^{\lambda}\right] \\
&\stackrel{\text{(b)}}{=} \bbE_{\bh\sim F(\calP'_{-i})}\left[\left(\frac{F(\calP_{-i})(\bh)}{F(\calP'_{-i})(\bh)}\right)^{\lambda}\right]
\end{align*}
Since the $i$'th data point deterministically maps to the $j^*$'th output by the mechanism $\calR$, the expectation term in the RHS of (a) has no dependence on the $i$'th data point, so we can safely remove that, which gives (b).
This proves Lemma~\ref{lem:cvx_rdp}.
\end{proof}

\subsection{Proof of Corollary~\ref{corol:cvx_rdp_repeat}}\label{app:proof_cvx_rdp_repeat}
\begin{corollary*}[Restating Corollary~\ref{corol:cvx_rdp_repeat}]
Consider any $m\in\{0,1,\hdots,n-1\}$.
Let $\calD_{m+1}^{(n)}=\left(d'_n,\ldots,d'_n,d_n\right)$ and $\calD_{m+1}'^{(n)}=\left(d'_n,\ldots,d'_n\right)$.
Then for any $\calC\in\binom{[n-1]}{m}$, we have
\begin{align*}
\bbE_{\bh\sim F(\calP'_{\calC})}\left[\left(\frac{F(\calP_{\calC})(\bh)}{F(\calP'_{\calC})(\bh)}\right)^{\lambda}\right] &\leq \bbE_{\bh\sim \calM(\calD_{m+1}'^{(n)})}\left[\left(\frac{\calM(\calD_{m+1}^{(n)})(\bh)}{\calM(\calD_{m+1}'^{(n)})(\bh)}\right)^{\lambda}\right]. 
\end{align*}
\end{corollary*}
\begin{proof}
Recall from Lemma~\ref{lem:convex-combinations} and the notation defined in Appendix~\ref{app:reduce_special_case}, that for any $\calC\subseteq[n-1]$, we have
$\calP_{\calC} = \calP'_{|\calC|,n} \bigcup \widetilde{\calP}_{[n-1]\setminus\calC} \bigcup \{\bp_n\}$ and $\calP'_{\calC} = \calP'_{|\calC|,n} \bigcup \widetilde{\calP}_{[n-1]\setminus\calC} \bigcup \{\bp'_n\}$, where $\calP'_{|\calC|,n}=\{\bp'_n,\hdots,\bp'_n\}$ with $|\calP'_{|\calC|,n}|=|\calC|$ and $\widetilde{\calP}_{[n-1]\setminus\calC}=\{\tilde{\bp}_i:i\in\left[n-1\right]\setminus \calC\}$.

Now, repeatedly applying Lemma~\ref{lem:cvx_rdp} over the set of distributions $\tilde{\bp}_i\in\widetilde{\calP}_{[n-1]\setminus\calC}$, we get that 
\begin{align*}
\bbE_{\bh\sim F(\calP'_{\calC})}\left[\left(\frac{F(\calP_{\calC})(\bh)}{F(\calP'_{\calC})(\bh)}\right)^{\lambda}\right] &\leq \bbE_{\bh\sim F\(\calP'_{|\calC|,n} \bigcup \{\bp'_n\}\)}\left[\left(\frac{F\(\calP'_{|\calC|,n} \bigcup \{\bp_n\}\)(\bh)}{F\(\calP'_{|\calC|,n} \bigcup \{\bp'_n\}\)(\bh)}\right)^{\lambda}\right] \\
&=\bbE_{\bh\sim \calM(\calD_{m+1}'^{(n)})}\left[\left(\frac{\calM(\calD_{m+1}^{(n)})(\bh)}{\calM(\calD_{m+1}'^{(n)})(\bh)}\right)^{\lambda}\right]
\end{align*}
In the last equality, we used that $\calP'_{|\calC|,n} \bigcup \{\bp_n\}$ has $|\calC|+1=m+1$ distributions which are associated with the $(m+1)$ data points $\{d'_n,\hdots,d'_n,d_n\}$ ($m$ of them are equal to $d'_n$); similarly, $\calP'_{|\calC|,n} \bigcup \{\bp'_n\}$ also has $|\calC|+1=m+1$ distributions which are associated with the $(m+1)$ data points $\{d'_n,\hdots,d'_n,d'_n\}$ (all of them are equal to $d'_n$). This implies that for every $\bh\in\calA_B^{m+1}$, $F\(\calP'_{|\calC|,n} \bigcup \{\bp_n\}\)(\bh)$ and $F\(\calP'_{|\calC|,n} \bigcup \{\bp'_n\}\)(\bh)$ are distributionally equal to $\calM(\calD_{m+1}^{(n)})(\bh)$ and $\calM(\calD_{m+1}'^{(n)})(\bh)$, respectively.

This proves Corollary~\ref{corol:cvx_rdp_repeat}.
\end{proof}

\section{Omitted Details from Section~\ref{sec:special_form}}

\subsection{Proof of Lemma~\ref{lemm:MomAcc_RV}}\label{app:supp_MomAcc_rv}
\begin{lemma*}[Restating Lemma~\ref{lemm:MomAcc_RV}]
The random variable $X$ has the following properties:
\begin{enumerate}
\item $X$ has zero mean, i.e., $\mathbb{E}_{\bh\sim\calM(\calD_m)}\left[X(\bh)\right]=0$.
\item The variance of $X$ is equal to 
\begin{equation*}
\mathbb{E}_{\bh\sim\calM(\calD_m)}\left[X(\bh)^{2}\right]=m\left(\sum_{j=1}^{B}\frac{p_j'^2}{p_j} -1\right).
\end{equation*}
\item For $i\geq3$, the $i$th moment of $X$ is bounded by 
\begin{equation*}
\mathbb{E}_{\bh\sim\calM(\calD_m)}\left[(X(\bh))^{i}\right] \leq \mathbb{E}_{\bh\sim\calM(\calD_m)}\left[|X(\bh)|^{i}\right]\leq i \Gamma\left(i/2\right)\left(2m\nu^2\right)^{i/2},
\end{equation*}
where $\nu^2=\frac{\left(e^{\epsilon_0}-e^{-\epsilon_0}\right)^2}{4}$ and $\Gamma\left(z\right)=\int_{0}^{\infty}x^{z-1}e^{-x}dx$ is the Gamma function. 
\end{enumerate}
\end{lemma*}

\begin{proof}
For simplicity of notation, let $\mu_0,\mu_1$ denote the distributions $\calM(\calD_m),\calM(\calD'_m)$, respectively. As shown in \eqref{ratio_mu}, for any $\bh\in\calA_B^m$, we have
$$X(\bh)=m\(\frac{\mu_{1}(\bh)}{\mu_{0}(\bh)} -1\) = \left(\sum_{j=1}^{B}a_jh_j\right)-m,$$
where $a_j=\frac{p'_j}{p_j}\in\left[e^{-\eps_0},e^{\eps_0}\right]$ for all $j\in\left[B\right]$.

Now we show the three properties.
\begin{enumerate}
\item The mean of the random variable $X$ is given by
\begin{align*}
\bbE_{\bh\sim\mu_0}\left[X(\bh)\right] &= m\bbE_{\bh\sim\mu_0}\left[\frac{\mu_{1}(\bh)}{\mu_{0}(\bh)} -1\right] = m\sum_{\bh\in\calA_B^m}\mu_0(\bh)\(\frac{\mu_{1}(\bh)}{\mu_{0}(\bh)} -1\) 
= m\sum_{\bh\in\calA_B^m}\(\mu_1(\bh) - \mu_{0}(\bh)\) 
= 0
\end{align*}
%

\item The variance of the random variable $X$ is given by
\begin{align*}
\mathbb{E}_{\bh\sim\mu_0}&\left[X\left(\bh\right)^{2}\right]=\bbE_{\bh\sim\mu_{0}}\left[\left(\sum_{j=1}^{B}a_jh_j - m\right)^{2}\right]\\ 
&=m^2 \bbE_{\bh\sim\mu_{0}}\left[\sum_{j=1}^{B}\sum_{l=1}^{B}a_ja_l\frac{h_jh_l}{m^2}-2\sum_{j=1}^{B}a_j\frac{h_j}{m}+1\right] \\
&=m^2\bbE_{\bh\sim\mu_{0}}\left[\sum_{j=1}^{B}a_j^2\frac{h_j^2}{m^2}+\sum_{j=1}^{B}\sum_{l\neq j}a_ja_l\frac{h_jh_l}{m^2} -2\sum_{j=1}^{B}a_j\frac{h_j}{m}+1\right] \notag \\
&=m^2\left[\sum_{j=1}^{B}\frac{(p'_j)^2}{p_j^2}\frac{\bbE_{\bh\sim\mu_{0}}[h_j^2]}{m^2}+\sum_{j=1}^{B}\sum_{l\neq j}\frac{p_j'p_l'}{p_jp_l}\frac{\bbE_{\bh\sim\mu_{0}}[h_jh_l]}{m^2} -2\sum_{j=1}^{B}\frac{p_j'}{p_j}\frac{\bbE_{\bh\sim\mu_{0}}[h_j]}{m} + 1 \right]\\
&\stackrel{\text{(b)}}{=}m^2\left[\sum_{j=1}^{B}\frac{(p'_j)^2}{p_j^2}\frac{(mp_j(1-p_j)+m^2p_j^2)}{m^2}+\sum_{j=1}^{B}\sum_{l\neq j}\frac{p_j'p_l'}{p_jp_l}\frac{(-mp_jp_l+m^2p_jp_l)}{m^2} -2\sum_{j=1}^{B}\frac{p_j'}{p_j}\frac{p_jm}{m} + 1\right]\\
&=m^2\left[\sum_{j=1}^{B}\(\frac{(p'_j)^2\left(1-p_j\right)}{p_jm}+(p'_j)^2\)+\sum_{j=1}^{B}\sum_{l\neq j}\( -\frac{p'_jp'_l}{m}+p'_jp'_l\)-1 \right]\\
&=m^2\left[ \frac{1}{m}\(\sum_{j=1}^{B}\frac{(p'_j)^2(1-p_j)}{p_j}-\sum_{j=1}^{B}\sum_{l\neq j}p'_jp'_l\)+\sum_{j=1}^B(p'_j)^2+\sum_{j=1}^{B}\sum_{l\neq j}p'_jp'_l-1 \right] \\
&=m^2\left[ \frac{1}{m}\(\sum_{j=1}^{B}\frac{(p'_j)^2}{p_j} - \sum_{j=1}^{B}(p'_j)^2 - \sum_{j=1}^{B}\sum_{l\neq j}p'_jp'_l\)+\sum_{j=1}^B(p'_j)^2+\sum_{j=1}^{B}\sum_{l\neq j}p'_jp'_l-1\right] \\
&\stackrel{\text{(c)}}{=}m\left(\sum_{j=1}^{B}\frac{(p'_j)^2}{p_j}-1\right).
\end{align*}
Here, step $\text{(b)}$ uses properties of multinomial distribution: $\bbE_{\bh\sim\mu_{0}}[h_j]=mp_j$, $\bbE_{\bh\sim\mu_{0}}[h_j^{2}]=mp_j(1-p_j)+m^2p_j^2$, and $\bbE_{\bh\sim\mu_{0}}[h_jh_l] = -mp_jp_l + m^2p_jp_l$ for $j\neq l$. Step $\text{(c)}$ follows because $\sum_{j=1}^B(p'_j)^2+\sum_{j=1}^{B}\sum_{l\neq j}p'_jp'_l = \left(\sum_{j=1}^Bp'_j\right)^2=1$, as $\bp'=(p'_1,\hdots,p'_B)$ is a probability distribution.


\item Let $Y_i$ denote the random variable associated with the output of the local randomizer at the $i$'th client. So, $\Pr\left[Y_i=j\right]=p_j$ for $j\in\left[B\right]$. Recall that $h_j$ denote the number of clients that map to the $j$'th element from $[B]$. This implies that for any $j\in[B]$, we have $h_j=\sum_{i=1}^{m}\mathbbm{1}_{\{Y_i=j\}}$. For any $i\in[m]$, define a random variable $X_i=\(\sum_{j=1}^{B}a_j\mathbbm{1}_{\{ Y_i=j\}}\)-1$, where $a_j=\frac{p'_j}{p_j}$. Observe that $X_1,\ldots,X_m$ are zero mean i.i.d.\ random variables, because for any $i\in[m]$, we have $\mathbb{E}\left[X_i\right]=\(\sum_{j=1}^{B}a_jp_j\)-1=0$. With these definitions, we can equivalently represent $X(\bh)=\(\sum_{j=1}^Ba_jh_j\)-m$ as $X(\bh)=\sum_{i=1}^{m}X_i$, which is the sum of $m$ zero mean i.i.d.\ r.v.s. 
Furthermore, since $a_j\in\left[e^{-\epsilon_0},e^{\epsilon_0}\right]$ for any $j\in[B]$, we have $X_i\in\left[e^{-\epsilon_0}-1,e^{\epsilon_0}-1\right]$. Since any bounded r.v.\ $Z\in[a,b]$ is a sub-Gaussian r.v.\ with parameter $\frac{(b-a)^2}{4}$ (see~\cite[Lemma~$1.8$]{rigollet2015high})), we have that $X_i$ is a sub-Gaussian r.v.\ with parameter $\nu^{2}=\frac{\left(e^{\epsilon_0}-e^{-\epsilon_0}\right)^2}{4}$, i.e.,
\begin{equation*}
\mathbb{E}\left[e^{sX_i}\right]\leq e^{\frac{s^2\nu^2}{2}}, \qquad \forall s\in\bbR.
\end{equation*} 
It follows that $X\left(\mathbf{h}\right)=\sum_{i=1}^{m}X_i$ is also a sub-Gaussian random variable with parameter $m\nu^2$. The remaining steps are similar to bound the moments of a sub-Gaussian random variable. We write them here for completeness. From Chernoff bound we get 
\begin{align*}
\Pr\left[X\geq t\right]&\leq \min_{s\geq0} \frac{\mathbb{E}\left[e^{sX}\right]}{e^{st}}\\
&\leq\min_{s\geq0}  \frac{e^{\frac{s^2m\nu^2}{2}}}{e^{st}}\\
&\stackrel{\text{(b)}}{\leq}  e^{-\frac{t^2}{2m\nu^2}}
\end{align*}
where (b) follows by setting $s=\frac{t}{m\nu^2}$. Similarly, we can bound the term $\Pr\left[-X\geq t\right]$. Thus, we get 
\begin{equation*}
\Pr\left[|X|\geq t\right]\leq 2 e^{-\frac{t^2}{2m\nu^2}} 
\end{equation*}
Hence, the $i$'th moment of the random variable $X$ can be bounded by
\begin{align*}
\mathbb{E}\left[|X|^{i}\right]&= i \int_{0}^{\infty} t^{i-1}\Pr\left[|X|\geq t\right] dt\\
&\leq 2i \int_{0}^{\infty} t^{i-1} e^{-\frac{t^2}{2m\nu^2}} dt\\
&\stackrel{\text{(b)}}{=} i\left(2m\nu^2\right)^{i/2} \int_{0}^{\infty} u^{i/2-1} e^{-u} du\\
&= i\left(2m\nu^2\right)^{i/2} \Gamma\left(i/2\right),
\end{align*}
where step (b) follows by setting $u=\frac{t^2}{2m\nu^2}$ (change of variables). 
In the last step, $\Gamma\left(z\right)=\int_{0}^{\infty}x^{z-1}e^{-x}dx$ denotes the Gamma function. 
Thus, we conclude that for every $i\geq3$, we have $\mathbb{E}\left[|X|^{i}\right]\leq i \Gamma\left(i/2\right)\left(2m\nu^2\right)^{i/2}$, where $\nu^2=\frac{\left(e^{\epsilon_0}-e^{-\epsilon_0}\right)^2}{4}$. 
\end{enumerate}
This completes the proof of Lemma~\ref{lemm:MomAcc_RV}.
\end{proof}

\subsection{Proof of Lemma~\ref{lemma_bound_sup}}\label{app:supp_MomAcc_bound_sup}
\begin{lemma*}[Restating Lemma~\ref{lemma_bound_sup}]
We have the following bound:
\begin{equation*}
\sup_{(\bp,\bp')\in\calT_{\eps_0}}\(\sum_{j=1}^{B}\frac{p_j'^2}{p_j}-1\) = \frac{\left(e^{\eps_0}-1\right)^2}{e^{\eps_0}}.
\end{equation*}
\end{lemma*}
\begin{proof}
For any $(\bp,\bp')\in\calT_{\eps_0}$, define $f(\bp,\bp')=\sum_{j=1}^{B}\frac{(p'_j)^2}{p_j}$. 
Since the function $g\left(x,y\right)=\frac{x^2}{y}$ is convex in $(x,y)$ for $y>0$, it implies that the objective function $f(\bp,\bp')$ is also convex in $(\bp,\bp')$. It is easy to verify that $\calT_{\eps_0}$ is a polytope.

Since we maximize a convex function $f(\bp,\bp')$ over a polytope $\calT_{\eps_0}$, the optimal solution is one of the vertices of the polytope. Note that any vertex $(\bp,\bp')$ of the polytope in $B$ dimensions satisfies all the $B$ LDP constraints (i.e., $e^{-\eps_0}\leq \frac{p_j}{p'_j}\leq e^{\eps_0}, j=1,\hdots,B$) with equality. Without loss of generality, assume that the optimal solution $(\tilde{\bp},\tilde{\bp}')$ is a vertex such that $\frac{\tilde{p}'_j}{\tilde{p}_j} = e^{\eps_0}$ for $j=1,\hdots,l$ and $\frac{\tilde{p}'_j}{\tilde{p}_j} = e^{-\eps_0}$ for $j=l+1,\hdots,B$, for some $l\in [B]$. Thus, we have
\begin{align*}
1 = \sum_{j=1}^{B}\tilde{p}'_j &= e^{\eps_0}\sum_{j=1}^{l}\tilde{p}_j + e^{-\eps_0}\sum_{j=l+1}^{B}\tilde{p}_j 
= e^{\eps_0}\sum_{j=1}^{l}\tilde{p}_j + e^{-\eps_0}\Big(1-\sum_{j=1}^{l}\tilde{p}_j\Big) 
= e^{-\eps_0} + (e^{\eps_0} - e^{-\eps_0})\sum_{j=1}^{l}\tilde{p}_j 
\end{align*}
Rearranging the above gives $\sum_{j=1}^{l}\tilde{p}_j=\frac{1}{e^{\eps_0}+1}$. This implies $\sum_{j=1}^{l}\tilde{p}'_j=\frac{e^{\eps_0}}{e^{\eps_0}+1}$, which in turn implies $\sum_{j=l+1}^{B}\tilde{p}'_j=\frac{1}{e^{\eps_0}+1}$.
Now the result follows from the following set of equalities:
\begin{align*}
f\left(\tilde{\bp},\tilde{\bp}'\right) &= \sum_{j=1}^{B}\frac{(\tilde{p}'_j)^2}{\tilde{p}_j} 
= \sum_{j=1}^{l}\frac{\tilde{p}'_j}{\tilde{p}_j} \tilde{p}'_j + \sum_{j=l+1}^{B}\frac{\tilde{p}'_j}{\tilde{p}_j} \tilde{p}'_j 
=e^{\eps_0}\sum_{j=1}^{l}\tilde{p}'_j+e^{-\eps_0}\sum_{j=l+1}^{B}\tilde{p}'_j\\
&= \frac{e^{2\eps_0}}{e^{\eps_0}+1}+\frac{1}{e^{\eps_0}\left(e^{\eps_0}+1\right)} 
= \frac{(e^{\eps_0})^3+1}{e^{\eps_0}(e^{\eps_0}+1)}  
= \frac{\left(e^{\eps_0}-1\right)^2}{e^{\eps_0}}+1,
\end{align*}
where the last equality uses the identity $x^3+1 = (x+1)(x^2-x+1)$. This completes the proof of Lemma~\ref{lemma_bound_sup}. 
\end{proof}

}

\end{document}